\documentclass[11pt,usletter]{article}
\usepackage{graphicx}
\usepackage{subcaption}
\usepackage[margin = 2.5cm]{geometry}
\usepackage{color}
\usepackage{amsmath}
\usepackage{amssymb}
\usepackage{amsthm}
\usepackage{bbm}
\usepackage{tikz}
\usepackage[hidelinks]{hyperref}
\usepackage{comment}
\usepackage{algorithm,algpseudocode}
\usepackage{tcolorbox}
\usepackage{enumerate}
\graphicspath{{./images/}}

\usepackage[normalem]{ulem}

\newcommand{\bm}[1]{{\boldsymbol#1}}
\newcommand{\bR}{\mathbb{R}}
\newcommand{\bC}{\mathbb{C}}
\newcommand{\bN}{\mathbb{N}}

\DeclareMathOperator{\supp}{supp}
\let\Re\relax

\DeclareMathOperator{\Re}{Re}

\newtheorem{lemma}{Lemma}
\newtheorem{remark}{Remark}

\newtheorem{theorem}{Theorem}

\newcommand\numberthis{\addtocounter{equation}{1}\tag{\theequation}}

\title{%LASSO-type greedy algorithms for weighted sparse recovery \\
The greedy side of the LASSO: New algorithms for weighted sparse recovery  via loss function-based orthogonal matching pursuit}
\author{Sina Mohammad-Taheri\footnote{Department of Mathematics and Statistics, Concordia University, Montr\'eal, QC, Canada.} $^{,}$\footnote{Corresponding author. E-mail: \href{mailto:sina.mohammad-taheri@mail.concordia.ca}{sina.mohammad-taheri@mail.concordia.ca}} { } and Simone Brugiapaglia$^*$}
\begin{document}

\maketitle

\begin{abstract}
We propose a class of greedy algorithms for weighted sparse recovery by considering new loss function-based generalizations of Orthogonal Matching Pursuit (OMP). Given a (regularized) loss function, the proposed algorithms alternate the iterative construction of the signal support via greedy index selection and a signal update based on solving a local data-fitting problem restricted to the current support. 
We show that greedy selection rules associated with popular weighted sparsity-promoting loss functions admit explicitly computable and simple formulas. Specifically, we consider $ \ell^0 $- and $ \ell^1 $-based versions of the weighted LASSO (Least Absolute Shrinkage and Selection Operator), the Square-Root LASSO (SR-LASSO) and the Least Absolute Deviations LASSO (LAD-LASSO). Through numerical experiments on Gaussian compressive sensing and high-dimensional function approximation, we demonstrate the effectiveness of the proposed algorithms by empirically showing that they can outperform standard OMP (with respect to accuracy and computational cost) and inherit desirable characteristics from the corresponding loss functions, such as SR-LASSO's noise-blind optimal parameter tuning and LAD-LASSO's fault tolerance. In doing so, our study sheds new light on the connection between greedy sparse recovery and convex relaxation.

\end{abstract}

\paragraph{Keywords:} weighted sparsity, greedy algorithms, orthogonal matching pursuit, LASSO, square-root LASSO, least absolute deviations LASSO. 

\section{Introduction}

Sparse recovery  lies at the heart of modern data science, signal processing, and statistical learning. Its goal is to reconstruct an $ N $-dimensional $ s $-sparse signal $ x $ (i.e., such that $ \|x\|_0:=|\{j:x_j \neq 0\}| \leq s $) from $ m $ (possibly noisy) linear measurements $ y = Ax + e$, where $A$ is an $ m \times N $ measurement (sensing,  mixing, or dictionary) matrix and $e$ is an $m$-dimensional noise vector. In this paper, we focus in particular on the \emph{compressed sensing} framework \cite{candes2006robust,donoho2006compressed}, corresponding to the underdetermined regime (i.e., $m < N$). For a general treatment of sparse recovery, compressed sensing and their numerous applications in data science, signal processing, and scientific computing we refer to, e.g., the books \cite{adcock2022sparse, adcock2021compressive, elad2010sparse, eldar2012compressed, foucart2013mathematical, hastie2015statistical, lai2021sparse, vidyasagar2019introduction}.

Sparse recovery techniques are typically divided into two main categories: convex relaxation methods and iterative algorithms. In convex relaxation methods, sparse solutions are identified by solving convex optimization programs such as those based on $ \ell^1 $ minimization. Popular examples are \emph{(Quadratically-Constrained) Basis Pursuit} and the \emph{Least Absolute Shrinkage and Selection Operator} (\emph{LASSO}).
On the other hand, iterative algorithms aim at computing a sparse solution through explicit iterative algorithmic procedures that combine techniques from numerical linear algebra with sparsity-enhancing ideas. These include thresholding-based algorithms such as \emph{Iterative Hard Thresholding} (\emph{IHT}) and \emph{Hard Thresholding Pursuit} (\emph{HTP}), and greedy algorithms such as \emph{Compressive Sampling Matching Pursuit} (\emph{CoSaMP}) and \emph{Orthogonal Matching Pursuit} (\emph{OMP})---the main object of study of this paper. For a detailed overview of these and other sparse recovery techniques we refer readers to, e.g.,    \cite{adcock2021compressive,foucart2013mathematical,lai2021sparse,temlyakov2011greedy,vidyasagar2019introduction}.

Over the last few years, motivated by the need to incorporate prior knowledge about the target signal into sparse reconstruction methods, a substantial amount of research has been devoted to \emph{weighted} sparse recovery. In a variety of applications, ranging from compressive imaging to surrogate modelling and uncertainty quantification, it has been shown both empirically and theoretically that a careful choice of weights can improve both reconstruction accuracy and sample complexity with respect to unweighted $ \ell^1 $ minimization. A non-exhaustive list of works in this direction includes  \cite{adcock2018infinite,adcock2019correcting,adcock2017compressed,adcock2022sparse,adcock2021compressive,bah2016sample,candes2008enhancing,chkifa2018polynomial,friedlander2011recovering,huang2008adaptive,khajehnejad2011analyzing,peng2014weighted,rauhut2016interpolation,yu2013sufficient,zou2006adaptive} and references therein. 

Although weighted sparse recovery has been extensively investigated from the perspective of convex relaxation through weighted $\ell^1$ minimization, iterative algorithms are far from being well studied in the weighted setting. To the best of our knowledge, iterative algorithms for weighted sparse recovery have only been considered in a handful of works  \cite{adcock2020sparse,jo2013iterative,li2013weighted, wu2013weighted}. The main goal of our paper is to reduce this gap. With this aim, adopting an approach that merges convex relaxation and iterative algorithms, we propose new  LASSO-based weighted greedy algorithms of OMP type. 

\subsection{Main contributions}

The main contributions of this paper can be summarized as follows.
\begin{enumerate}
\item Adopting a \emph{loss function-based} perspective (see \S\ref{s:summary_main} and \S\ref{sec:greedy_selection}), we propose a new class of greedy algorithms able to promote weighted sparse recovery based on the OMP paradigm.
They are defined via theoretically-justified greedy index selection rules based on maximal reduction of weighted LASSO-type loss functions (see Theorems~\ref{thm:general}, \ref{thm:WLASSO}, \ref{thm:WSRLASSO} and \ref{thm:WLADLASSO}). These include the weighed (unconstrained) LASSO and two of its most notable variants: the weighted  \emph{Square-Root LASSO} (\emph{SR-LASSO}) and \emph{Least Absolute Deviations LASSO} (\emph{LAD-LASSO}). This loss function-based perspective allows one to adapt OMP to various structured signal models and sources of errors corrupting the data.

\item The proposed algorithms are numerically shown to outperform standard OMP (with respect to both accuracy and computational cost) and inherit the desirable characteristics of the underlying loss functions. In particular, those based on the SR- and LAD-LASSO, have noise-blind tuning parameter selection strategies and fault-tolerance, respectively. In addition, thanks to the presence of a regularization term, our greedy algorithms prevent overfitting and, consequently, improve the robustness of OMP with respect to the number of iterations. Numerical evidence in this direction is presented in \S\ref{sec:numerics}. These results shed new light on the connection between convex relaxation methods and iterative (specifically, greedy) algorithms.

\item The proposed algorithms admit a reliable stopping criterion and a significant reduction in runtime, thanks to the regularization effect mentioned above.

\end{enumerate}
We conclude with a remark about the novelty of our contributions in relation to an OMP variant proposed in \cite{adcock2020sparse}. A comprehensive literature review can be found in \S\ref{literature_review}
\begin{remark}
Our construction in this work builds upon a variant of OMP proposed in \cite{adcock2020sparse} that relies on a weighted $ \ell^0 $-based LASSO formulation. Here, adopting a more general loss function-based perspective, we extend the work of \cite{adcock2020sparse} to a broader class of loss functions including $ \ell^1 $-based LASSO and other variants of the LASSO family, i.e., weighted SR- and LAD-LASSO. Moreover, we extend the weighted OMP strategy proposed in \cite{adcock2020sparse} to the case of $\ell^0$-based SR- and LAD-LASSO. See Appendix~\ref{app:l0_greedy_selection}.
\end{remark}

\subsection{Summary of the main results}
\label{s:summary_main}
We now provide an overview of our main results, referring to \S\ref{sec:greedy_selection} for a detailed technical discussion.
Our objective is to construct a signal that minimizes a loss function of the form 
\begin{equation}
\label{eq:def_G_intro}
G(z) := F(z) + \lambda R(z), \quad \forall z \in \mathbb{F}^N,
\end{equation}
where $\mathbb{F} = \mathbb{R}$ or $\mathbb{C}$, and $ F, R: \mathbb{F}^N \to [0, +\infty) $ are a data-fidelity and a regularization term, respectively, and $ \lambda \geq 0 $ is a tuning parameter. For weighted LASSO, SR-LASSO, and LAD-LASSO loss functions (see equations \eqref{eq:loss_WLASSO}, \eqref{eq:loss_WSRLASSO}, and \eqref{eq:loss_WLADLASSO}, respectively) $ F $ is an $\ell^2$- or $\ell^1$-based data-fidelity term and $R$ is a weighted $\ell^1$ norm. We aim at minimizing $G$ in a greedy fashion. Following the OMP paradigm, we construct the support  of the signal one index at a time. Specifically, at iteration $ k $, the support set $ S^{(k)} $ of the approximation $x^{(k)}$ is updated according to the following \emph{greedy index selection rule}:
$$ S^{(k)} = S^{(k - 1)} \cup \{j^{(k)}\},\quad \text{where}\quad j^{(k)} \in \arg\max_{j \in [N]} \Delta(x^{(k - 1)}, S^{(k - 1)}, j), $$
where $ \Delta(x^{(k-1)},S^{(k-1)},j) $ is the loss reduction resulting from adding a single index $j$ to the support $S^{(k-1)}$ of $x^{(k-1)}$ and with $[N]:=\{1,\ldots,N\}$. $\Delta$ is implicitly defined by
\begin{equation}
\label{eq:def_Delta_intro}
\min_{t \in \mathbb{F}} G(z + te_j) = G(z) - \Delta(z, S, j), \quad \text{where } S: = \supp(z), \quad \forall z \in \mathbb{F}^N.
\end{equation}
 Then, the signal is updated by solving a local optimization problem restricted to the newly constructed support $ S^{(k)} $, i.e.,
\begin{equation}
\label{eq:local_opt_intro}
	x^{(k)} \in \arg\min_{z \in\mathbb{F}^N} F(z)\quad \text{s.t}\quad \supp(z) \subseteq S^{(k)}.
\end{equation}
The corresponding loss function-based OMP algorithm is presented in Algorithm~\ref{alg:general_OMP} (adopting a stopping criterion based on the number of iterations).
\begin{algorithm}[h!]
	\caption{Loss function-based OMP}
	\label{alg:general_OMP}
	\begin{algorithmic}[1]
		\State \textbf{Inputs:}
              $G : \mathbb{F}^N \to [0,+\infty)$ (loss function of the form \eqref{eq:def_G_intro}), with $\mathbb{F} = \mathbb{R}$ or $\mathbb{C}$;
		$ A \in \mathbb{F}^{m \times N} $ (measurement matrix);
		$ y \in \mathbb{F}^m $ (measurement vector);
		$ K \in [N] $ (number of iterations).
            \State \textbf{Output:} $ x^{(K)} \in \mathbb{F}^N $ (approximate $ K $-sparse solution to $ Az = y $).
		\Procedure{Loss-function based OMP}{$G$, $ A $, $ y $,  $ K $}
		\State Let $ x^{(0)} = 0 $ and $ S^{(0)} = \emptyset $
		\For{$ k = 1, \dots, K $}
		\State \normalsize Find $ j^{(k)} \in \arg\max_{j \in [N]}\Delta(x^{(k - 1)}, S^{(k - 1)}, j) $, with $\Delta$ defined as in \eqref{eq:def_Delta_intro}
		\State Define $ S^{(k)} = S^{(k - 1)} \cup \{j^{(k)}\} $
		\State Compute $ x^{(k)} $ by solving \eqref{eq:local_opt_intro}
		\EndFor
		\State \Return $x^{(K)}$
		\EndProcedure
	\end{algorithmic}
\end{algorithm}
\begin{remark}[Standard OMP]
  The standard OMP algorithm is a special case of Algorithm~\ref{alg:general_OMP} when $ G $ is the least-squares loss function, i.e., $ G(z) = F(z) = \|y - Az\|_2^2 $ and for $\lambda = 0$. 
  \end{remark}
  To demonstrate that Algorithm~\ref{alg:general_OMP} is practically implementable, we ought to show that the loss reduction factor $ \Delta (x, S, j)$ is (ideally, easily) computable. The main technical contribution of the paper is to show that this is indeed the case for the weighted LASSO, SR-LASSO, and LAD-LASSO loss functions (referred to as ``$*$-LASSO'' below). This is summarized in the following result, which unifies Theorems~\ref{thm:WLASSO}, \ref{thm:WSRLASSO} and \ref{thm:WLADLASSO}.
\begin{theorem}[Weighted $*$-LASSO-based greedy selection rules]
	\label{thm:general}
	Let $ \lambda \geq 0 $, $ S \subseteq [N] $, 
 $$ 
 A \in \mathbb{F}^{m \times N}\quad 
 \begin{cases}
    \text{with $\mathbb{F} = \mathbb{C}$ and $ \ell^2 $-normalized columns} & \text{(LASSO and SR-LASSO)} \\
    \text{with $\mathbb{F} = \mathbb{R}$ and nonzero columns} & \text{(LAD-LASSO)}
 \end{cases}
 $$ 
 and $ x \in \mathbb{F}^N $ satisfying 
$$ 
	x \in \arg\min_{z \in\mathbb{F}^N} F(z)\quad \text{s.t}\quad \supp(z) \subseteq S.
$$
 Then, the loss reduction $\Delta(x,S,j)$ defined in \eqref{eq:def_Delta_intro} admits explicit formulas provided by \eqref{eq:LASSO_WOMP_quantity}, \eqref{eq:SRLASSO_WOMP_quantity} and \eqref{eq:LADLASSO_WOMP_quantity}, respectively, for the weighted LASSO, SR-LASSO and LAD-LASSO loss functions (see \eqref{eq:loss_WLASSO}, \eqref{eq:loss_WSRLASSO} and \eqref{eq:loss_WLADLASSO}).
\end{theorem}

\subsection{Literature review}
\label{literature_review}
%\paragraph{Weighted sparsity.} 
Weights have been employed in sparse recovery methods for various purposes. For instance, in the seminal work \cite{candes2008enhancing}, the authors propose to solve a sequence of (re)weighted $ \ell^1 $ minimization problems to enhance sparse signal recovery. In our context, weights can generally be thought of as a way of incorporating prior information about the signal into a sparse recovery model. In \emph{adaptive LASSO} \cite{huang2008adaptive, zou2006adaptive}, a data-driven but careful choice of weights is shown to admit near oracle properties. Works such as \cite{friedlander2011recovering,yu2013sufficient} show that replacing the $ \ell^1 $-norm with its weighted version can improve recovery assuming that accurate (partial) support knowledge is provided. A similar result was derived in \cite{khajehnejad2011analyzing} from a probabilistic point of view where the signal support is assumed to be formed by two subsets with different probability of occurrence. Further studies of weighted $ \ell^1 $ minimization and its impactful application in the context of function approximation from pointwise samples and uncertainty quantification include \cite{adcock2018infinite, adcock2019correcting, adcock2022sparse, peng2014weighted, rauhut2016interpolation}. The notion of weighted sparsity was formalized in \cite{rauhut2016interpolation}. Weighted sparsity is  related to structured sparsity (see \cite{baraniuk2010model}). In fact it allows one to \emph{promote} structures (rather than being a structure itself). For example, in the context of high-dimensional function approximation (see \cite{adcock2022sparse} and references therein) weights are able to promote so-called \emph{sparsity in lower sets}, which largely contributes to mitigating the curse of dimensionality in the sample complexity. Using the weighted $ \ell^1 $ minimization to improve the sample complexity was also addressed in \cite{bah2016sample} in the signal processing context. Apart from convex $ \ell^1 $-minimization, weights are implemented in algorithms such as weighted IHT \cite{jo2013iterative}, and weighted OMP \cite{adcock2020sparse, bouchot2017multi, li2013weighted,wu2013weighted} (see below).

OMP and its non-orthogonalized version, Matching Pursuit (MP), were introduced in \cite{mallat1993matching,pati1993orthogonal} for time-frequency dictionaries, and later analyzed in, e.g., \cite{needell2009uniform,tropp2004greed,tropp2007signal}. Well-known advantages of OMP are its simple and intuitive formulation and its computational efficiency, especially for small values of sparsity. A lot of research has been devoted to improve OMP, e.g., by allowing the algorithm to select several indices at each iteration or combining it with thresholding strategies \cite{davies2008faster, donoho2012sparse,needell2009cosamp,needell2009uniform}, or by optimizing the greedy selection rule \cite{rebollo2002optimized}. The loss function-based perspective adopted in our work is related to the approach in \cite{shalev2010trading,zhang2011sparse}. However, there are at least two key differences with our setting: (i) we do not assume the loss function to be differentiable and (ii) the corresponding greedy selection criterion is not based on the gradient of the loss function. Our framework extends both the standard OMP algorithm and the weighted OMP algorithm proposed in \cite{adcock2020sparse}, based on $\ell^0$ regularization. To the best of our knowledge, the only other works that incorporate weights into OMP, but with different weighting strategies than those proposed here, are \cite{bouchot2017multi,li2013weighted,wu2013weighted}. 

Let us finally consider the family of greedy coordinate descent algorithms (see, e.g., \cite{li2009coordinate}). They aim to solve a given optimization problem by selecting one coordinate index at a time and minimizing the loss function with respect to the corresponding entry while freezing all the others. Although their greedy selection method coincides with the one adopted in this paper, greedy coordinate descent algorithms differ from loss function-based OMP since their greedy index selection is not combined with the solution of a local data-fitting optimization problem of the form \eqref{eq:local_opt_intro}. In addition, the greedy coordinate selection in \cite{li2009coordinate} is only explicitly computed for the unweighted LASSO, whereas here we derive explicit greedy index selection rules for weighted LASSO, SR-LASSO and LAD-LASSO. 

\subsection{Outline of the paper}
The rest of the paper is organized as follows. In \S2 we discuss in detail the loss function-based OMP framework summarized in \S\ref{s:summary_main} and present weighted $*$-LASSO-based greedy selection rules. Then, we illustrate the practical performance of $*$-LASSO-based OMP through numerical experiments in \S\ref{sec:numerics} and outline open problems and future research directions in \S\ref{s:conclusions}. Appendix~\ref{app:proofs} contains the proofs of Theorems~\ref{thm:WLASSO}, \ref{thm:WSRLASSO} and \ref{thm:WLADLASSO}, stated in \S\ref{sec:greedy_selection}. In Appendix~\ref{app:l0_greedy_selection}, we present $ \ell^0 $-based variants of the proposed algorithms.

\section{LASSO-based weighted OMP}

In this section we present LASSO-based weighted OMP (WOMP) algorithms. In order to theoretically justify our methodology, we first review the rationale behind greedy algorithms such as OMP, emphasizing the role played by certain (regularized) loss functions. 

\subsection{Loss function-based OMP}
\label{ss:loss_function_based_OMP}
Greedy algorithms such as OMP are iterative procedures characterized by the following two steps:
\begin{enumerate}[(i)]
\item the iterative construction of signal's support by means of greedy index selection; \label{greedy:support}
\item the computation of signal's entries on (a subset of) the constructed support by solving a ``local'' optimization problem.  \label{greedy:optimization}
\end{enumerate}
In this section, we describe a general paradigm to perform these two operations (and, consequently, design greedy algorithms) from the perspective of loss functions. Specifically, we consider an optimization problem of the form
\begin{equation}
\label{eq:general_loss}
\min_{z \in\mathbb{C}^N} G(z) := \min_{z \in\mathbb{C}^N} \left(F(z) + \lambda R(z)\right),
\end{equation}
where $G, F, R:\mathbb{C}^N \to [0,+\infty)$ and $\lambda \geq 0$. Here $G$ is a (regularized) loss function, composed by a \emph{data fidelity} term $F$ and a \emph{regularization} term $R$, balanced by a \emph{tuning parameter} $\lambda$.

Aiming to minimize $G$, in Step \eqref{greedy:support} an OMP-type greedy algorithm constructs the signal support by selecting the index (or indices) leading to a maximal reduction of the loss function $G$---this is why this type of algorithm is called ``greedy''. 
Specifically, given a support set $S^{(k-1)}$ and an approximation $x^{(k-1)}$, at iteration $k$ the algorithm constructs a new index set $S^{(k)}$ as follows:
\begin{equation*}
	S^{(k)} = S^{(k-1)} \cup \{j^{(k)}\}, \quad \text{where} \;
	j^{(k)} \in \arg\max_{j\in[N]} \Delta(x^{(k-1)}, S^{(k - 1)}, j), 
\end{equation*}
with $\Delta:\mathbb{C}^N \times 2^{[N]} \times [N] \to [0,+\infty)$ (where $2^X$ denotes the power set of $X$) implicitly defined by
\begin{equation}
\label{eq:def_Delta}
\min_{t \in \bC}G(x^{(k-1)} + te_j) := G(x^{(k-1)}) - \Delta(x^{(k-1)}, S^{(k - 1)}, j).
\end{equation}
Here $\Delta(z, S, j)$ is the loss function reduction corresponding to adding the index $j \in [N]$ to the support $S \subseteq [N]$ and given a current approximation $z \in \mathbb{C}^N$. In fact, rearranging the above relation leads to
$$
\Delta(x^{(k-1)}, S^{(k - 1)}, j)  = \max_{t \in \bC} [ G(x^{(k-1)}) - G(x^{(k-1)} + te_j) ].
$$

After a suitable updated support $S^{(k)}$ is identified, in Step \eqref{greedy:optimization} the approximation $x^{(k-1)}$ is updated as $x^{(k)}$ by solving a local data-fitting optimization problem. This optimization problem takes the form
\begin{equation}
	\label{eq:general_greedy_step_2}
	x^{(k)} \in \arg\min_{z \in \bC^N}F(z)\quad \text{s.t.}\; \supp(z) \subseteq S^{(k)}.
\end{equation}
Note that this local optimization only involves the data-fidelity term $F$ and not the regularization term $R$. As we will see, this will lead to theoretical benefits in order to formally certify that $\Delta$ corresponds to the maximal reduction of $G$. Moreover, the choice of the support $S^{(k)}$ is already regularized. Therefore, the local optimization performs only a data fitting step onto the regularized subspace $\{z\in\mathbb{C}^N : \supp(z) \subseteq S^{(k)}\}$. This can be summarized in the following iteration.

\begin{center}
\begin{tcolorbox}[colback = white]
\textbf{Loss function-based OMP iteration}
\begin{align}
	\label{eq:general_greedy_step_1}
	S^{(k)} &= S^{(k-1)} \cup \{j^{(k)}\}, \quad \text{with} \;
	j^{(k)} \in \arg\max_{j\in[N]} \Delta(x^{(k-1)}, S^{(k - 1)}, j) \text{ and } \Delta \text{ as in \eqref{eq:def_Delta}}\\
	\label{eq:general_greedy_step_2}
	x^{(k)} & \in \arg\min_{z \in \bC^N}F(z)\quad \text{s.t.}\; \supp(z) \subseteq S^{(k)}
\end{align}
\end{tcolorbox}
\end{center}
We now revisit the standard OMP algorithm and the weighted OMP algorithm of \cite{adcock2020sparse} in light of the above framework.

\paragraph{Standard OMP.} With the above discussion in mind, we consider the least-squares loss function, without regularization (i.e., $\lambda =0$),
\begin{equation}
	\label{eq:least_squares_loss}
    G^{\mathrm{LS}}(z)= F^{\mathrm{LS}}(z) := \|y - Az\|_2^2, \quad \forall z \in \mathbb{C}^N,
\end{equation}
where $ y \in \bC^m $ is a vector of measurements (or observations) and $ A \in \bC^{m \times N} $ is a measurement (or design) matrix with $\ell^2$-normalized columns. With this choice, steps \eqref{eq:general_greedy_step_1} and \eqref{eq:general_greedy_step_2} correspond to the following well-known iteration of the OMP algorithm:
\begin{align}
S^{(k)} &= S^{(k-1)} \cup \{j^{(k)}\},\quad \text{where}\;j^{(k)} \in \arg\max_{j \in [N]} \Delta^{\mathrm{LS}},\quad \Delta^{\mathrm{LS}} = |(A^*(y - Ax^{(k-1)}))_j|, \label{eq:omp_step_1}\\
x^{(k)} &\in \arg\min_{z \in \bC^N}\|y - Az\|_2^2\quad \text{s.t.}\quad \supp(z) \subseteq S^{(k)} \label{eq:omp_step_2}.
\end{align}
Interestingly, in OMP the index selected at each iteration maximizes, at the same time, the correlation between the columns of the matrix $ A $ and the residual vector $ r^{(k-1)}:= y - Ax^{(k-1)} \in \bC^m $ and the least-squares loss reduction. In fact, it is possible to show that (see, e.g., \cite[Lemma 3.3]{foucart2013mathematical}) the problem
$$ \min_{t \in \bC}G^{\mathrm{LS}}(z + te_j) = G^{\mathrm{LS}}(z) - |(A^*(y - Az))_j|^2 $$
prescribes the choice of the new greedy index.

We note that sparsity is not directly promoted by minimizing the least-squares loss function $G^{\mathrm{LS}}$ of OMP. In fact, in OMP the sparsity of the approximated solution is directly related with the number of iterations. 
Specifically, each iteration adds a single index to the support $S^{(k)}$. Hence, 
running $ s $ iterations of OMP generates an $s$-sparse vector (i.e., with $\|x^{(s)}\|_0 \leq s$). Although very powerful in the case of standard sparsity, standard OMP does not directly allow one to promote other sparsity structures, such as, e.g., weighted sparsity \cite{rauhut2016interpolation}.\\

In this paper, we focus on algorithms that are able to promote weighted sparsity  \cite{rauhut2016interpolation}. Recall that, given a vector of weights $ w \in \bR^N $ with $ w > 0 $, the weighted $\ell^0_w$- and $\ell^1_w$-norm of a vector $ z \in \bC^N $ are defined as %
\begin{equation}
\label{eq:def_l0w_l1w_norms}
\|z\|_{0, w} := \sum_{j \in \supp(z)}w_j^2 \\
\quad \text{and} \quad
\|z\|_{1, w} := \sum_{j \in \supp(z)}w_j|x_j|,
\end{equation}
respectively \cite{rauhut2016interpolation}. As the use of weights allows one to encourage hidden structures in the ground truth signal, it is highly application dependent. Examples of specific applications with explicit weight choices include sparse polynomial approximation~\cite{adcock2022sparse}, recovery with partial support information \cite{friedlander2011recovering,yu2013sufficient} and sparse-in-levels signal reconstruction \cite{adcock2021compressive}. In accordance with the loss function-based perspective adopted in this section, we promote weighted sparsity through the regularized loss function $G$, in particular, by suitable choice of the regularization term $R$.
This idea was recently employed in \cite{adcock2020sparse} in the context of sparse high-dimensional function approximation, as illustrated in the next paragraph.

\paragraph{$\ell^0_w$-based Weighted OMP ($\ell^0_w$-WOMP).} Inspired by the unconstrained LASSO formulation (see \S\ref{subsub:weighted_l1_formulations}) \cite{adcock2020sparse} suggested to adopt the $\ell^0_w$-regularized least squares loss function
\begin{equation}
\label{eq:l0_lasso_loss}
	G_{\ell^0_w}(z) = \|y - Az\|_2^2 + \lambda \|z\|_{0, w} \quad \forall z \in \mathbb{C}^N.
\end{equation}
Although the loss function $G_{\ell^0_w}$ is nonconvex and discontinuous, the corresponding loss reduction function $\Delta_{\ell^0_w}$ can be explicitly computed as
\begin{equation}
		\label{eq:l0_WOMP_quantity}
		\Delta_{\ell^0_w}(x, S, j) = \begin{cases}
			\max\{|(A^*(Ax - y))_j|^2 - \lambda w_j^2, 0\} & j \notin S \\ 
			\max\{\lambda w_j^2 - |x_j|^2, 0\} & j \in S, \; x_j \neq 0   \\
			0 & j \in S, \; x_j = 0
		\end{cases},
\end{equation}
under the assumption that $A$ has $\ell^2$-normalized columns and that 
\begin{equation*}
		x \in \arg\underset{z \in \bC^N}{\min} \|y - Az\|_2^2 \quad \text{s.t.} \quad \mathrm{supp}(z) \subseteq S.
\end{equation*}
For more details and a proof of this result, we refer to \cite[Proposition~1]{adcock2020sparse}.
This leads to what we will refer to as the $\ell^0_w$-weighted OMP ($\ell^0_w$-WOMP) algorithm, defined by the following iteration:
	\begin{align}
		S^{(k)} &= S^{(k-1)} \cup j^{(k)},\quad \text{where}\;j^{(k)} = \arg\max_{j\in[N]} \Delta_{\ell^0_w}(x^{(k-1)}, S^{(k-1)}, j), \\
		x^{(k)} &\in \arg\min_{z \in \bC^N}\|y - Az\|_2^2\quad \text{s.t}\quad \supp(z) \subseteq S^{(k)} \label{eq:l0_WOMP_step_2}.
	\end{align}
Note that for $\lambda = 0$, the $\ell^0_w$-WOMP algorithms coincides with standard OMP.
On top of allowing one to incorporate weights, using a regularized loss function such as $G_{\ell^0_w}$ also improves the robustness of OMP with respect to the stopping criterion. In fact, the presence of a regularization term prevents the greedy algorithm from \emph{overfitting} due to an excessive number of iterations (see \cite{adcock2020sparse} for details on numerical results).\\

However, there is no free lunch. The possibility of including weights and the improved robustness with respect to the number of iterations come at the cost of adding an extra parameter $ \lambda $ that needs to be tuned appropriately. Unfortunately, the optimal value of $ \lambda $ (i.e., the value that minimizes the reconstruction error) depends on characteristics of the model such as the sparsity of the ground truth signal or the magnitude of the noise corrupting the measurements. This makes $\lambda$ challenging to tune in general. Luckily, some insights on how to tune $\lambda$ can be found in the convex optimization literature for LASSO-type loss functions. These are described in the next subsection and constitute the foundation upon which we will design the class of LASSO-inspired greedy algorithms proposed in this paper.

\subsection{LASSO-type loss functions for weighted $\ell^1$ minimization}\label{subsub:weighted_l1_formulations}
In this subsection we introduce different convex optimization programs that have been extensively used for weighted sparse signal recovery. Consider a vector of weights $ w \in\mathbb{R}^N$ with $w >0$. We aim to recover a sparse vector $ x \in \bC^N $ from measurements $ y = Ax + e  \in \bC^m $, where $ e \in \bC^m $ is an error or noise vector corrupting the measurements. This could include errors from various sources, such as physical noise (e.g., from measurement devices), numerical or discretization error (e.g., from numerical solvers), or sparse corruptions (e.g., from node failures in a parallel computing setting). 
In the context of weighted sparse recovery, an approximation to the signal $x$ from noisy measurements $y$ can be computed by solving one of the unconstrained weighted $\ell^1$ minimization problems discussed below.\footnote{Here we do not consider constrained programs such as quadratically constrained basis pursuit or constrained LASSO (see, e.g., \cite{adcock2021compressive,adcock2022sparse,rauhut2016interpolation}) since they do not fit our framework.} We organize our discussion based on the nature of the noise $e$ corrupting the measurements.

\paragraph{Bounded noise: weighted LASSO and SR-LASSO.} If  the noise satisfies a bound of the form $\|e\|_2\leq \eta$ for a small constant $\eta$ (that might be known or unknown in advance), weighted quadratically constrained basis pursuit is one of the most popular weighted $\ell^1$-minimization strategies \cite{adcock2021compressive, adcock2022sparse, rauhut2016interpolation}. However, it requires the knowledge of $\eta$ and it is a constrained optimization problem---hence, it is not of the form \eqref{eq:general_loss}. For this reason, we do not consider it further in this paper. A popular recovery strategy of the form \eqref{eq:general_loss} is the (unconstrained) \emph{weighted LASSO}, defined by the loss function
\begin{equation}
\label{eq:loss_WLASSO}
G_{\ell^1_w}(z) :=  \|y - Az\|_2^2 + \lambda \|z\|_{1, w}, \quad \forall z \in \mathbb{C}^N.
\end{equation}
The LASSO dates back at least to the pioneering works \cite{santosa1986linear,tibshirani1996regression} and since then has become one of the most widespread optimization problems in statistics and data science. Although the LASSO eliminates the need for an explicit knowledge of $\eta$, the choice of its tuning parameter $\lambda$ is not straightforward. The range of values of $\lambda$ that lead to theoretical optimal recovery guarantees scales linearly in $\|e\|_2$, i.e., more specifically, $ \lambda \asymp \|e\|_2/\sqrt{\|x\|_{0,w}} $, \cite{adcock2019correcting} (or, when $e$ is a normal random vector, on its standard deviation; see, e.g., \cite{bickel2009simultaneous,shen2015stable}). In practice, this means that $\lambda$ should often be tuned via \emph{cross validation} (see, e.g.,  \cite{hastie2009elements}) that, although generally accurate, is often computationally daunting. 

To alleviate this issue, an alternative strategy of the form \eqref{eq:general_loss} is the weighted \emph{Square-Root LASSO (SR-LASSO)}, whose loss function is defined by
\begin{equation}
\label{eq:loss_WSRLASSO}
G^{\textnormal{SR}}_{\ell^1_w}(z) := \|y - Az\|_2 + \lambda \|z\|_{1, w}, \quad \forall z \in \mathbb{C}^N.
\end{equation}
The (unweighted) SR-LASSO was proposed in \cite{belloni2011square}. It is a well known optimization problem in statistics (see, e.g., the book \cite{van2016estimation}), and has become increasingly popular in compressive sensing \cite{adcock2019correcting,adcock2022log,adcock2022sparse, foucart2022sparsity, petersen2022robust}. There is only a small difference between the objective functions of SR-LASSO and LASSO, i.e., the lack of the exponent 2 on the data-fidelity term of SR-LASSO. However, this slight difference gives rise to substantial changes. It has been demonstrated both theoretically and empirically that the optimal choice of $\lambda$ for (weighted) SR-LASSO is no longer dependent on the noise level, i.e., $ \lambda \asymp 1/\sqrt{\|x\|_{0,w}} $, which facilitates parameter tuning in the presence of unknown bounded noise \cite{adcock2019correcting,belloni2011square}.

\paragraph{Sparse corruptions: weighted LAD-LASSO.} When the noise corrupting the measurements is of the form 
$$e = e^{\textnormal{bounded}} + e^{\textnormal{sparse}},$$ 
where $\|e^{\textnormal{bounded}}\|_2$ and $\|e^{\textnormal{sparse}}\|_0$ are bounded, but $\|e^{\textnormal{sparse}}\|_2$ is possibly very large, the LASSO and SR-LASSO  are generally not able to achieve successful sparse recovery. A simple remedy is to use a data-fidelity term based on the $\ell^1$-norm, as opposed to the $\ell^2$-norm, thanks to its ability to promote sparsity on the residual. This is the idea behind the (unconstrained) weighted \emph{LAD-LASSO}, whose loss function is given by
\begin{equation}
\label{eq:loss_WLADLASSO}
G^{\textnormal{LAD}}_{\ell^1_w}(z) :=\|y - Az\|_1 + \lambda \|z\|_{1, w}, \quad \forall z \in \mathbb{C}^N.
\end{equation}
Early works on LAD-LASSO include \cite{laska2009exact,wang2007robust}. It is a regularized version of the classical \emph{Least Absolute Deviations} (\emph{LAD}) problem \cite{bloomfield1983least,candes2005decoding}. In addition to signal's weighted sparsity, in weighted LAD-LASSO, the optimal choice of the tuning parameter further depends on the sparsity of $ e^{\textnormal{sparse}} $, i.e., $ \lambda \asymp \sqrt{\|e^{\textnormal{sparse}}\|_0/\|x\|_{0,w}} $ \cite{adcock2019correcting, adcock2018compressed}. Nonetheless, the choice $ \lambda \asymp 1 $ usually works well in practice \cite{adcock2018compressed}.\\

\subsection{Two key questions}

Our objective for the rest of the paper is to study OMP-type greedy algorithms characterized by the iteration \eqref{eq:general_greedy_step_1}-\eqref{eq:general_greedy_step_2}, based on the weighted LASSO, SR-LASSO, and LAD-LASSO loss functions defined in \eqref{eq:loss_WLASSO}, \eqref{eq:loss_WSRLASSO}, and \eqref{eq:loss_WLADLASSO}, respectively. Our investigation is driven by two main questions: 

\begin{itemize}
\item [(Q1)] \emph{Is the quantity $\Delta$ in \eqref{eq:def_Delta}, defining the OMP-type greedy selection rule, explicitly computable for the weighted LASSO, SR-LASSO, and LAD-LASSO loss functions?}

\item [(Q2)] \emph{Are the favorable properties of the weighted SR-LASSO, and LAD-LASSO inherited by the corresponding OMP-type greedy algorithms?}
\end{itemize}
We will provide affirmative answers to both (Q1) and (Q2). The answer to (Q1) will be accompanied by explicit formulas for $\Delta$ and rigorous loss-function reduction guarantees, discussed in \S\ref{sec:greedy_selection}. The answer to (Q2) will be based on numerical evidence, presented in \S\ref{sec:numerics}. Specifically, we will show that SR-LASSO-based OMP admits a \emph{noise robust} optimal parameter tuning strategy (i.e., the optimal value of $\lambda$ is independent to the noise level) and that LAD-LASSO-based OMP is \emph{fault tolerant}, i.e.\ able to correct for high-magnitude sparse corruptions. 

\begin{remark}[$\ell^0_w$-based regularization]
It is possible to consider $\ell^0_w$-based loss functions for the SR-LASSO and the LAD-LASSO, similarly to the $\ell^0_w$-regularized least squares loss defined in \eqref{eq:l0_lasso_loss} and employed in \cite{adcock2020sparse} (which correspond to an $\ell^0_w$-based LASSO formulation). However, we have observed experimentally that (Q2) does not admit an affirmative answer for the $\ell^0_w$-based analogs (see Experiments I and II in \S\ref{s:numerics_lambda}). For this reason, we refrained from studying the $\ell^0_w$-based formulations in detail in the present paper. Nonetheless, we provide explicit formulas for $\Delta$ for these variants in Appendix~\ref{app:l0_greedy_selection}. 
\end{remark}

\section{Greedy selection rules for weighted LASSO-type loss functions}
\label{sec:greedy_selection}

Equipped with the general loss function-based OMP paradigm presented in \S\ref{ss:loss_function_based_OMP}, we present three weighted OMP iterations based on the LASSO, square-root LASSO, and LAD-LASSO loss functions reviewed in \S\ref{subsub:weighted_l1_formulations}. The proofs of the results in this section can be found in Appendix~\ref{app:proofs}.

\subsection{LASSO-based OMP}
We start by considering the LASSO loss function $G_{\ell^1_w}$ defined in \eqref{eq:loss_WLASSO}, whose corresponding greedy selection rule is identified by the following result.
\begin{theorem}[LASSO-based greedy selection rule]
	\label{thm:WLASSO}
	Let $ \lambda \geq 0 $, $ S \subseteq [N] $, $ A \in \bC^{m \times N} $ with $ \ell^2 $-normalized columns, and $ x \in \bC^N $ be such that
\begin{equation}
	x \in \arg\min_{z \in \bC^N} \|y - Az\|_2^2 \quad \text{s.t.} \quad \mathrm{supp}(z) \subseteq S.
\end{equation}
Then, for every $ j \in [N] $, 
$$ \min_{t \in \mathbb{C}} G_{\ell^1_w}(x + te_j) = G_{\ell^1_w}(x) - \Delta_{\ell^1_w}(x, S, j), $$
where 
\begin{equation}
\label{eq:LASSO_WOMP_quantity}
\Delta_{\ell^1_w}(x, S, j) 
= \begin{cases}
\max\left\{|(A^*(Ax - y))_j| - \frac{\lambda}{2}w_j, 0\right\}^2 & j \notin S \\
\max\left\{
|x_j|(\lambda w_j-|x_j|), 
\lambda w_j \left(|x_j|-\frac{\lambda w_j}{4} - \left||x_j| - \frac{\lambda w_j}{2}\right|\right), 
0
\right\} & j \in S%,\; x \neq 0 \\
%0 & j \in S, \; x = 0
\end{cases}.
\end{equation}
\end{theorem}
This leads to the following LASSO-based OMP iteration:
	\begin{align}
		S^{(k)} &= S^{(k-1)} \cup \{j^{(k)}\}\quad \text{where}\quad j^{(k)} \in \arg\max_{j\in[N]} \Delta_{\ell^1_w}(x^{(k-1)}, S^{(k-1)}, j) \label{eq:LASSO_OMP_step_1}\\
		x^{(k)} &\in \arg\min_{z \in \bC^N}\|y - Az\|_2^2\quad \text{s.t. }\supp(z) \subseteq S^{(k)} \label{eq:LASSO_OMP_step_2}.
	\end{align}

\subsection{SR-LASSO-based OMP}
For the SR-LASSO loss function $G_{\ell^1_w}^{\mathrm{SR}}$ defined in \eqref{eq:loss_WSRLASSO}, we have the following result.
\begin{theorem}[SR-LASSO-based greedy selection rule]
	\label{thm:WSRLASSO}
	Let $ \lambda \geq 0 $, $ S \subseteq [N] $, $ A \in \bC^{m \times N} $ with $ \ell^2 $-normalized columns, and $ x \in \bC^N $ satisfying 
\begin{equation}
	x \in \arg\min_{z \in \bC^N} \|y - Az\|_2 \quad \text{s.t.} \quad \mathrm{supp}(z) \subseteq S.
\end{equation}
Then, for every $ j \in [N] $,
$$ \min_{t\in \mathbb{C}} G_{\ell^1_w}^{\mathrm{SR}}(x + te_j) = G_{\ell^1_w}^{\mathrm{SR}}(x) - \Delta_{\ell^1_w}^{\mathrm{SR}}(x, S, j), $$
where 
\begin{equation}
	\label{eq:SRLASSO_WOMP_quantity}
	\Delta_{\ell^1_w}^{\mathrm{SR}}(x, S, j) = \begin{cases}
		\max \left\{\|r\|_2 - \lambda w_j |\langle r, a_j\rangle| - \sqrt{(1 - (\lambda w_j)^2)(\|r\|_2^2 - |\langle r, a_j\rangle|^2)}, 0\right\} &  j \notin S \\
		\|r\|_2 - \sqrt{\widetilde{\rho}^2 + \|r\|_2^2} + \lambda w_j\left( |x_j|  - \left||x_j| - \widetilde{\rho}\right|\right) &  j \in S\\ %\; x \neq 0 \\
		%0 &  j \in S, \; x = 0
	\end{cases},
\end{equation} 
with $ r = y - Ax $ and 
$$
\widetilde{\rho} 
    := \begin{cases}
    |x_j| & \lambda w_j \geq 1 \\
    \min\left\{|x_j|, \frac{\lambda w_j \|r\|_2}{\sqrt{1 - (\lambda w_j)^2}}\right\}& \lambda w_j < 1 
    \end{cases}.
$$
\end{theorem}
The corresponding SR-LASSO-based OMP iteration reads
	\begin{align}
			S^{(k)} &= S^{(k-1)} \cup \{j^{(k)}\}\quad \text{where}\quad j^{(k)} \in \arg\max_{j\in[N]} \Delta_{\ell^1_w}^{\mathrm{SR}}(x^{(k-1)}, S^{(k-1)}, j), \label{eq:SR-LASSO_OMP_step_1}\\
		x^{(k)} &\in \arg\min_{z \in \bC^N}\|y - Az\|_2\quad \text{s.t. }\supp(z) \subseteq S^{(k)}.\label{eq:SR-LASSO_OMP_step_2}
	\end{align}

\subsection{LAD-LASSO-based OMP}
Finally, we consider the LAD-LASSO loss function $G_{\ell^1_w}^{\mathrm{LAD}}$ defined in \eqref{eq:loss_WLADLASSO}. In this case, we restrict ourselves to the real-valued case for the sake of simplicity. In order to formulate the corresponding greedy selection rule, we need to introduce some auxiliary notation. First, we define an augmented version $\widetilde{A} \in \mathbb{R}^{(m+1) \times N}$ of the matrix $A \in \mathbb{R}^{m\times N}$ as 
\begin{equation}
\label{eq:def_Atilde}
\widetilde{A} := 
\begin{bmatrix}A \\ \lambda w^*\end{bmatrix}
\quad \text{ or, equivalently, } \quad 
\widetilde{A}_{ij}
:=
\begin{cases}
A_{ij} & i \in [m], \; j \in [N]\\
\lambda w_j & i =m+1, \; j\in [N]
\end{cases}.
\end{equation}
In addition, given $x\in \mathbb{R}^N$, we consider $N$ augmentations of the residual vector $r = Ax-y \in \mathbb{R}^N$ as the vectors $\widetilde{r}^{\, j} \in \mathbb{R}^{m+1}$, defined by
\begin{equation}
\label{eq:def_rjtilde}
\widetilde{r}^{\, j} := 
\begin{bmatrix}r \\ -\lambda w_j x_j\end{bmatrix}
\quad \text{ or, equivalently, } \quad
\widetilde{r}_i^{\, j}
:=
\begin{cases}
r_i & i \in [m]\\
-\lambda w_j x_j & i =m+1
\end{cases},
\quad \forall j \in [N].
\end{equation}
Let $\widetilde{a}_j$ be the $j$th column of $\widetilde{A}$ and $ \tau_j: [\|\widetilde{a}_j\|_0] \to \supp(\widetilde{a}_j) $ be a bijective map defining a nondecreasing rearrangement of the vector 
$$
\left(\frac{\widetilde{r}_i^{\, j}}{\widetilde{A}_{ij}}\right)_{i \in \supp(\widetilde{a}_j)} \in \bR^{\|\widetilde{a}_j\|_0},
$$
i.e., such that
\begin{equation}
\label{eq:def_tau_j_nondecreasing}
\frac{\widetilde{r}^{\, j}_{\tau_j(1)}}{\widetilde{A}_{\tau_j(1), j}} 
\leq \frac{\widetilde{r}^{\, j}_{\tau_j(2)}}{\widetilde{A}_{\tau_j(2), j}} 
\leq \dots 
\leq \frac{\widetilde{r}^{\, j}_{\tau_j(\|\widetilde{a}_j\|_0)}}{\widetilde{A}_{\tau_j(\|\widetilde{a}_j\|_0), j}}.
\end{equation}
We are now in a position to state our result.

\begin{theorem}[LAD-LASSO-based greedy selection rule]
	\label{thm:WLADLASSO}
Let $ \lambda \geq 0 $, $ S \subseteq [N] $, $ A \in \bR^{m \times N} $ with nonzero columns, and $ x \in \bR^N $ satisfying 
\begin{equation}
	x \in \arg\min_{z \in \bR^N} \|y - Az\|_1 \quad \mathit{s.t.} \quad \mathrm{supp}(z) \subseteq S.
\end{equation}
Then, for every $ j \in [N] $,
$$ \min_{t \in \bR} G_{\ell^1_w}^{\mathrm{LAD}}(x + te_j) = G_{\ell^1_w}^{\mathrm{LAD}}(x) - \Delta_{\ell^1_w}^\mathrm{LAD}(x, S, j), $$
where 
\begin{equation}
	\label{eq:LADLASSO_WOMP_quantity}
	\Delta_{\ell^1_w}^{\mathrm{LAD}}(x, S, j) = 
 \lambda w_j|x_j| + \|r\|_1
 - \left\|\widetilde{r}^{\, j} - \frac{\widetilde{r}^{\, j}_{\hat{i}(j)}}{\widetilde{A}_{\hat{i}(j), j}} \widetilde{a}_j\right\|_1,
\end{equation}
with $ \hat{i}(j) := \tau_j(\hat{k}(j)) $ and 
\begin{equation}
    \label{w_median_l1_A_tilde}
\hat{k}(j) := 
	\min\left\{k \in [\|\widetilde{a}_j\|_0]: \sum_{i = 1}^{k}\frac{|\widetilde{A}_{\tau_j(i), j}|}{\|\widetilde{a}_j\|_1} \geq \frac{1}{2}\right\},
\end{equation} 
and where $\widetilde{A}$, $\widetilde{r}^{\, j}$, and $\tau_j$ are defined as in \eqref{eq:def_Atilde}, \eqref{eq:def_rjtilde}, and \eqref{eq:def_tau_j_nondecreasing}, respectively.
\end{theorem}
This proposition leads to the LAD-LASSO-based OMP iteration
	\begin{align}
		S^{(k + 1)} &= S^{(k)} \cup j^{(k)}\quad \text{where}\quad j^{(k)} = \arg\max_j \Delta_\lambda^{\mathrm{LASSO}}(x^{(k)}, S^{(k)}, j) \label{eq:LADLASSO_WOMP_step_1} \\
		x^{(k + 1)} & \in \arg\min_{z \in \bC^N}\|y - Az\|_1\quad \text{s.t}\quad \supp(z) \subseteq S^{(k)} \label{eq:LADLASSO_WOMP_step_2}.
	\end{align}
Some remarks are in order.
\begin{remark}[On terminology]
    The least-squares projection step of LASSO- and SR-LASSO-based OMP ensures orthogonality between the residual and the span of selected columns at each iteration. This property is no longer valid in the LAD-LASSO case because \eqref{eq:LADLASSO_WOMP_step_2} does not define an orthogonal projection. With a slight abuse of terminology, we will refer to the method defined by \eqref{eq:LADLASSO_WOMP_step_1}--\eqref{eq:LADLASSO_WOMP_step_2} as a variant of OMP, despite the lack of orthogonality.
\end{remark}
\begin{remark}[Solving LAD problems]
\label{remark:solving_LAD}
 Unlike the least-squares projection step of LASSO- and SR-LASSO-based WOMP, the LAD problem \eqref{eq:LADLASSO_WOMP_step_2} does not admit an explicit solution in general. However, one can take advantage of efficient convex optimization algorithms to approximately solve it. We note in passing that, for small values of $k$, the corresponding LAD problems over $\mathbb{R}^k$ are much cheaper to solve than an $\ell^1$ minimization problem over $\mathbb{R}^N$. In this paper, we numerically solve LAD problems using the MATLAB CVX package \cite{gb08,cvx} with MOSEK solver \cite{mosek9mosek}.
\end{remark}
\begin{remark}[An alternative strategy]
\label{remark:ladlasso_complex}
	An alternative LAD-LASSO-based OMP iteration can be derived by relaxing the LAD-LASSO to an augmented LASSO or SR-LASSO problem. Notably, this strategy works naturally in the complex case. Recall that our objective is to minimize $G_{\ell^1_w}^\mathrm{LAD}$ defined in \eqref{eq:loss_WLADLASSO} over $ \mathbb{C}^N$.	%
Now, for any $z\in \mathbb{C}^N$ let $ c = y - Az \in \mathbb{C}^m$ or, equivalently, $y = B t$, where
	$$ B := \begin{bmatrix}
		A & I
	\end{bmatrix} \in \mathbb{C}^{m \times (N+m)}
 \quad \text{and} \quad
 t := \begin{bmatrix}
	z \\
	c \\
	\end{bmatrix}\in \mathbb{C}^{N+m}. $$
	With this change of variable, a minimizer $\hat{z}$ of $G_{\ell^1_w}^\mathrm{LAD}$ over $\mathbb{C}^N$ satisfies
	$$ 
 \begin{bmatrix}
     \hat{z}\\ \hat{c}
 \end{bmatrix}
 \in \arg\min_{t \in \bC^{N + m}}\|t\|_{1, v}\quad  \mathit{s.t.}\quad Bt = y, 
 \quad \text{where }
	 v = \begin{bmatrix}
		\lambda w \\
		\bm{1} \\
	\end{bmatrix} $$
  for some $\hat{c} \in \mathbb{C}^m$, and where $\bm{1} \in \mathbb{C}^m$ is the vector of ones
	(see \cite{adcock2018compressed,li2013compressed} and references therein). This basis pursuit problem can be relaxed to a quadratically-constrained basis pursuit problem 
	$$ \min_{t \in \bC^{N + m}}\|t\|_{1, v}\quad  \mathit{s.t.}\quad \|y - B t \|_2 \leq \eta, $$
	with $\eta >0$. Now, one could consider either a LASSO or SR-LASSO reformulation of this problem. For example, in the LASSO case one would consider a loss function of the form
	\begin{equation}
		G(t)= \|y - B t\|_2^2 + \mu\|t\|_{1, v}, \quad t \in \mathbb{C}^{m + N}
	\end{equation}
	with $ \mu > 0 $, which leads to a LASSO-based OMP method. It is worth observing that a possible disadvantage of this strategy is the introduction of an extra tuning parameter $\mu$. 
\end{remark}

\section{Numerical experiments}
\label{sec:numerics}
In this section we present numerical results for the proposed LASSO-based WOMP algorithms. All the numerical experiments were performed in MATLAB 2017b 64-bit on a laptop equipped with a 2.4 GHz Intel Core i5 processor and 8 GB of DDR3 RAM. In some experiments, we compare our proposed algorithms with convex optimization-based recovery strategies. In these cases, we use the MATLAB CVX package \cite{gb08,cvx} with MOSEK solver \cite{mosek9mosek} and set \texttt{cvx\_precision best}. For the sake of convenience, we sometimes use MATLAB's vector notation. For example, $10.\,\hat{}\,(1:2:5)$ denotes the vector $ (10^1, 10^3, 10^5) $. The source code needed to reproduce our numerical experiments can be found on the GitHub repository \url{http://github.com/sina-taheri/Greedy_LASSO_WOMP}.

The section is organized as follows. In \S\ref{s:numerical_settings}, we start by presenting three settings used to validate and test the proposed algorithms. In \S\ref{s:numerics_lambda}, we carry out a first set of experiments aimed at studying the effect of the tuning parameter on the recovery error for different levels of noise or corruption, and for different weights' values. \S\ref{s:numerics_iterations} is dedicated to investigating the connection between the iteration number of the proposed greedy methods and the recovery error. We conclude by illustrating experiments on algorithms' runtime, loss function decay and a discussion on the stopping criteria in \S\ref{s:runtime_loss_stop}.

\subsection{Description of the numerical settings}
\label{s:numerical_settings}

The three numerical settings employed in our  experiments are illustrated below.

\paragraph{(i) Sparse random Gaussian setting (sparse and unweighted).}
First, we generate an $ s $-sparse random Gaussian vector $ x \in \bR^{N} $ as follows. $ S $, the support of $ x $, is generated by randomly and uniformly drawing a
 subset of $[N]$ of size $s$ (this avoids repeated indices). Within the support, the entries $ x $ are independently sampled from a Gaussian distribution with zero mean and unit variance, i.e., $ x_i \sim \mathcal{N}(0, 1)$, for every $ i \in S $.  This vector is measured by a sensing matrix $ A \in \bR^{m \times N} $ obtained after an $ \ell^2 $-normalization of the columns of a random Gaussian matrix $G \in \bR^{m \times N} $ with independent entries $ G_{i,j} \sim \mathcal{N}(0, 1) $ for every $i \in [m],\ j \in [N] $. The objective is to recover the synthetically-generated signal $ x $ from corrupted measurements, i.e., 
\begin{equation}
 	\label{eq:nuermical_model}
	y = Ax + e^{\mathrm{bounded}} + e^{\mathrm{unbounded}} \in \bR^m,\quad \text{where} \quad \|e^{\mathrm{bounded}}\|_2 = \eta,\ \|e^{\mathrm{unbounded}}\|_0 \leq K.
 \end{equation} 
 Here, $ e^{\mathrm{bounded}} = \eta e'/\|e'\|_2 \in \bR^m $ is a $ \ell^2 $-normalized random Gaussian vector with independent entries, i.e., $ e'_i \sim \mathcal{N}(0,1) $ for every $i \in [m]$ and $ e^{\mathrm{unbounded}} \in \bR^m $ is a $ K $-sparse vector generated by randomly and independently drawing $K$ integers uniformly from $[m]$ and filling the corresponding entries with independent random samples from $\mathcal{N}(0, M^2)$ for some $M > 0$. When we have unbounded noise in our measurements, we choose $ M $ to be very large, while in other cases we simply set it to zero. In this setting we consider unweighted recovery, i.e., $ w = \bm{1} $, the vector of ones.

\paragraph{(ii) Sparse random Gaussian setting with oracle (sparse and weighted).} Using the same model as in the previous setting, we acquire noisy measurements $ y = A x + e^{\mathrm{bounded}} + e^{\mathrm{unbounded}} \in \bR^m $ of a random $ s $-sparse vector $ x \in \bR^N $. In this second setting, we assume to have some \emph{a priori} knowledge of the support of $ x $ and incorporate this knowledge through weights in order to improve reconstruction. More precisely, we assume to know a set $ S^\mathrm{oracle} $ that partially approximates (i.e., that has nontrivial intersection with) the support of $ x $. Then, we define the weight vector $ w \in \bR^N $ as
\begin{equation}
	\label{eq:oracle_weights}
	w_j := \begin{cases}
		w_0 & j \in S^\mathrm{oracle} \\
		1 & j \notin S^\mathrm{oracle}
	\end{cases},
\end{equation}
for a suitable $ w_0 \in [0, 1] $. Note that if $ w_0 $ is chosen to be small, the contribution of signal coefficients weighted by $ w_0 $ is attenuated in the LASSO-type loss function. Consequently, activation of the corresponding indices is promoted in the greedy index selection stage of WOMP.

\paragraph{(iii) Function approximation (compressible and weighted).} In the third setting, the goal is to approximate a multivariate function
$$ f: D \to \bR,\quad D = [-1, 1]^d, $$
with $ d \gg 1 $, from pointwise evaluations $ f(t_1), f(t_2), \dots, f(t_m) $, where $t_1,\ldots,t_m$ are independently and identically sampled from a probability distribution $ \varrho $ over $ D $.  Here we briefly summarize how to perform this task efficiently via compressed sensing and refer the reader to the book \cite{adcock2022sparse} for a comprehensive treatment of the topic. This problem is mainly motivated by the study of quantity of interests in parametric models such as parametric differential equations, with applications to uncertainty quantification \cite{smith2013uncertainty}. Considering a basis of orthogonal polynomials $ \{\Psi_\nu\}_{\nu \in \mathbb{N}_0^d} $ for $ L^2_\varrho(D) $ (i.e., the Hilbert space of square-integrable functions over $D$ weighted by the probability measure $\varrho$). We aim to compute an approximation of the form
$$ f_{\Lambda} := \sum_{j \in [N]} x_{\nu_j}\Psi_{\nu_j} \approx f, \quad \text{where } \underbrace{\Lambda := \{\nu_j\}_{j \in [N]}}_{\text{truncation set}} \subset \mathbb{N}_0^d \quad \text{and} \quad  N \gg m, $$
and where $ x = (x_{\nu_j})_{j \in [N]} \in \bR^N $. This can be reformulated as a linear system in the coefficients $ x $, namely,
\begin{equation}
	\label{eq:func_approx_cs_model}
	y = Ax + e,
\end{equation}
where the measurement matrix $ A \in \bR^{m \times N} $ and the measurement vector $ y \in \bC^m $ are defined as
$$ A_{ij} := \frac{1}{\sqrt{m}}\Psi_{\nu_j}(t_i), \quad y_i := \frac{1}{\sqrt{m}}f(t_i), \quad \forall i \in [m], \; \forall j \in [N], 
$$
and where $ e \in \bC^m $ is the noise vector, including the inherent truncation error (depending on $\Lambda$) and, possibly, other types of error (e.g., numerical, model, or physical error). Under suitable smoothness conditions on $ f $, such as holomorphy, the vector of coefficients $ x $ is approximately sparse or compressible (see \cite[Chapter 3]{adcock2022sparse}). Therefore, the problem of approximating the function $ f $ is recast as finding a compressible solution $ x $ to the linear system \eqref{eq:func_approx_cs_model}. As a test function, we consider the so-called \emph{iso-exponential} function, defined as
\begin{equation}
	\label{eq:func_approx_isoexp}
	f(t) = \exp\left(-\sum_{i = 1}^{d} t_i/(2d)\right), \quad \forall t  \in D,
\end{equation}
which can be shown to be well approximated by sparse  polynomial expansions (see \cite[\S{A.1.1}]{adcock2022sparse}). Recovering $ x $ using LASSO-based WOMP algorithms, we set weights as
\begin{equation}
	\label{eq:func_approx_intrinsic_weights}
	w_j := \|\Psi_{\nu_j}\|_{L^\infty(D)} = \sup_{t\in D}|\Psi_{\nu_j}(t)|,\quad \forall j \in [N],
\end{equation}
known as \emph{intrinsic weights}. Note that these weights admit explicit formulas for, e.g., Legendre and Chebyshev orthogonal polynomials (see \cite[Remark 2.15]{adcock2022sparse}). In this paper, we will employ just Legendre polynomials.

\subsection{Recovery error versus tuning parameter}
\label{s:numerics_lambda}
The aim of Experiments I, II, and III presented in this section is to investigate the interplay between the tuning parameter $ \lambda $ and the recovery accuracy of LASSO-type WOMP algorithms in the three settings described in \S\ref{s:numerical_settings}, for different noise levels and  weight values. We recover $ x $ for a range of values of the tuning parameter $ \lambda $, at a fixed iteration number of the LASSO-type WOMP algorithms. We measure accuracy via the relative $\ell^2$-error
$$ E_\lambda = \frac{\|\hat{x}_\lambda - x\|_2}{\|x\|_2}, $$
where $ \hat{x}_\lambda $ denotes the computed approximation to $ x $ when the tuning parameter is set to $ \lambda $. Hence, we  plot the recovery error as a function of $ \lambda $. We repeat this experiment $ N_{\mathrm{trial}} $ number of times for different levels of noise and corruptions (in Experiments I and II), or different weight values (in Experiment III). The results of these statistical simulations are visualized using \emph{boxplots}, whose median values are linked by solid curves.

In Experiments I and II  we also consider $ \ell^0 $-based variants of LASSO-type WOMP algorithms. The $\ell^0$-based variant of LASSO WOMP was proposed in \cite{adcock2020sparse} and the greedy selection rules for $\ell^0$-based SR- and LAD-LASSO WOMP are derived in Appendix~\ref{app:l0_greedy_selection}. They constitute natural alternatives to the loss functions presented in \S\ref{sec:greedy_selection}, and we study their performance to justify our choice of $ \ell^1 $-based loss functions in this paper. 

\paragraph{Experiment I (sparse random Gaussian setting).} We begin with the sparse random Gaussian setting. Fig.~\ref{fig:sparse_gaussian_error_vs_lambda} shows results for recovery performed via $ \ell^0 $- and $ \ell^1 $-based WOMP algorithms and for measurements are corrupted by different levels of noise. 
\begin{figure*}[t!]
	\begin{subfigure}[t]{0.33\linewidth}
		\includegraphics[width = \textwidth]{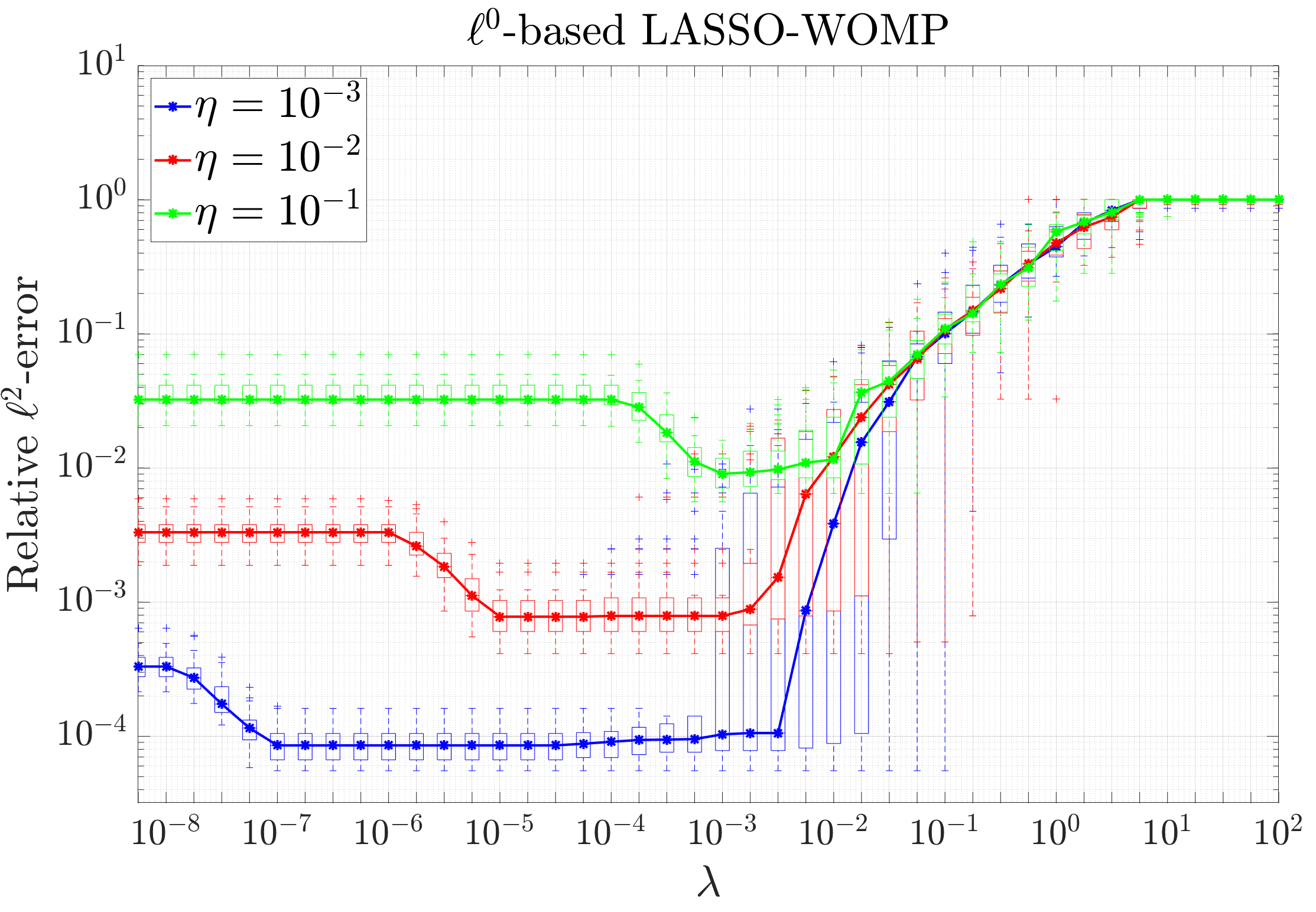}
	\end{subfigure}\hspace{0.5em}%
	\begin{subfigure}[t]{0.33\linewidth}
		\includegraphics[width = \textwidth]{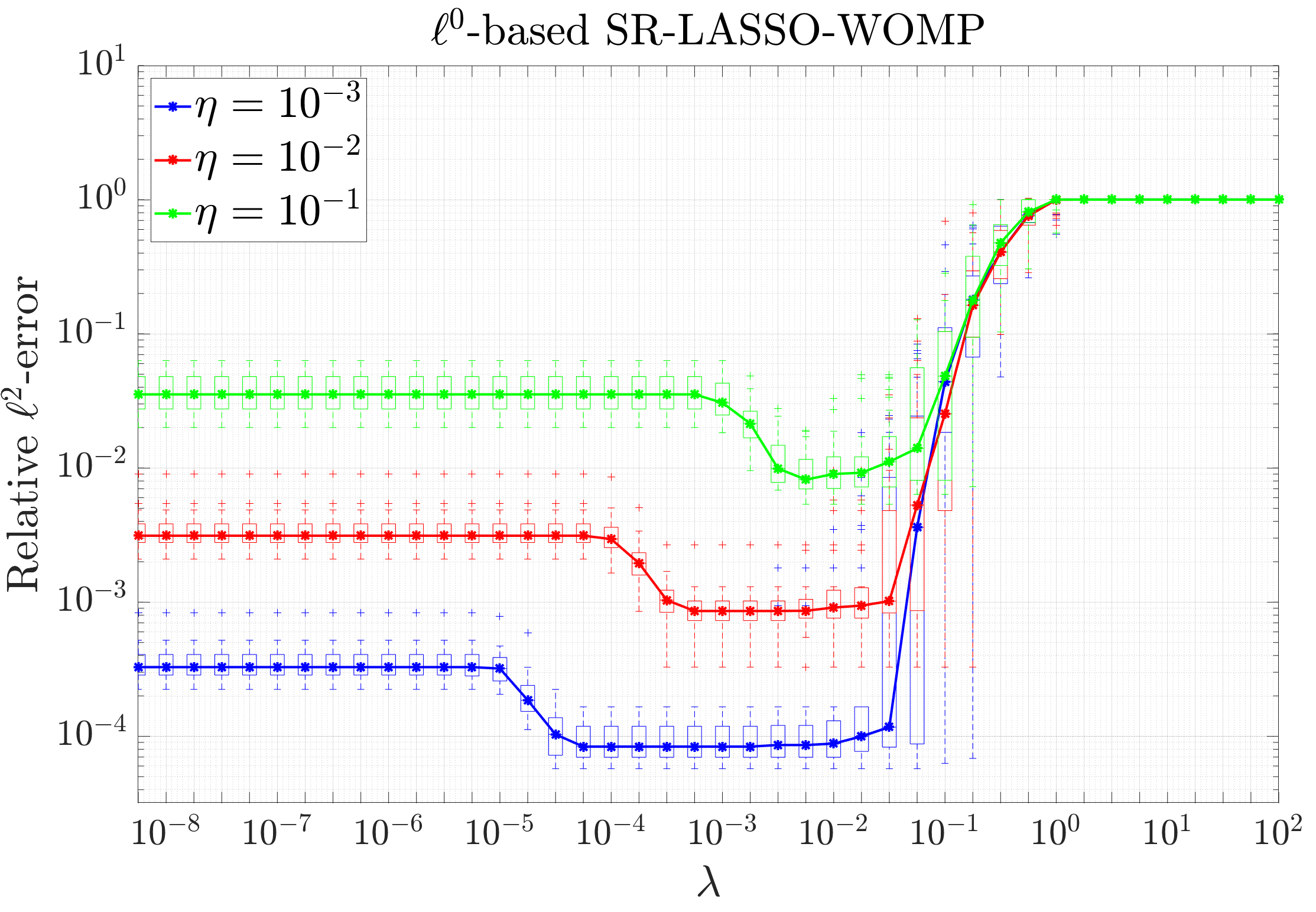}
	\end{subfigure}\hspace{0.5em}%
	\begin{subfigure}[t]{0.33\linewidth}
		\includegraphics[width = \textwidth]{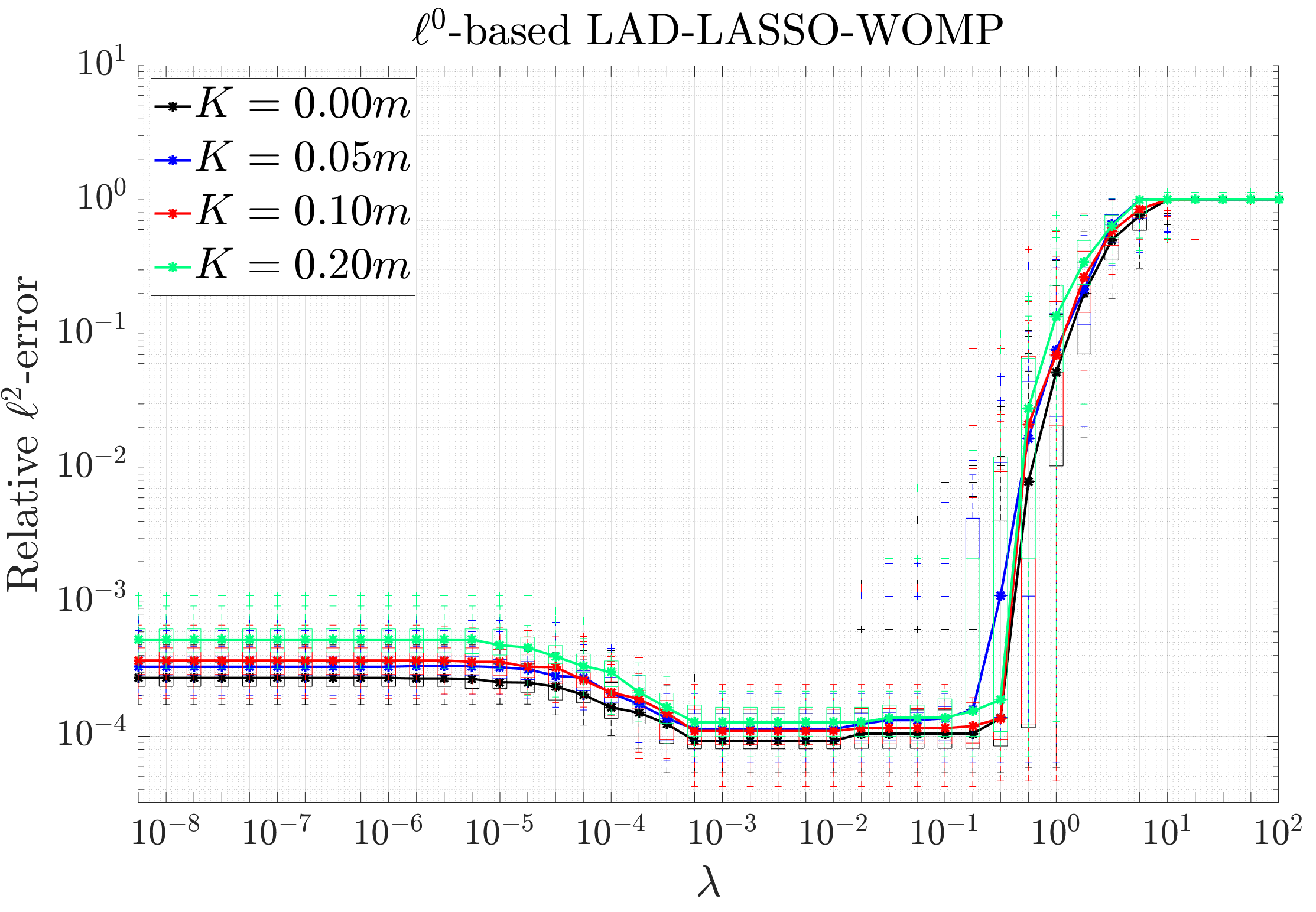}
	\end{subfigure}\hspace{0.5em}%
	
	\begin{subfigure}[t]{0.33\linewidth}
		\includegraphics[width = \textwidth]{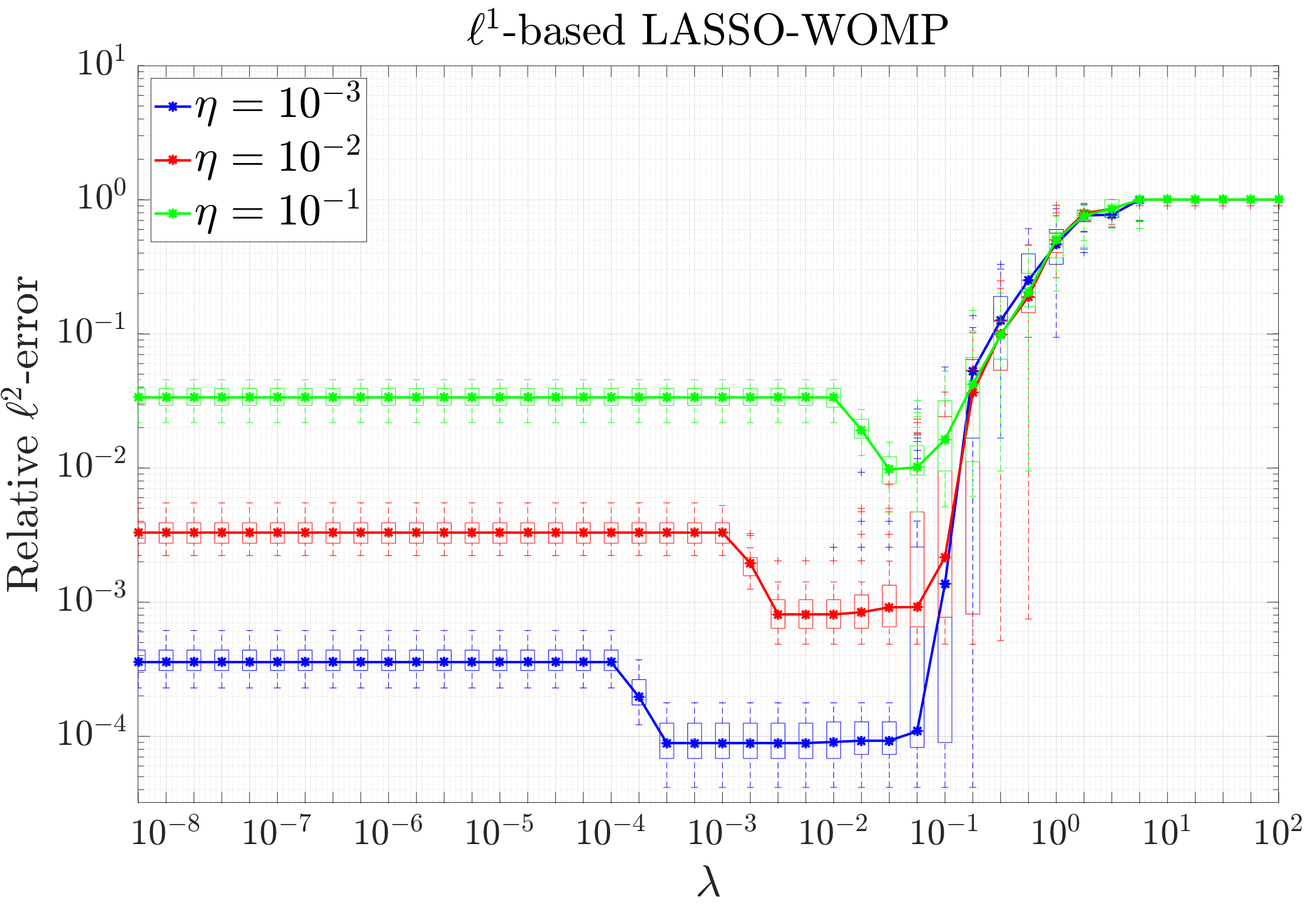}
	\end{subfigure}\hspace{0.5em}%
	\begin{subfigure}[t]{0.33\linewidth}
		\includegraphics[width = \textwidth]{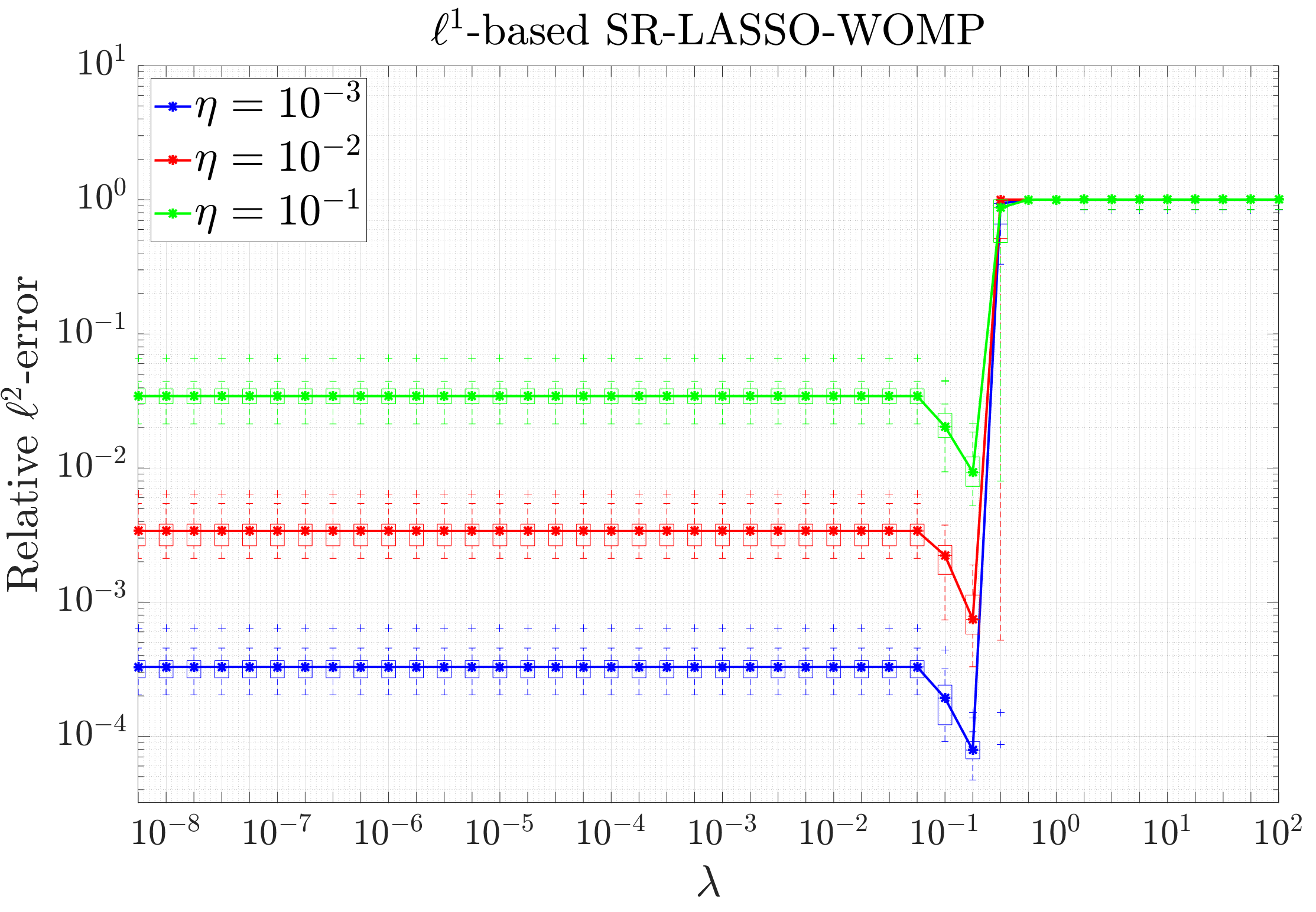}
	\end{subfigure}\hspace{0.5em}%
	\begin{subfigure}[t]{0.33\linewidth}
		\includegraphics[width = \textwidth]{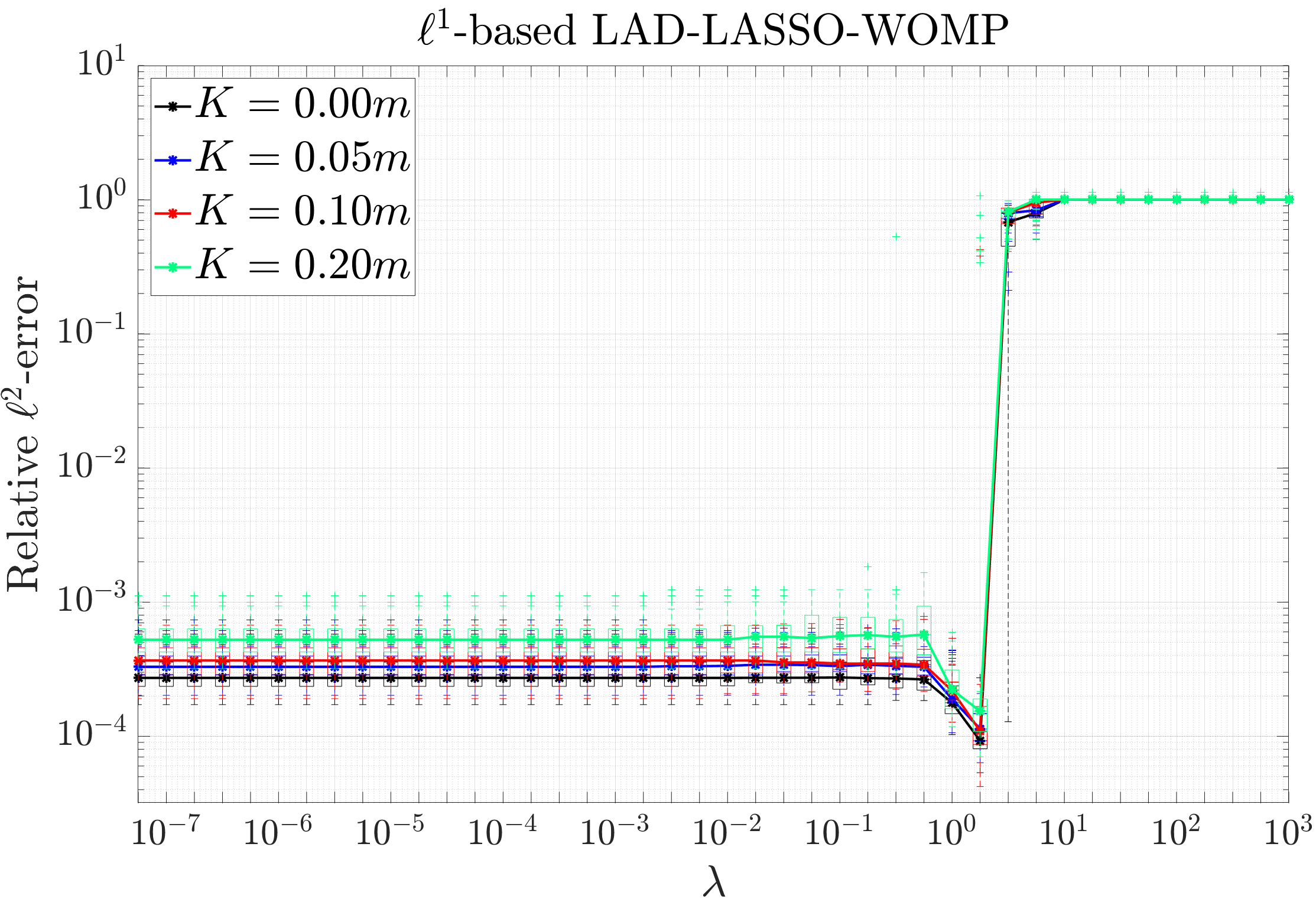}
	\end{subfigure}\hspace{0.5em}%
	\caption{Relative error as a function of the tuning parameter (Experiment I, sparse random Gaussian setting). We compare the recovery accuracy of $ \ell^0 $- and $ \ell^1 $-based WOMP algorithms for different noise or corruption levels, as in  \eqref{eq:nuermical_model}.}
	\label{fig:sparse_gaussian_error_vs_lambda}
\end{figure*}
In the LASSO and SR-LASSO WOMP cases, we let
$$ N = 300,\ m = 150,\ s = 10,\ \eta = \|e^{\mathrm{bounded}}\|_2 \in 10 \, .\hat{} \, (-3:-1),\ M = 0. $$
For LAD-LASSO WOMP, we fix
$$ N = 300,\ m = 150,\ s = 10,\ \eta = 10^{-3},\ M = 100,\ K \in \{0, 0.05m, 0.1m, 0.2 m\}. 
$$
Both the $ \ell^0 $- and $ \ell^1 $-based algorithms are able to reach a relative $\ell^2$-error below the noise level for appropriate choices of the tuning parameter $\lambda$. We note that every experiment has optimal values of $ \lambda $ for which the recovery error associated with a certain noise level is minimized. These optimal values are independent of the noise level for $\ell^1$-based SR-LASSO and on the corruption level for both $\ell^0$- and $\ell^1$-based LAD-LASSO WOMP. An analogous phenomenon can be observed for the corresponding $\ell^1$ minimization programs \cite{adcock2019correcting}. Finally, it is worth noting that the optimal values of $\lambda$ depend on the noise level for the $\ell^0$-based SR-LASSO formulation. 

\paragraph{Experiment II (function approximation).} Next we consider the high-dimensional function approximation setting. We approximate the high-dimensional function defined in \eqref{eq:func_approx_isoexp} with $d = 5$, where
$$ N = |\Lambda| = 426,\ m = 200, $$
and $ M $, $ \eta $ and $ K $ as before. Specifically, the truncation set $\Lambda=\Lambda_n^{\text{HC}}$ is a hyperbolic cross of order $n \in \mathbb{N}_0$, defined as
\begin{equation}
    \label{eq:hyperbolic_cross_set}
    \Lambda_n^{\text{HC}} := \left\{\nu = (\nu_k)_{k = 1}^d \in \bN_0^d: \prod_{k = 1}^d(\nu_k + 1) \leq n + 1\right\},
\end{equation}
In this experiment, we let $ n = 18 $. Note that in the function approximation setting, even when $\eta = 0$, samples are intrinsically corrupted by noise. This is due to the truncation error introduced by $\Lambda$ (see \cite[Chapter~7]{adcock2022sparse}). Moreover, we recall that in this experiment we use weights $ w \in \bR^N $ defined as in  \eqref{eq:func_approx_intrinsic_weights}. 

Fig.~\ref{fig:func_approx_error_vs_lambda} shows the results of this experiment. 
\begin{figure*}[t!]
	\begin{subfigure}[t]{0.33\linewidth}
		\includegraphics[width = \textwidth]{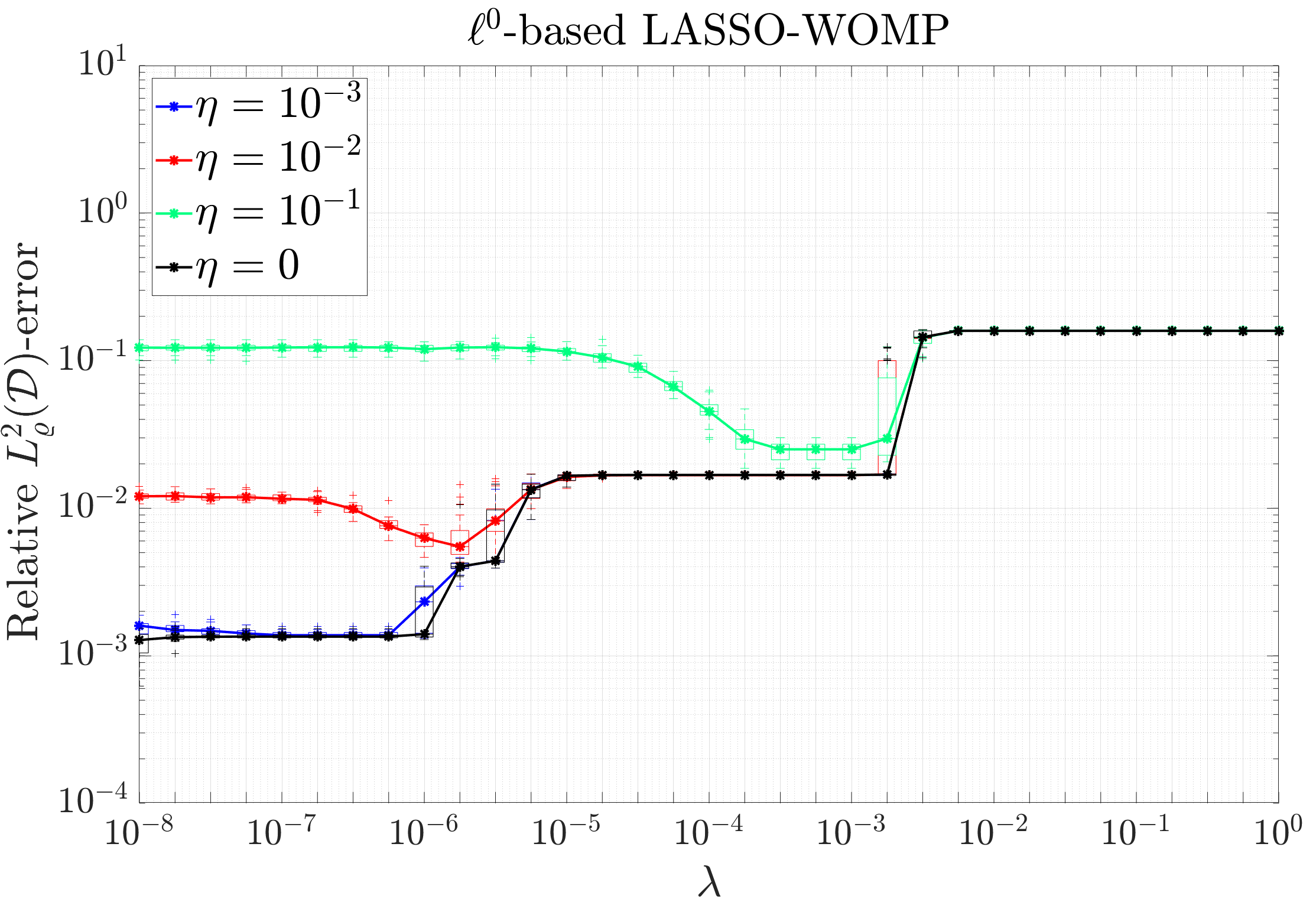}
	\end{subfigure}\hspace{0.5em}%
	\begin{subfigure}[t]{0.33\linewidth}
		\includegraphics[width = \textwidth]{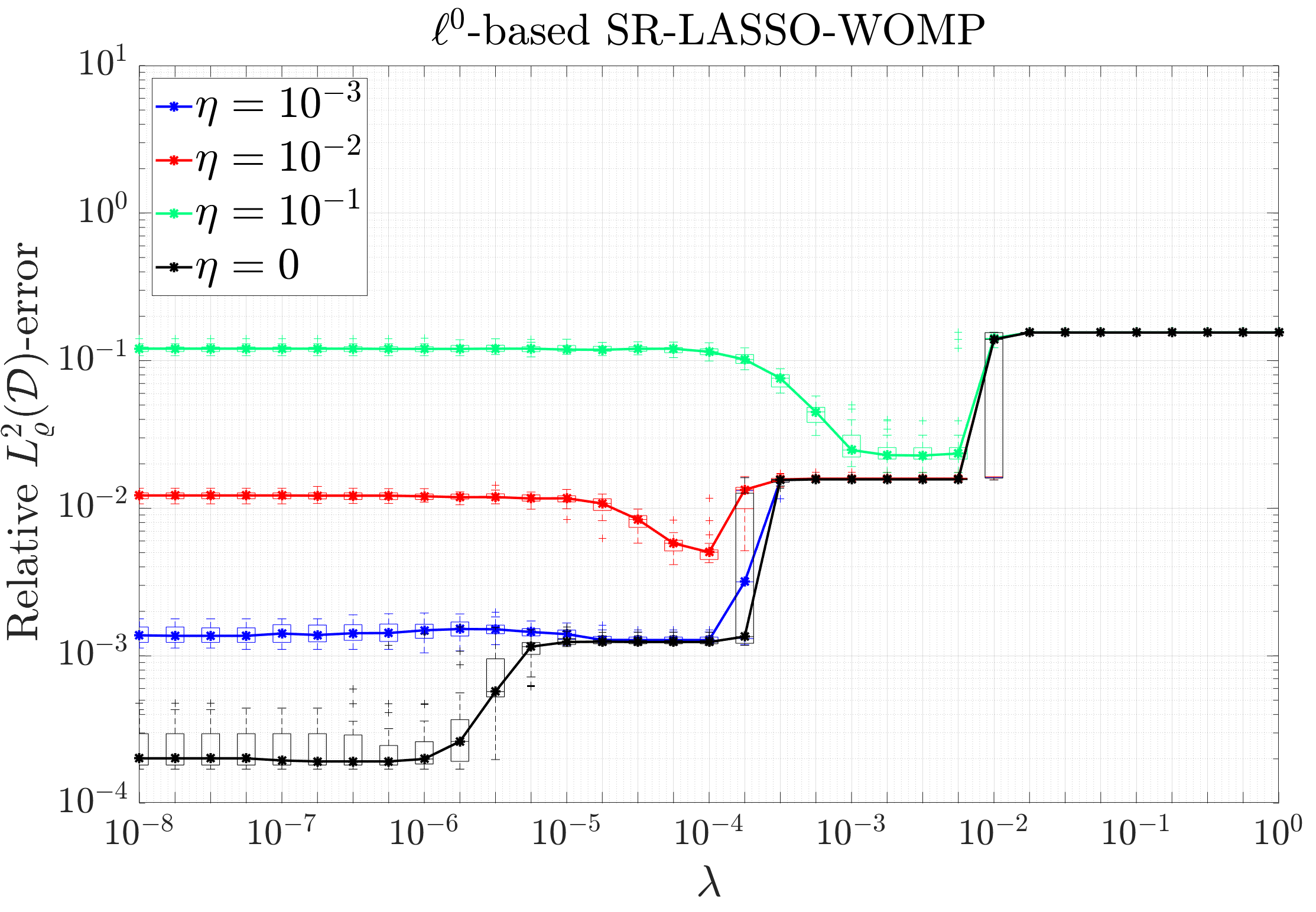}
	\end{subfigure}\hspace{0.5em}%
	\begin{subfigure}[t]{0.33\linewidth}
		\includegraphics[width = \textwidth]{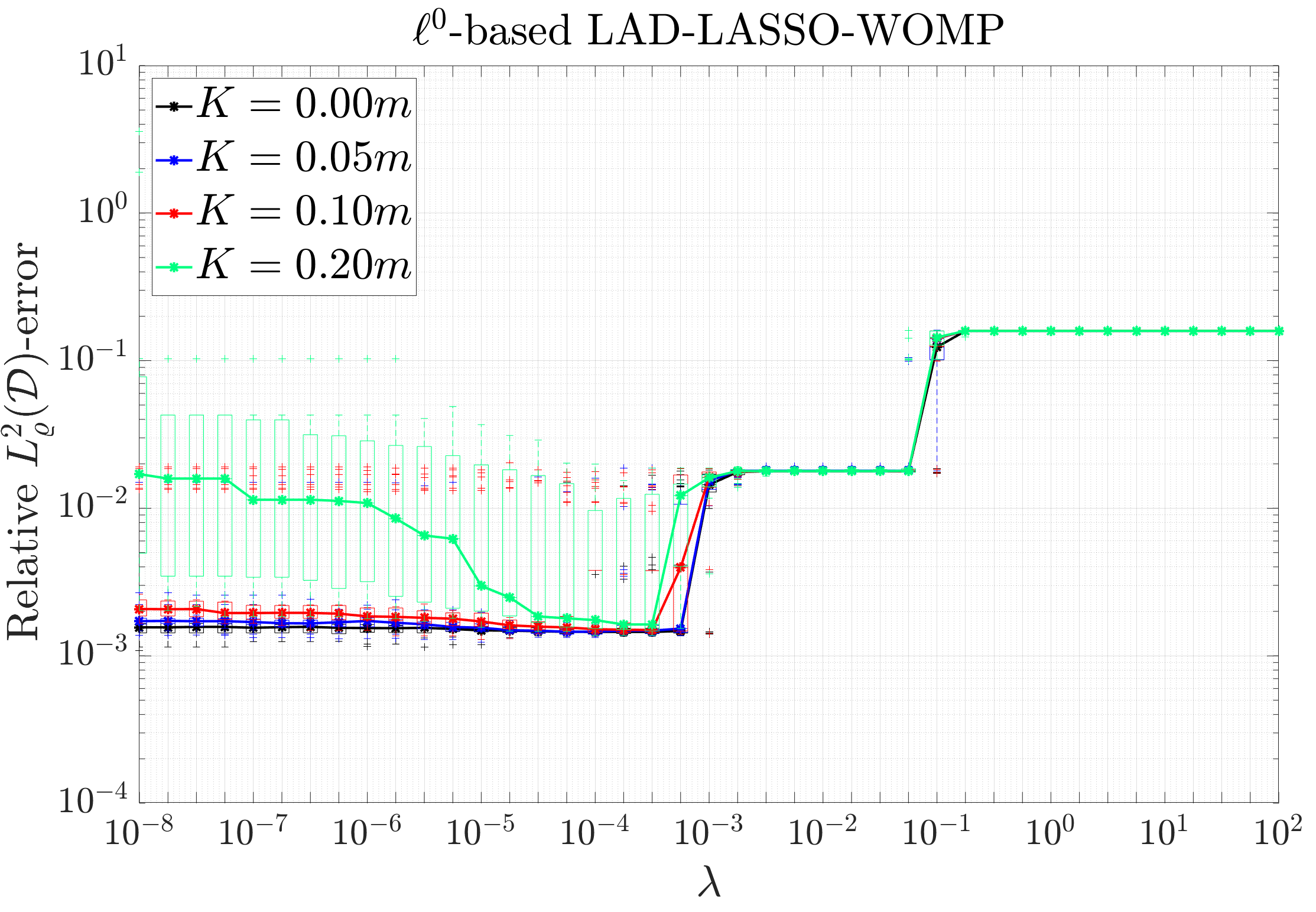}
	\end{subfigure}\hspace{0.5em}%
	\begin{subfigure}[t]{0.33\linewidth}
		\includegraphics[width = \textwidth]{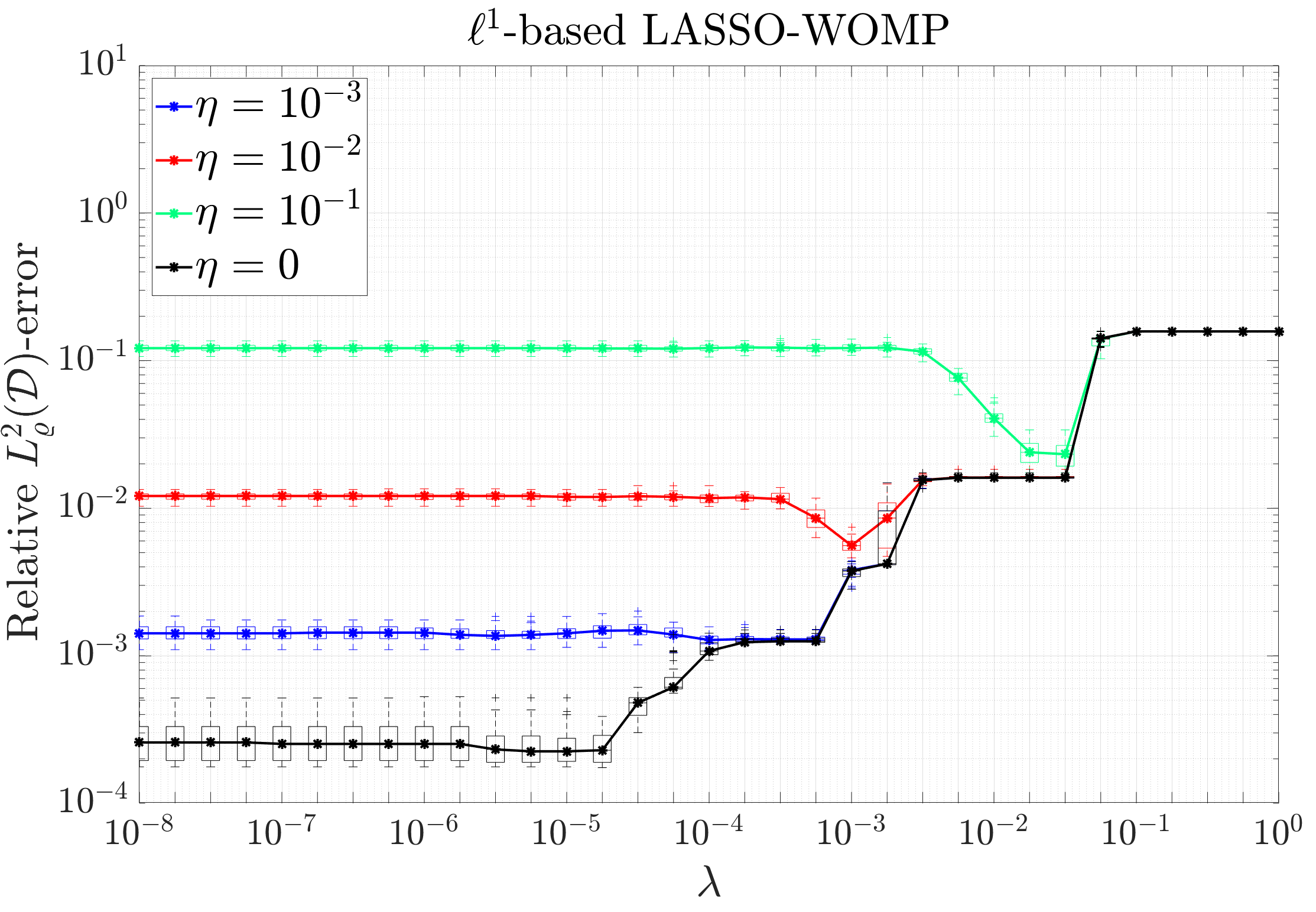}
	\end{subfigure}\hspace{0.5em}%
	\begin{subfigure}[t]{0.33\linewidth}
		\includegraphics[width = \textwidth]{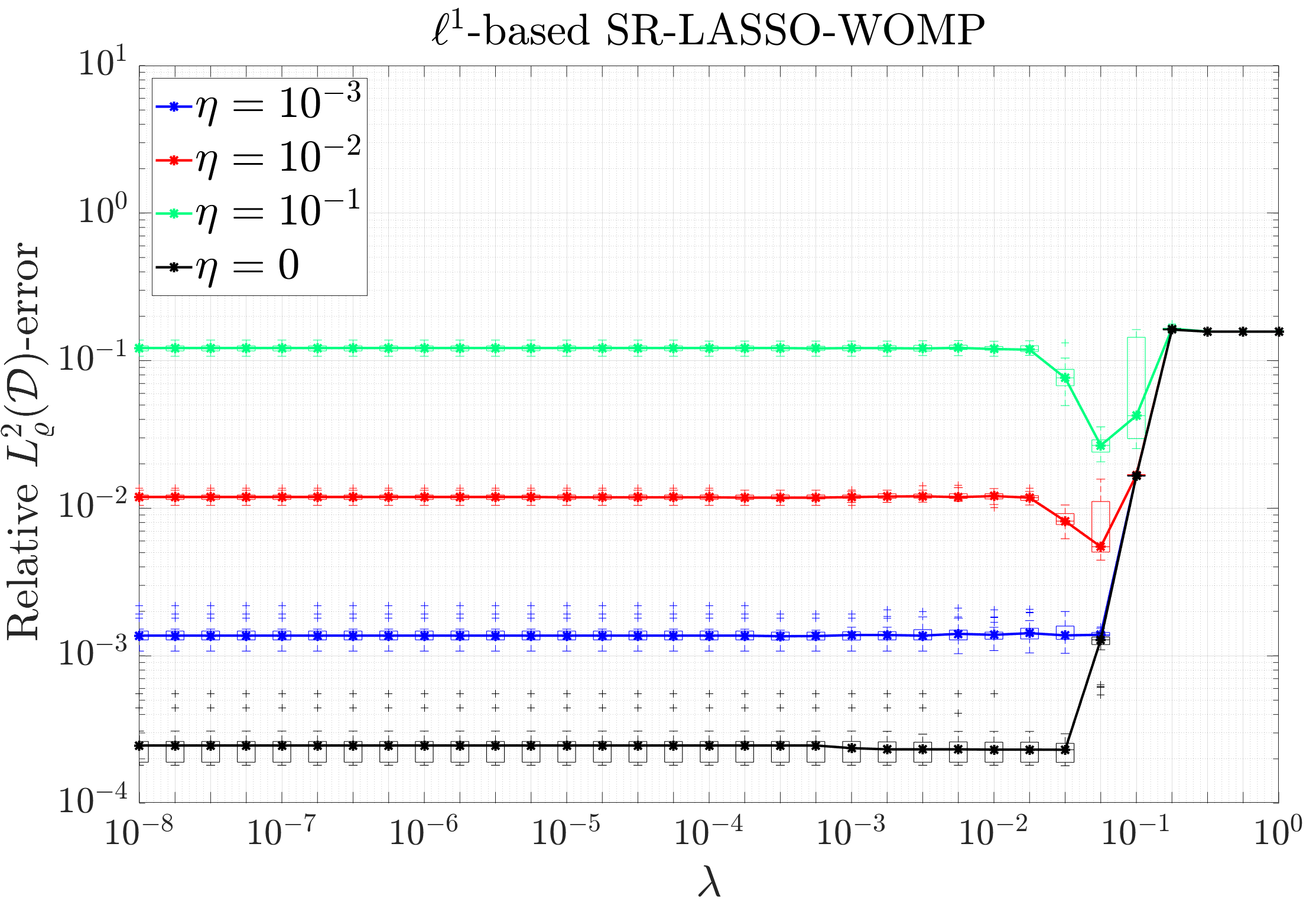}
	\end{subfigure}\hspace{0.5em}%
	\begin{subfigure}[t]{0.33\linewidth}
		\includegraphics[width = \textwidth]{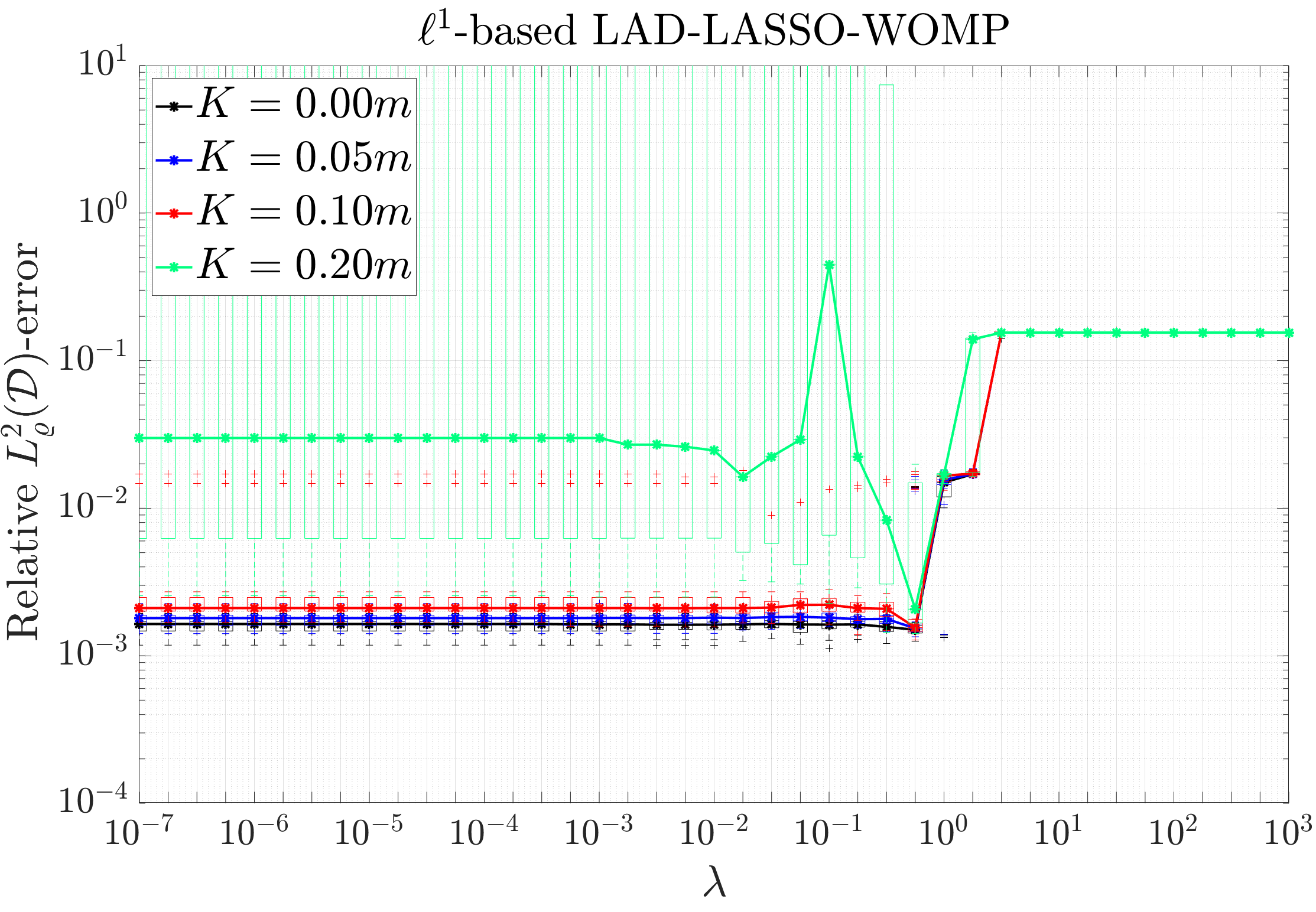}
	\end{subfigure}\hspace{0.5em}%
	\caption{Relative error as a function of the tuning parameter (Experiment II, function approximation). We compare the recovery accuracy of $ \ell^0 $- and $ \ell^1 $-based WOMP algorithms for different noise or corruption levels, as in \eqref{eq:nuermical_model}.}
	\label{fig:func_approx_error_vs_lambda}
\end{figure*}
Note that in this setting the relative $L^2_{\varrho}(D)$-error and the relative $\ell^2$-error coincide because of orthonormality of the polynomial basis $\{\Psi_\nu\}_{\nu\in\mathbb{N}_0^d}$.
Observations analogous to those made in Experiment I hold in this case as well, with some differences. First, Fig.~\ref{fig:func_approx_error_vs_lambda} shows even more clearly than Fig.~\ref{fig:sparse_gaussian_error_vs_lambda} the superiority of the $ \ell^1 $-based SR-LASSO approach with respect to its $ \ell^0 $-based counterpart. From it, we can see that only for $ \ell^1 $-based SR-LASSO WOMP the optimal values of $ \lambda $ are vertically aligned and thus independent of the noise level. Second, when the corruption level is large ($K = 0.2m$), $ \ell^0 $-based LAD-LASSO WOMP is more robust to the choice of $\lambda$ than its $\ell^1$-based counterpart. 

\paragraph{Experiment III (sparse random Gaussian setting with oracle).} In the final experiment of this section we consider the  sparse random Gaussian setting with oracle, and we illustrate the benefits provided by weights in for signal recovery via WOMP. We employ the same parameter settings as Experiment I, with the difference that this time $ N = 500 $ and we do not consider $\ell^0$-based WOMP variants, we fix the noise level, and test different choices of weights.
We set the noise level to $ \eta = 10^{-3} $ for LASSO and SR-LASSO, and corruptions with $ K = 0.1m $ for LAD-LASSO. As mentioned earlier, the prior knowledge from $ S^\mathrm{oracle} $ is incorporated into the weight vector $ w \in \bR^N $. Here we assume the oracle to have \emph{a priori} knowledge of just half of the support of $x$. In order to create $ S^\mathrm{oracle} $, half of the support entries are randomly chosen and are used to generate the weight vector $ w \in \bR^N $ as in \eqref{eq:oracle_weights} with $w_0 =10^{-3}$.

The results of this experiments are shown in Fig.~\ref{fig:oracle_error_vs_lambda}. 
\begin{figure*}[t!]
	\begin{subfigure}[t]{0.33\linewidth}
		\includegraphics[width = \textwidth]{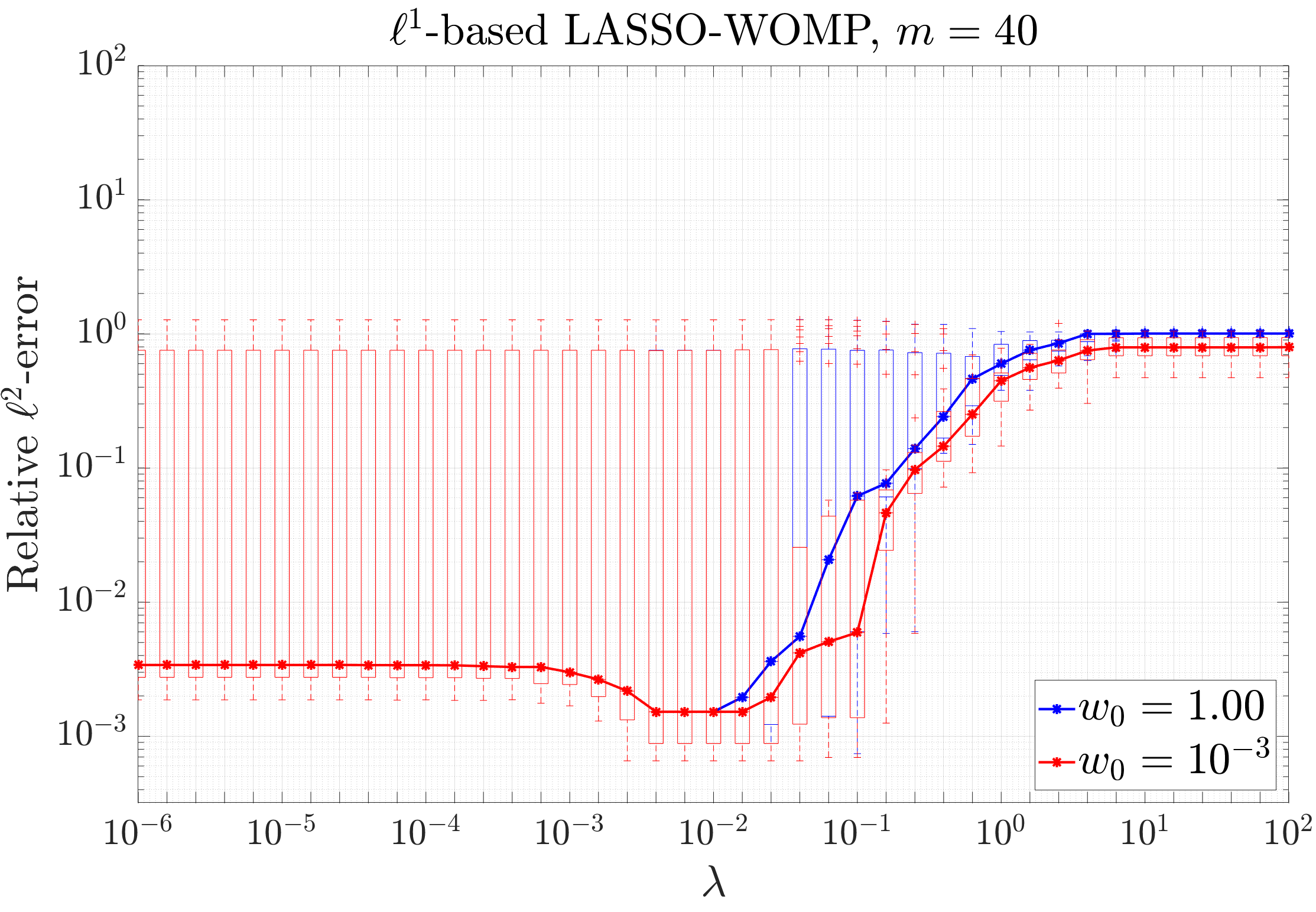}
	\end{subfigure}\hspace{0.5em}%
	\begin{subfigure}[t]{0.33\linewidth}
		\includegraphics[width = \textwidth]{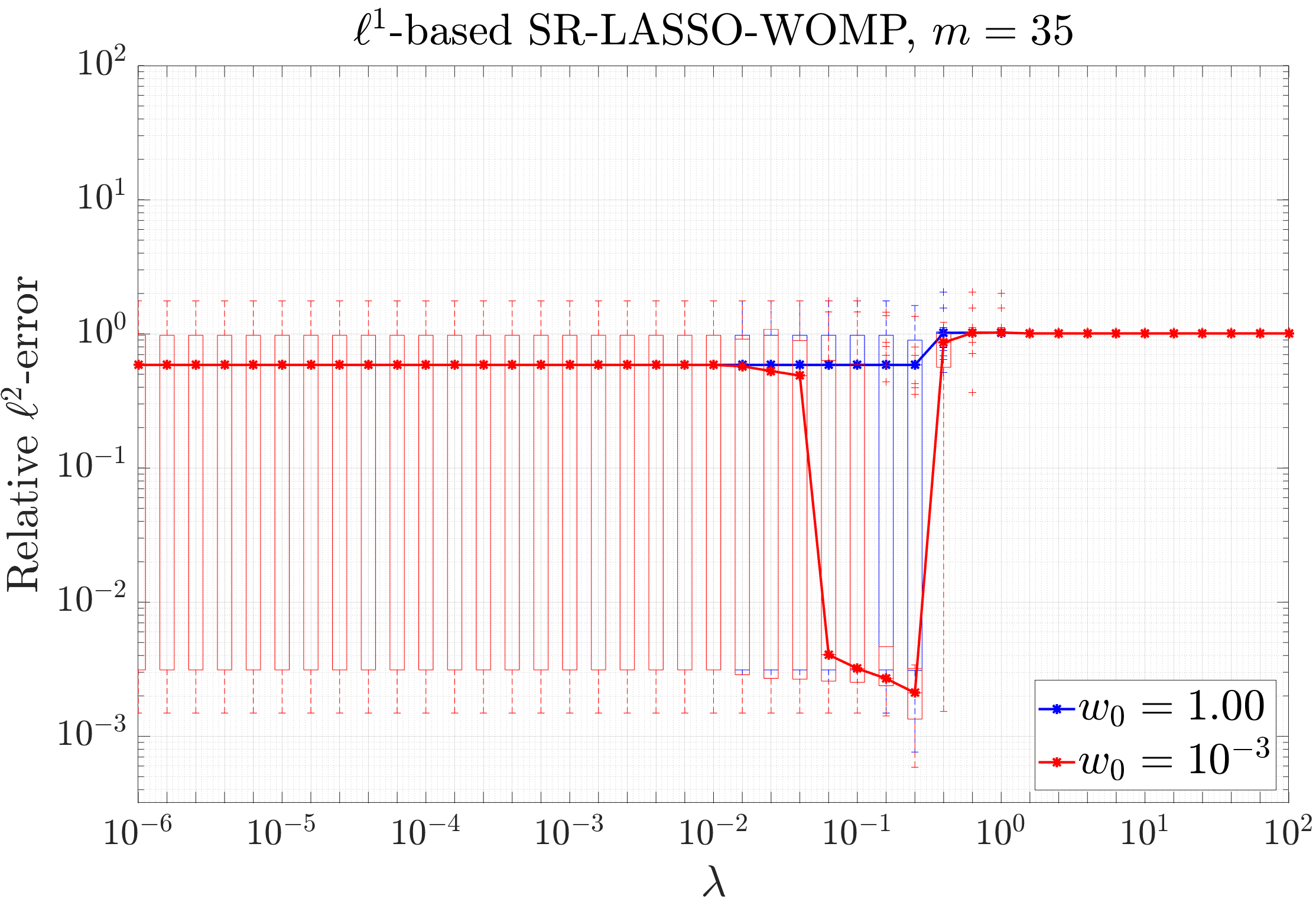}
	\end{subfigure}\hspace{0.5em}%
	\begin{subfigure}[t]{0.33\linewidth}
		\includegraphics[width = \textwidth]{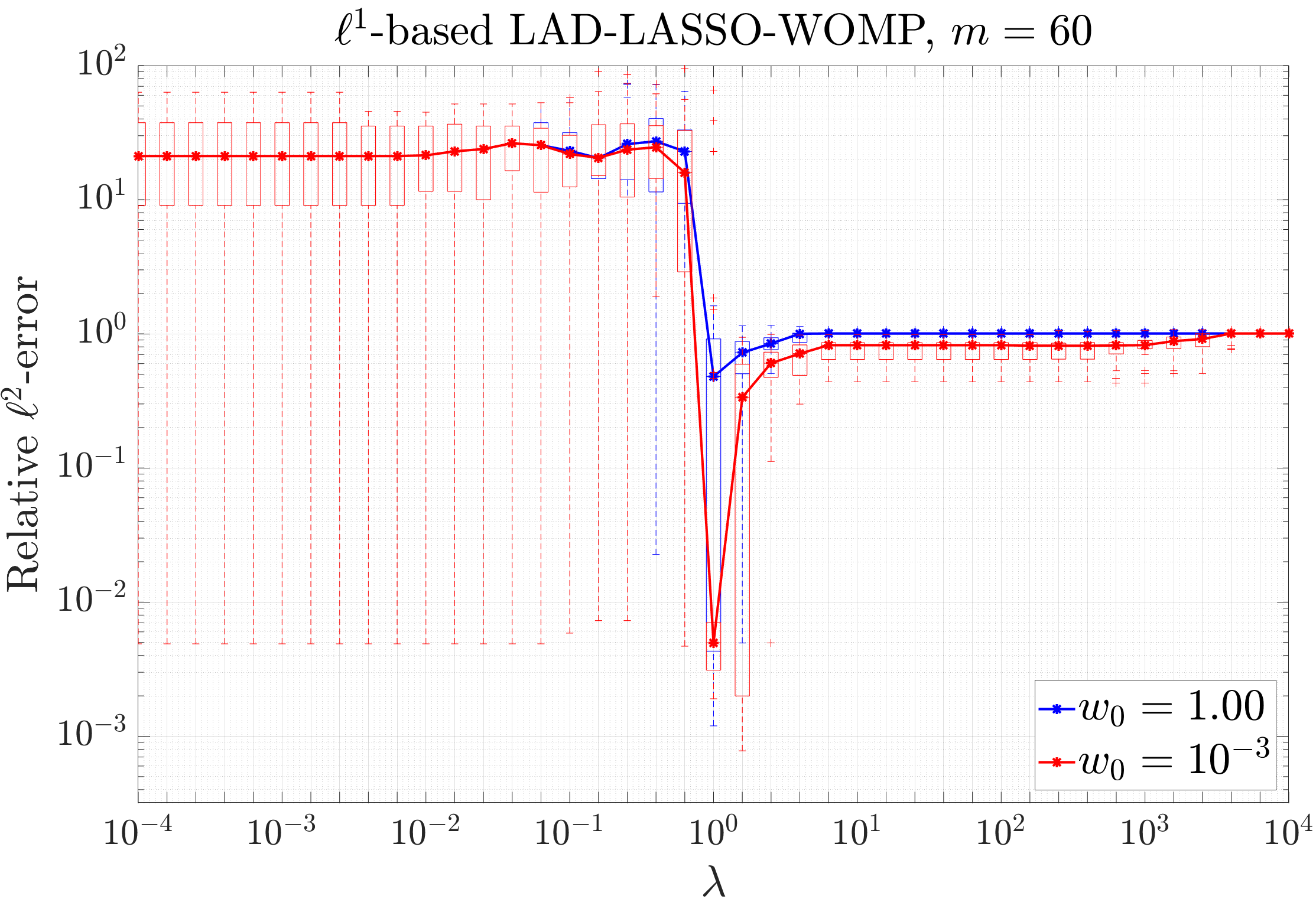}
	\end{subfigure}\hspace{0.5em}%
	
	\begin{subfigure}[t]{0.33\linewidth}
		\includegraphics[width = \textwidth]{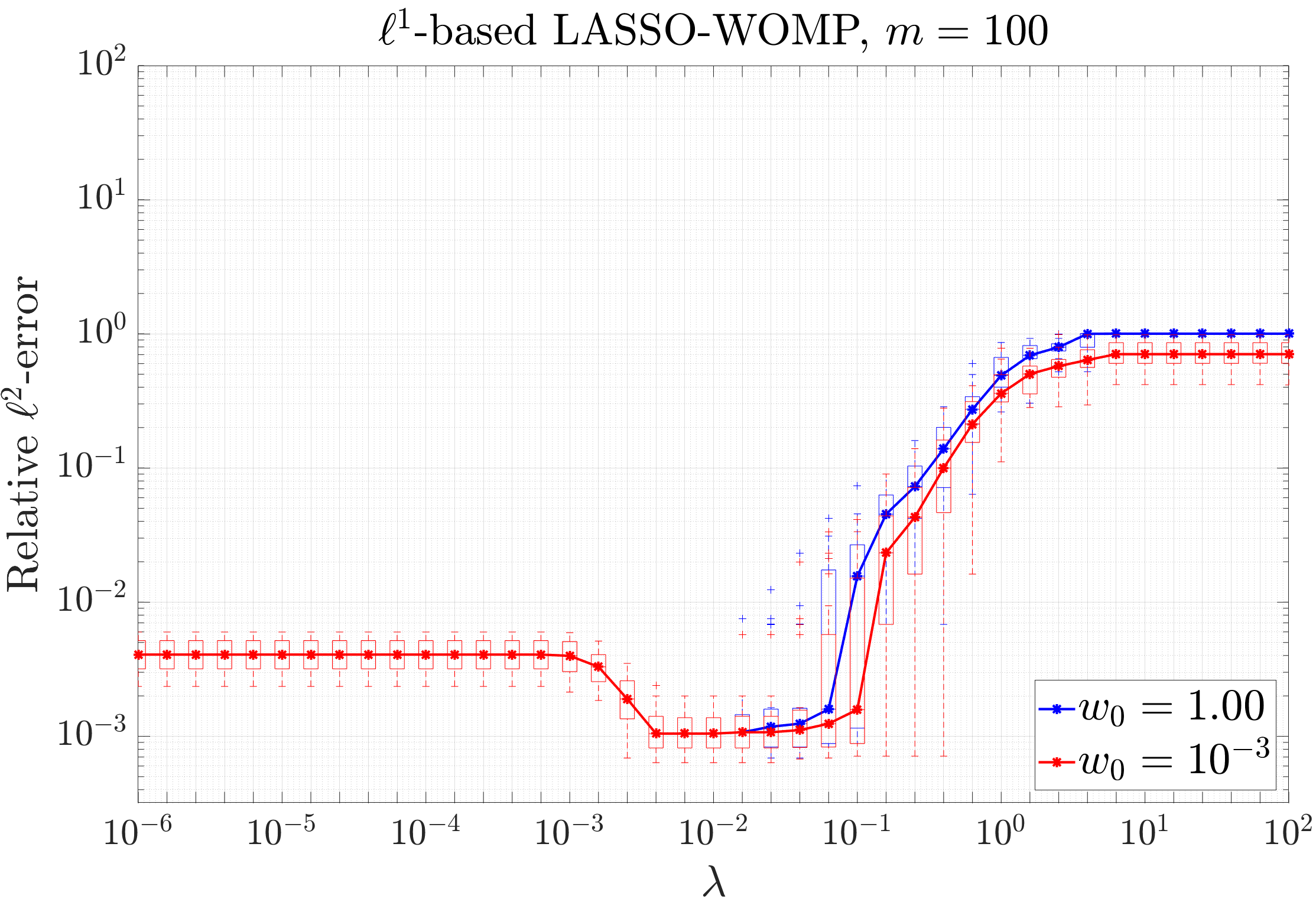}
	\end{subfigure}\hspace{0.5em}%
	\begin{subfigure}[t]{0.33\linewidth}
		\includegraphics[width = \textwidth]{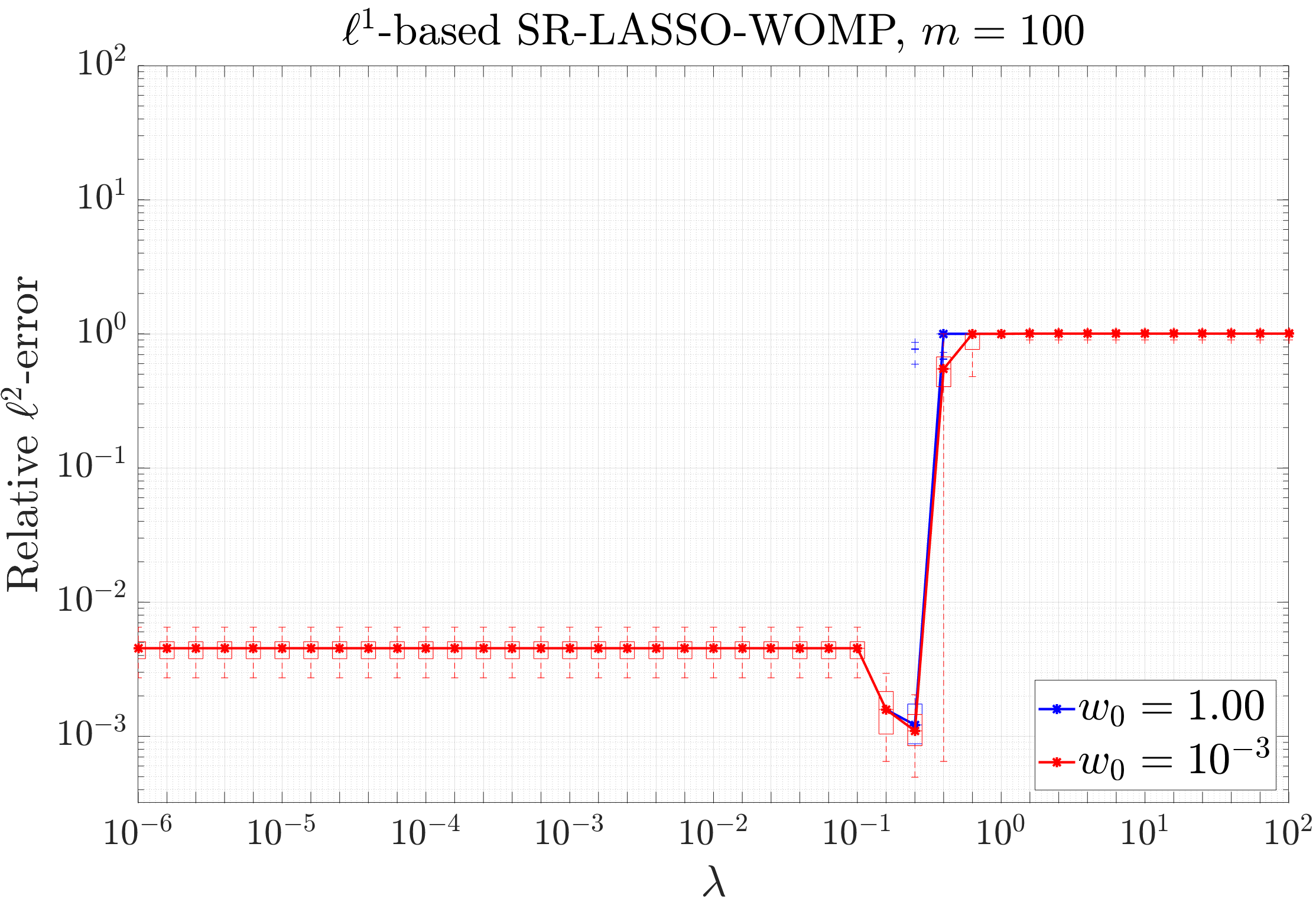}
	\end{subfigure}\hspace{0.5em}%
	\begin{subfigure}[t]{0.33\linewidth}
		\includegraphics[width = \textwidth]{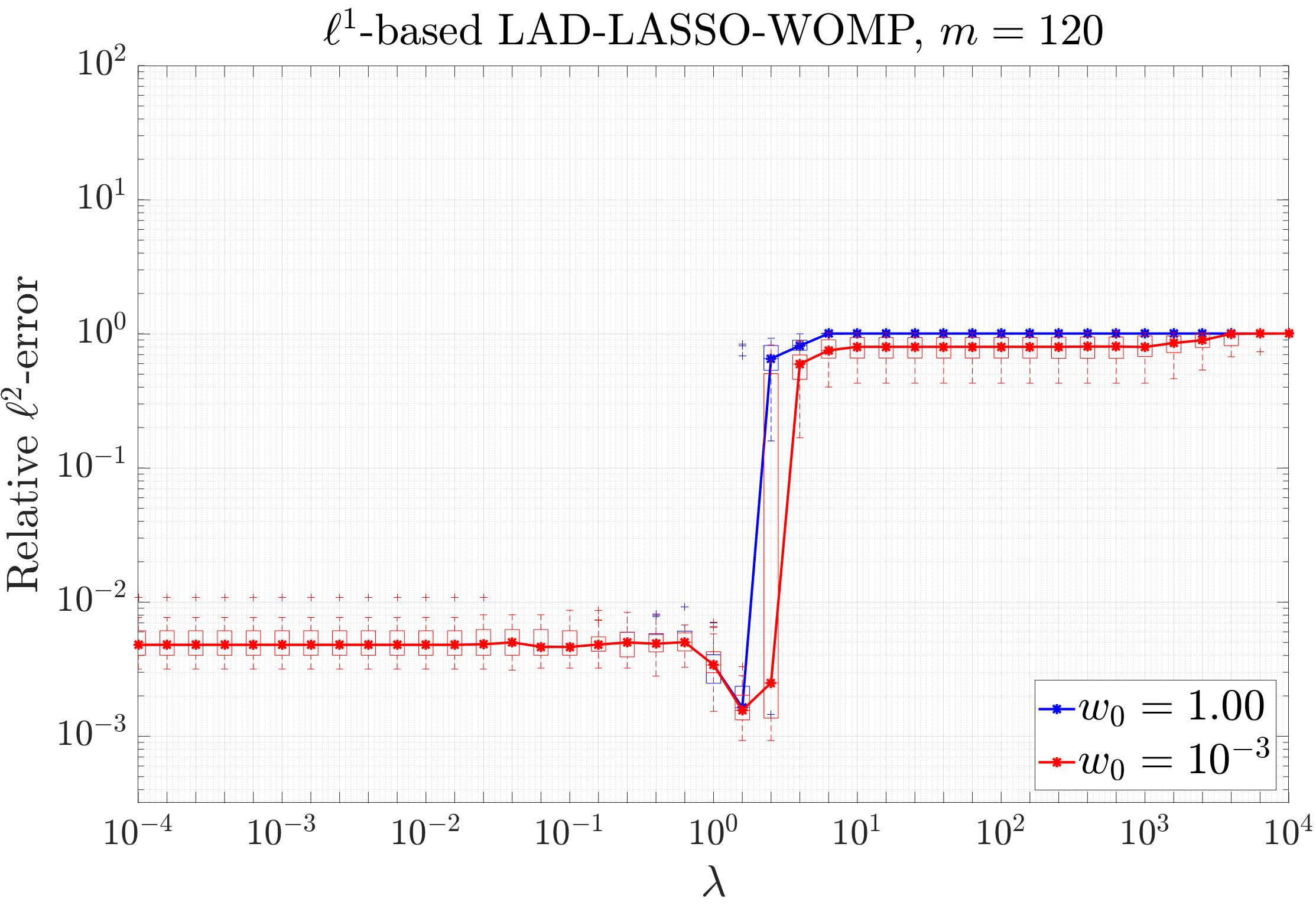}
	\end{subfigure}\hspace{0.5em}%
	\caption{Relative error as a function of the tuning parameter (Experiment III, sparse random Gaussian setting with oracle). Different $ \ell^1 $-based WOMP algorithms are tested for a fixed noise level, different choices of weights depending on the parameter $w_0$ (see \eqref{eq:oracle_weights}), and for low (top row) and high (bottom row) values of $ m $.}
	\label{fig:oracle_error_vs_lambda}
\end{figure*}
Recovery is performed for different numbers of measurements, namely, $m = 40, 100$ for LASSO, $ m = 35, 100 $ for SR-LASSO WOMP and $m=60, 120$ for LAD-LASSO WOMP. We observe that weights are able to improve reconstruction in all settings. This phenomenon has been previously known in the literature (see, e.g., \cite{friedlander2011recovering,bah2016sample}), and in this experiment is particularly evident in the SR-LASSO and LAD-LASSO cases (second and third column in Fig.~\ref{fig:oracle_error_vs_lambda}).

\subsection{Recovery error versus iteration number}
\label{s:numerics_iterations}
In the last two experiments (IV and V), we study the recovery error as a function of the number of iterations of the proposed greedy algorithms. This will highlight the benefits due to the presence of a regularization term in the loss function. We compute the relative $\ell^2$-error at iteration $ k $ and for a specific value of $ \lambda $ as
$$ E^{(k)}_{\lambda} = \frac{\|\hat{x}^{(k)} - x\|_2}{\|x\|_2},
\quad k \in [N_{\mathrm{iter}}],\ \lambda \in \mathcal{L}, $$
where $N_{\mathrm{iter}}$ is the maximum number of iterations and $\mathcal{L}$ is a suitable set of tuning parameters. We repeat the above process for $ N_{\mathrm{trial}} $ random trials. The setup for Experiments IV and V is detailed below.
\paragraph{Experiment IV (sparse random Gaussian setting).} For LASSO and SR-LASSO WOMP, we fix
    $$ N = 200,\ m = 100,\ s = 15,\ N_{\mathrm{iter}} = 150,\ \eta = 10^{-3},\ M = 0. $$
    For LAD-LASSO WOMP, we let
    $$ N = 200,\ m = 100,\ s = 15,\ N_{\mathrm{iter}} = 150,\ \eta = 10^{-3},\ M = 100,\ K = 0.05m. $$
\paragraph{Experiment V (function approximation).}  We choose
    $$ d = 10,\ n = 8,\ N = |\Lambda| = 471,\ m = 200,\ N_{\mathrm{iter}} = 250 $$
    where $ n $ is the order of hyperbolic cross set defined in \eqref{eq:hyperbolic_cross_set}, and $ M $, $ \eta $ and $ K $ as in Experiment III.\\

Figs.~\ref{fig:error_vs_iteration_sparse_gaussian} and \ref{fig:error_vs_iteration_func_approx} show the relative recovery $\ell^2$-error of $ \ell^1 $-based WOMP algorithm for Experiments IV and V, respectively. 
\begin{figure*}[t!]
	\begin{subfigure}[t]{0.33\linewidth}
		\includegraphics[width = \textwidth]{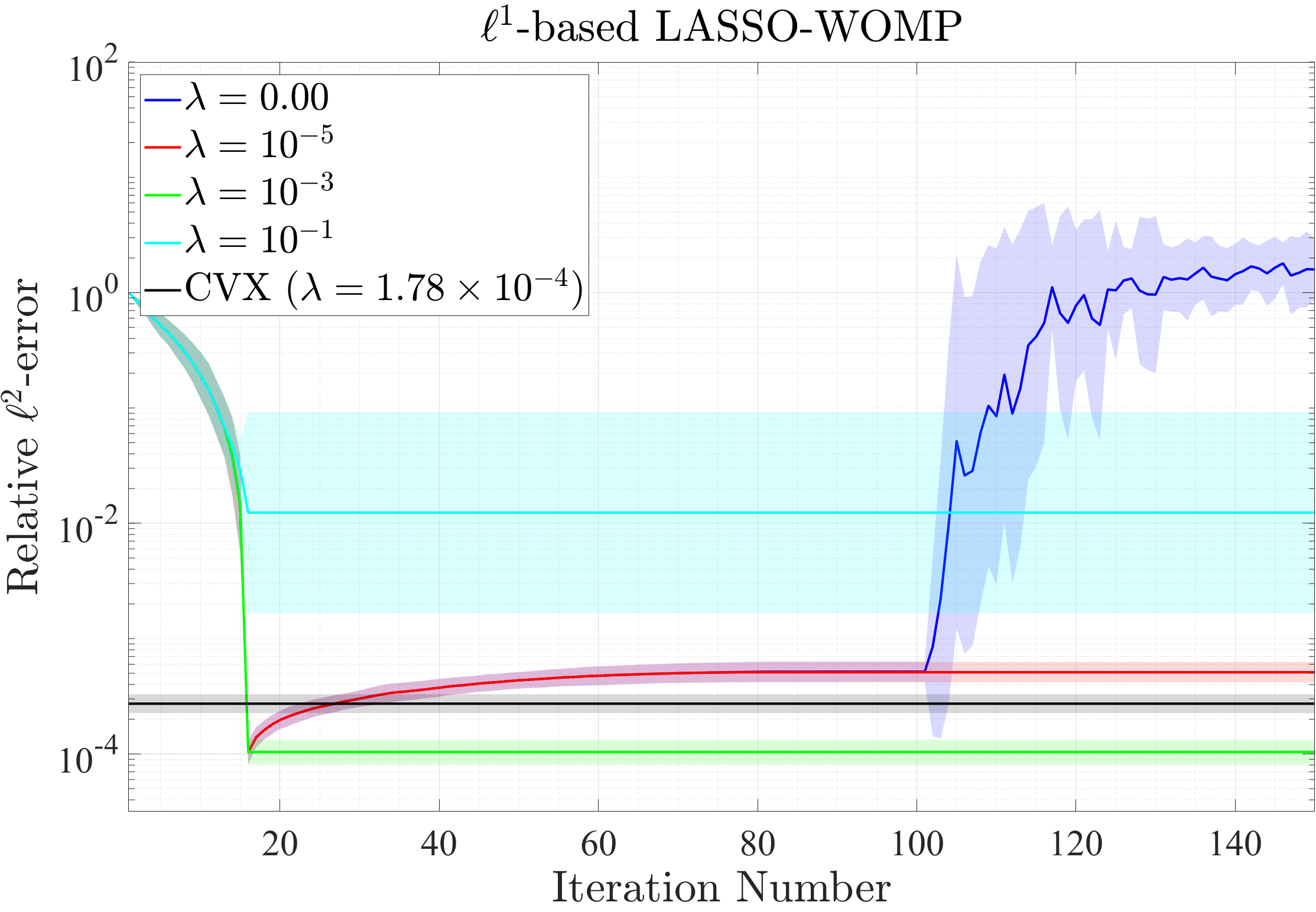}
	\end{subfigure}\hspace{0.5em}%
	\begin{subfigure}[t]{0.33\linewidth}
		\includegraphics[width = \textwidth]{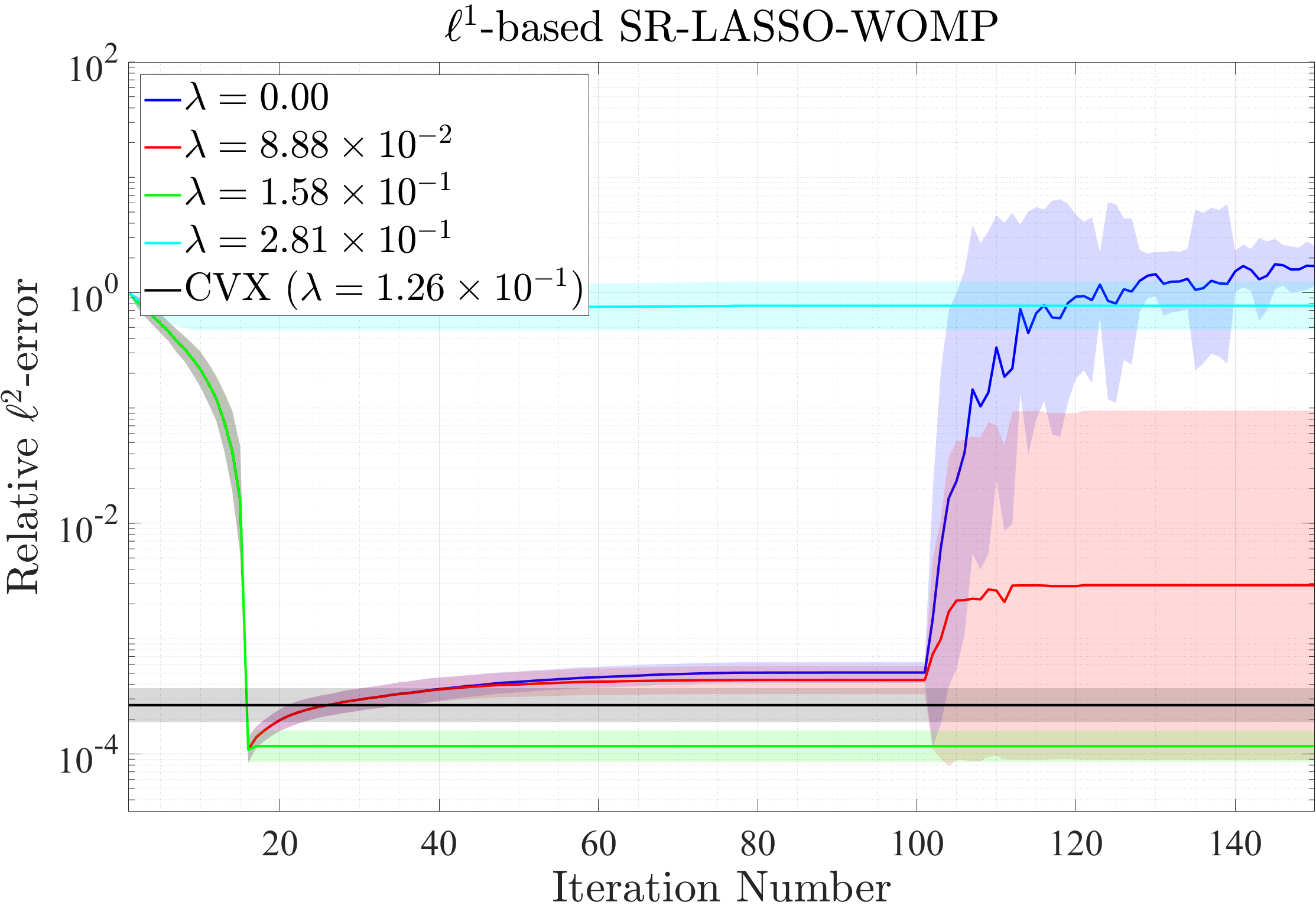}
	\end{subfigure}\hspace{0.5em}%
	\begin{subfigure}[t]{0.33\linewidth}
		\includegraphics[width = \textwidth]{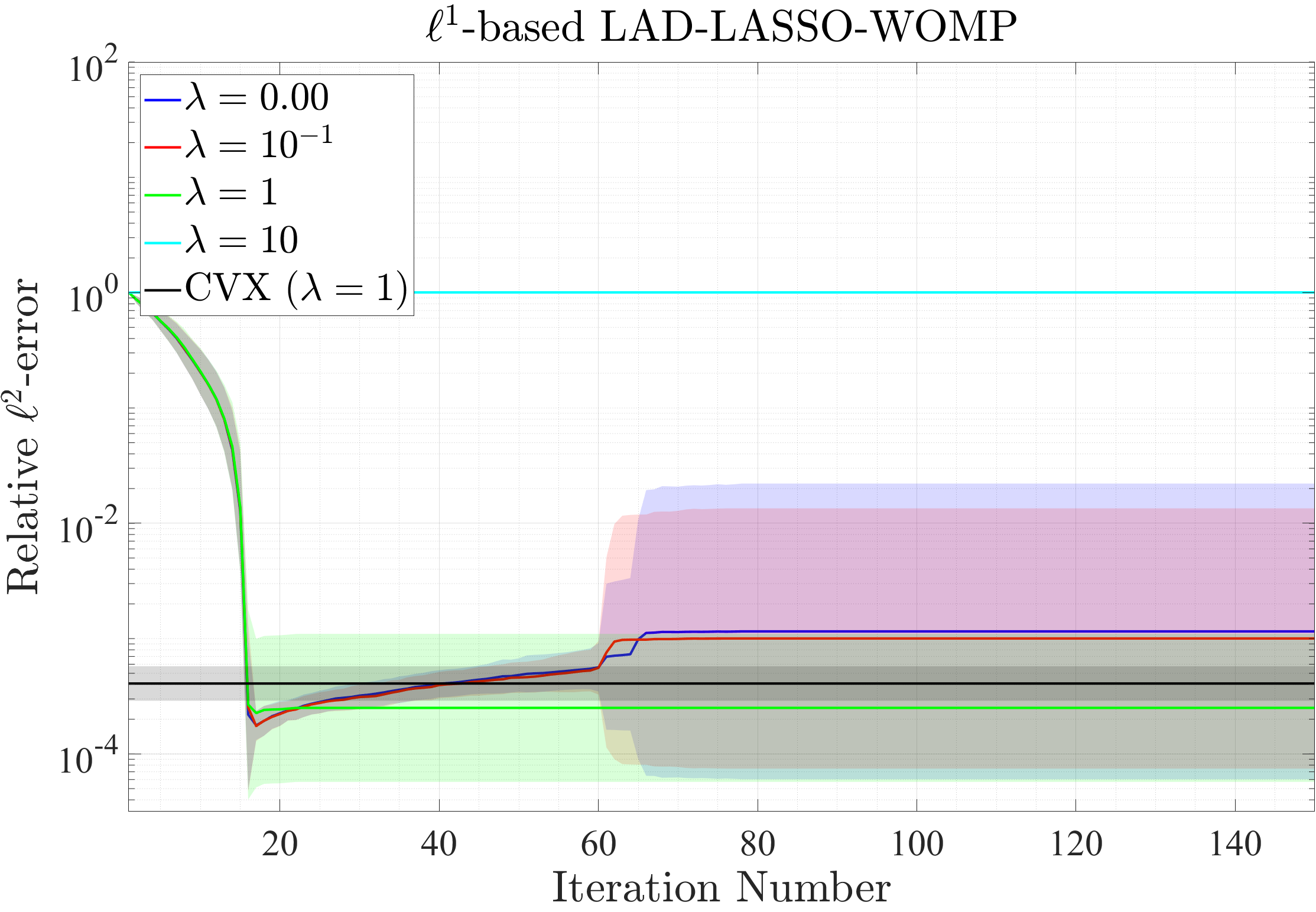}
	\end{subfigure}
	\caption{Relative error as a function of the iteration number (Experiment IV, sparse random Gaussian setting). The proposed $\ell^1$-based WOMP formulations are tested for different values of the tuning parameter $\lambda$. The black curve corresponds to recovery via convex optimization of the corresponding loss function.}
	\label{fig:error_vs_iteration_sparse_gaussian}
\end{figure*}
For better visualization, we use \emph{shaded plots}. 
\begin{figure*}[t!]
	\begin{subfigure}[t]{0.33\linewidth}
		\includegraphics[width = \textwidth]{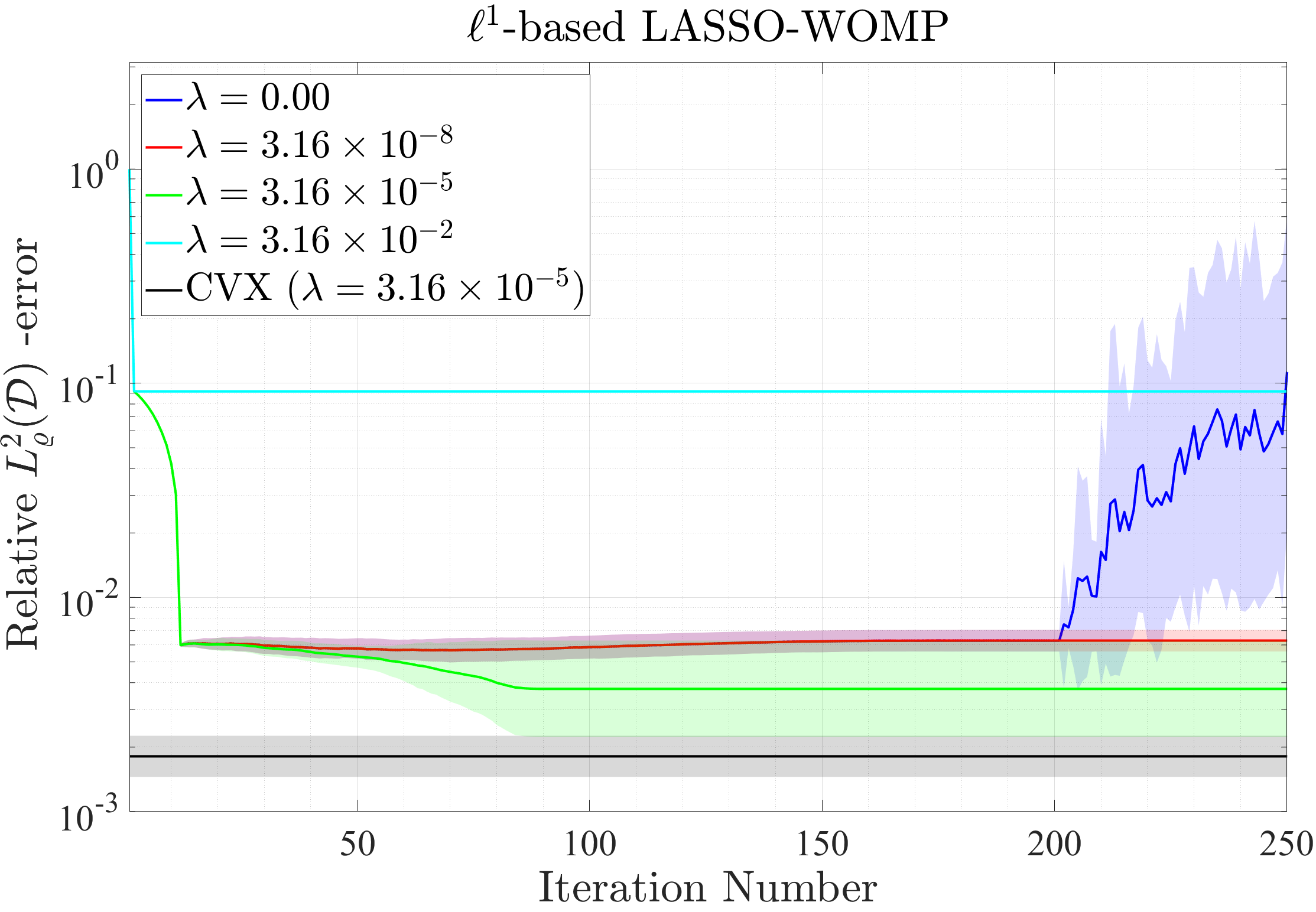}
	\end{subfigure}\hspace{0.5em}%
	\begin{subfigure}[t]{0.33\linewidth}
		\includegraphics[width = \textwidth]{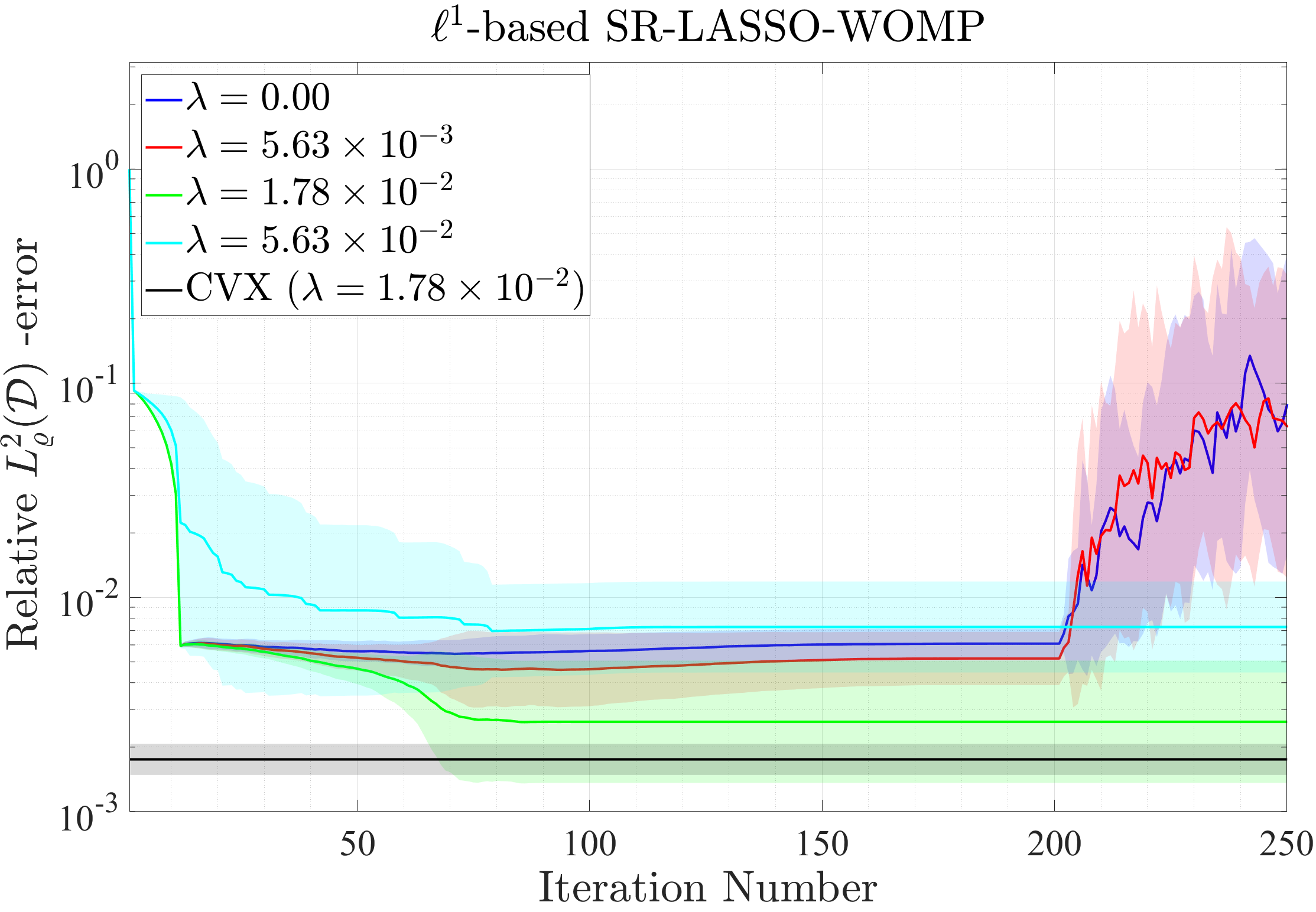}
	\end{subfigure}\hspace{0.5em}%
	\begin{subfigure}[t]{0.33\linewidth}
		\includegraphics[width = \textwidth]{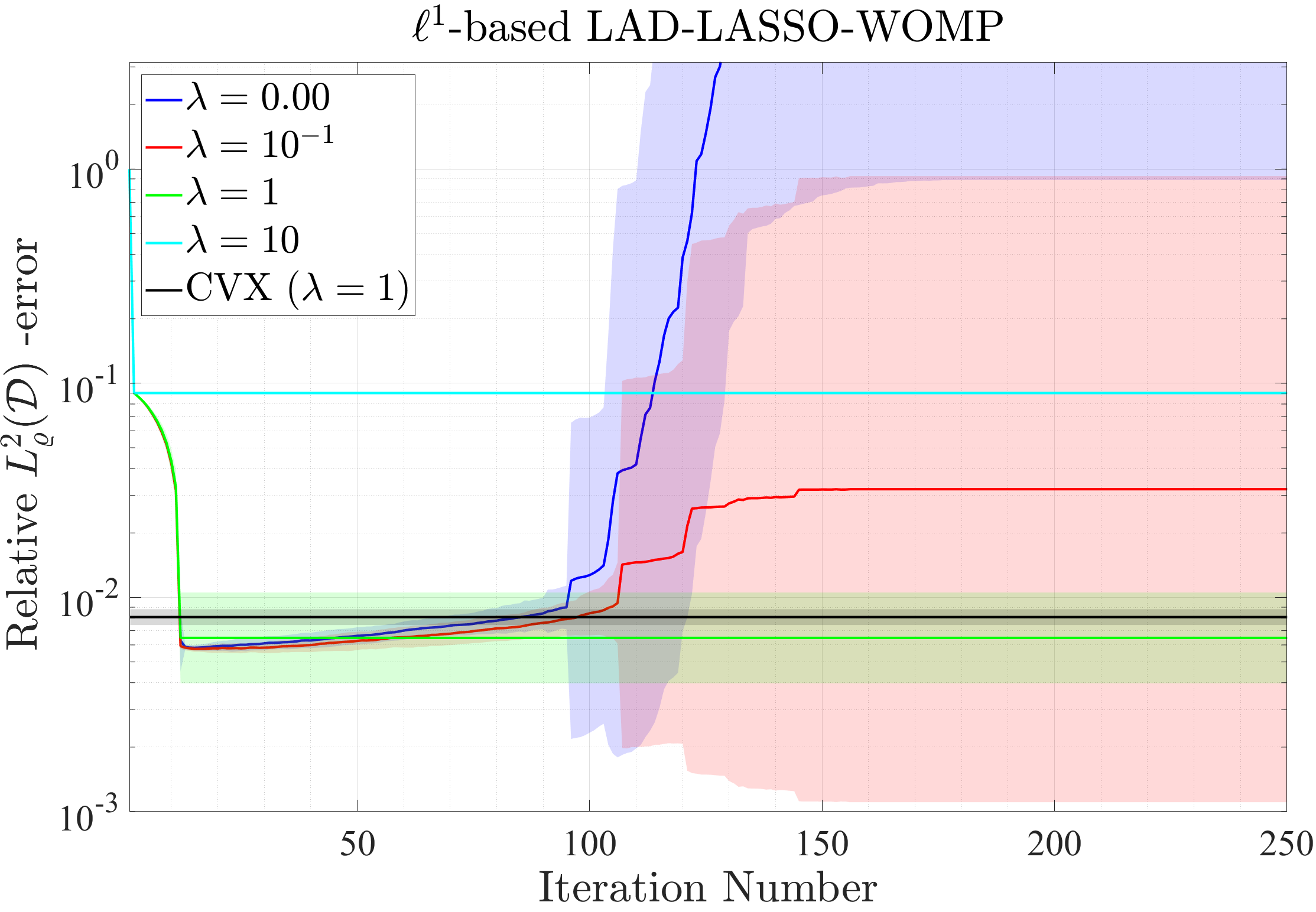}
	\end{subfigure}
	\caption{Relative error as a function of the iteration number  (Experiment V, function approximation). The proposed $\ell^1$-based WOMP formulations are tested for different values of the tuning parameter $\lambda$. The black curve corresponds to recovery via convex optimization of the corresponding loss function.}
	\label{fig:error_vs_iteration_func_approx}
\end{figure*}
The solid curves represent the mean relative error as a function of iteration number. The upper and lower boundaries of the shaded areas are designated by plotting the discrete points $ (k, 10^{\mu^k_{\lambda} + \sigma^k_{\lambda}}) $ and $ (k, 10^{\mu^k_{\lambda} - \sigma^k_{\lambda}}),\ k \in [N_{\mathrm{iter}}],\ \lambda \in \mathcal{L}$. Here $ \mu^k_{\lambda} $ and $ \sigma^k_{\lambda} $ denote, respectively, the sample mean and the sample standard deviation of the $\log_{10}$-transformed relative $\ell^2$-error at iteration $ k $, i.e.,
$$  \mu^k_{\lambda} = \frac{1}{N_{\mathrm{trial}}}\sum_{i = 1}^{N_{\mathrm{trial}}}\log((E^{(k)}_{\lambda})_i)\quad \text{and} \quad
\sigma^k_{\lambda} = \sqrt{\frac{1}{N_{\mathrm{trial}} - 1}\sum_{i = 1}^{N_{\mathrm{trial}}}\left(\log((E^{(k)}_{\lambda})_i) - (\mu^k_{\lambda})_i\right)^2}. $$
(See also \cite[\S A.1.3]{adcock2022sparse} for more details.) The set $ \mathcal{L} $ of tuning parameters always consists of $ \lambda  = 0 $ (in blue), as well as the best empirical $ \lambda $ (in green), an underestimated $\lambda$ (in red), and an overestimated one (in cyan). By the best empirical $ \lambda $, we mean the $ \lambda $ that achieves the smallest empirical relative error $ E_{\lambda}^{(k)} $, over a wide range of explored values and on average for $ N_{\text{trial}} $ random trials. Moreover, we compare each $ \ell^1 $-based WOMP formulation, with the corresponding convex optimization problem, with optimally tuned $ \lambda $ (in black). This experiment confirms once again that when $ \lambda $ is tuned appropriately, $ \ell^1 $-based WOMP algorithms can effectively perform sparse recovery from compressive measurements. In particular, for suitably chosen values of $ \lambda $, WOMP is robust with respect to the iteration number. On this note, we observe that standard OMP (corresponding to $ \lambda = 0 $ for LASSO and SR-LASSO) begins to severely overfit when the iteration number is larger than $ m $. The reason behind this phenomenon is that in standard OMP there is no regularization mechanism that prevents the greedy selection from adding more indices to the support than number of measurements. Therefore, after $ m $ iterations the least-squares fitting in standard OMP leads to severe overfitting and the algorithm starts diverging. A similar phenomenon is observed in \cite{adcock2020sparse} for $\ell^0$-based LASSO WOMP.

\subsection{Runtime, loss reduction and stopping criterion}
\label{s:runtime_loss_stop}
In this subsection we further demonstrate the robustness of the proposed algorithms with respect to the number of iterations. To do so, we consider additional numerical experiments aimed at illustrating the benefits of regularization for the computational efficiency of the algorithms and their stopping criteria. As already mentioned, this robustness is due to the fact that the support stops increasing once the sparsity of the reconstructed signal reaches a saturation point. This removes the need to perform a data-fitting step in the subsequent iterations, which, in turn, leads to significant runtime savings. In addition, this feature allows one to consider a reliable stopping criterion, i.e., halting the algorithm when the loss function reaches a steady state, or equivalently when the greedy selection leads to an index that already belongs to the current support.

\paragraph{Experiment VI (runtime).}
Thanks to the mechanism discussed above, the increase in runtime overhead of $*$-LASSO WOMP algorithms' iterations becomes negligible after a certain point. To show this, we consider the same numerical setting as in Experiment IV for sparse random Gaussian signals, but this time we measure algorithms' runtime as a function of the iteration number. Figure~\ref{fig:runtime_vs_iteration_sparse_gaussian} clearly illustrates that for an appropriate choice of $ \lambda $, WOMP imposes no significant increase in runtime overhead after a certain number of iteration corresponding to the signal's sparsity. 
\begin{figure*}[t!]
	\begin{subfigure}[t]{0.33\linewidth}
		\includegraphics[width = \textwidth]{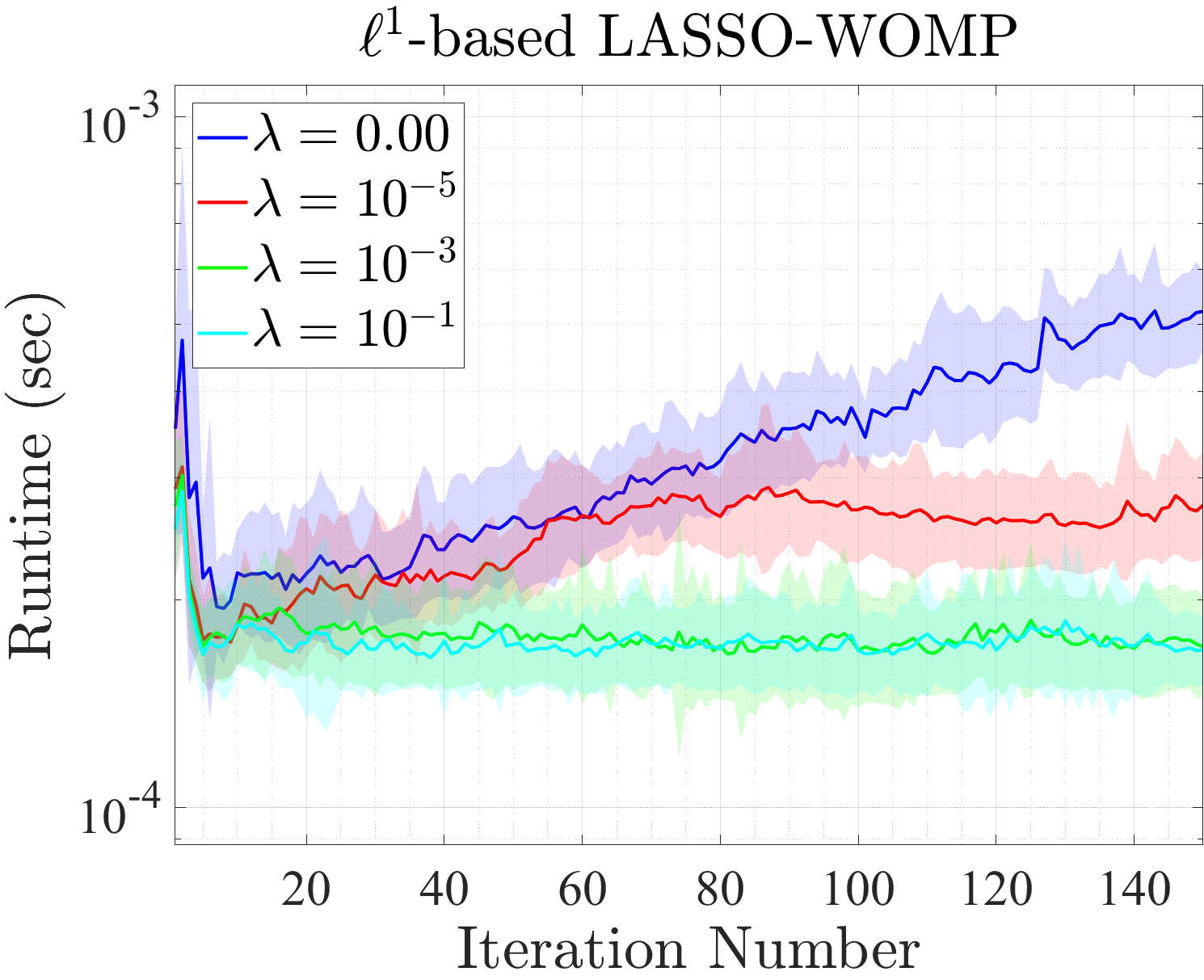}
	\end{subfigure}\hspace{0.5em}%
	\begin{subfigure}[t]{0.33\linewidth}
		\includegraphics[width = \textwidth]{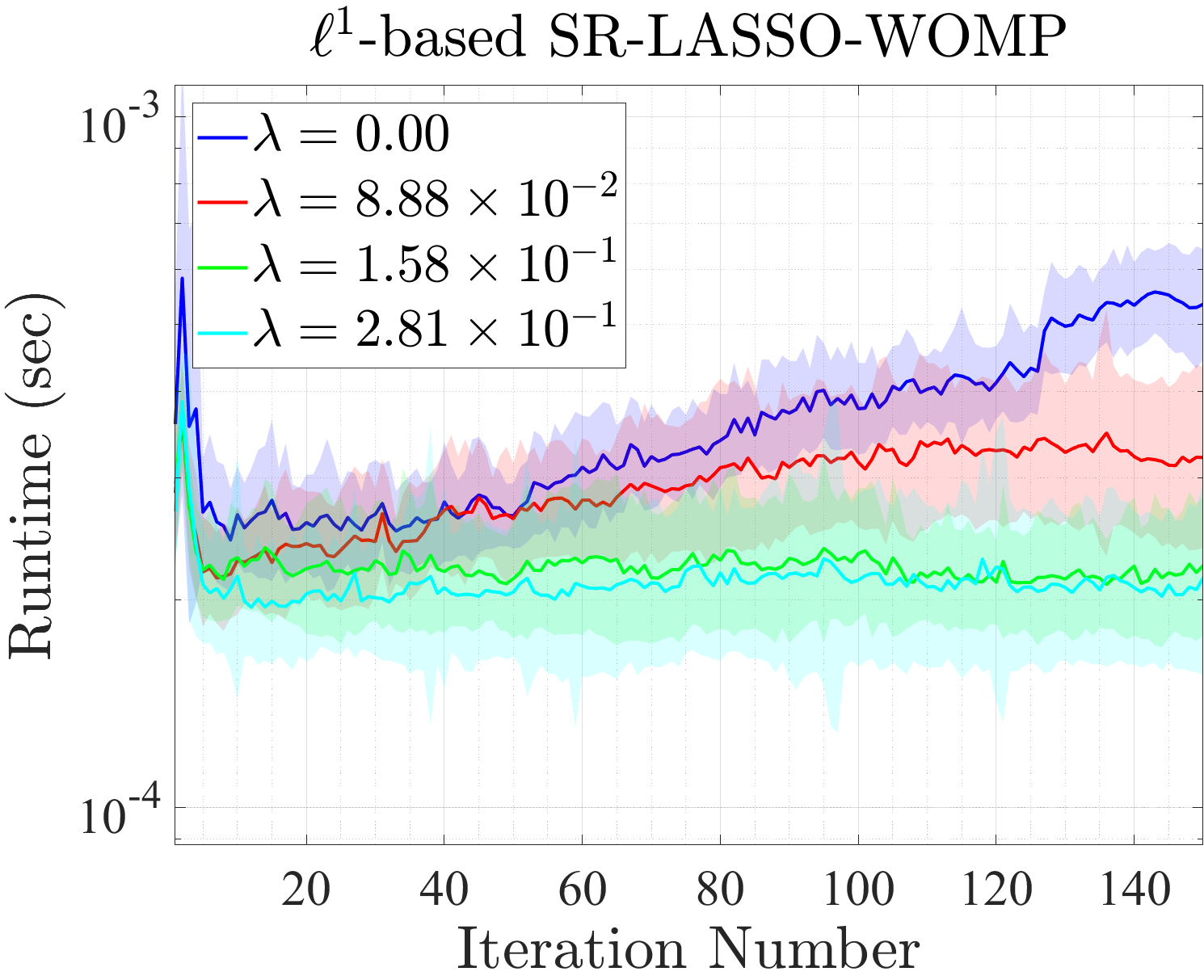}
	\end{subfigure}\hspace{0.5em}%
	\begin{subfigure}[t]{0.33\linewidth}
		\includegraphics[width = \textwidth]{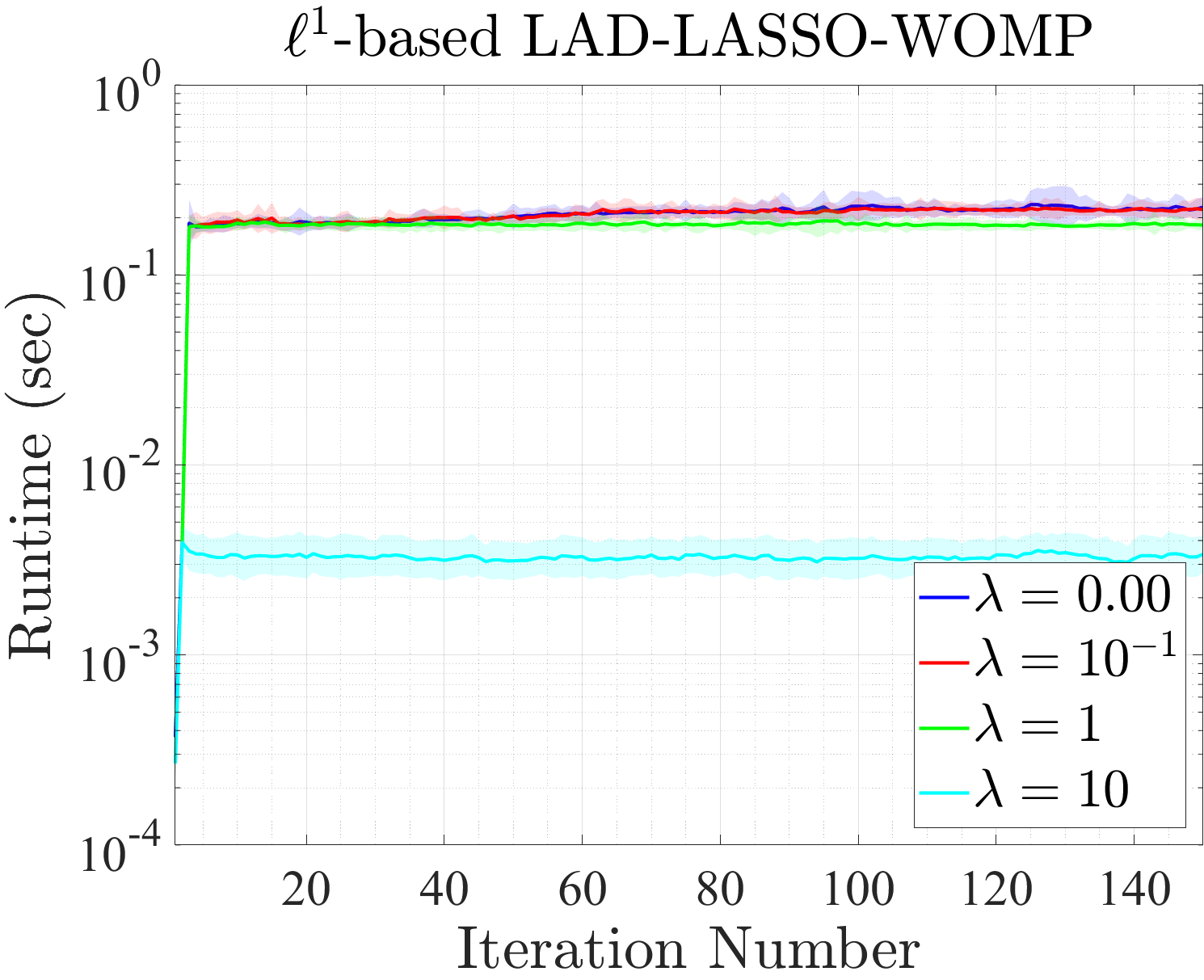}
	\end{subfigure}
	\caption{Algorithms' runtime as a function of the iteration number (same data as Experiment VI, sparse random Gaussian setting). The proposed $\ell^1$-based WOMP formulations are tested for different values of the tuning parameter $\lambda$. Colors are consistent with the ones in Figure \ref{fig:error_vs_iteration_sparse_gaussian}.}
	\label{fig:runtime_vs_iteration_sparse_gaussian}
\end{figure*}
This is not the case for standard OMP (associated with $ \lambda  = 0 $ in LASSO and SR-LASSO WOMP) as it has to solve increasingly large least-squares problems as the iterations proceed. A similar phenomenon is also observed in the context of function approximation. In higher dimensions (greater values of $ m $), this advantage becomes even more pronounced, as solving the least-squares becomes more computationally expensive (for the sake of conciseness, we omit these plots from the paper). 

Next we revisit a point mentioned in passing in Remark~\ref{remark:solving_LAD}. As mentioned, with $ N $ large and $ s $ small, reducing the dimensionality of a high-dimensional problem over $ \mathbb{R}^{m \times N} $ into smaller consecutive problems over $ \mathbb{R}^{m \times k} $, where $ k $ is the iteration number, can offer computational advantages. This phenomenon is well captured by Figure \ref{fig:error_vs_runtime_ladlassowomp}, where we plot the runtime of CVX and LAD-LASSO WOMP as a function of the solution sparsity, as well as a ``dartboard'' plot of relative $ \ell^2 $-error with respect to runtime of these methods for different values of sparsity.
\begin{figure*}[t!]
	\centering
	\includegraphics[height = 5cm]{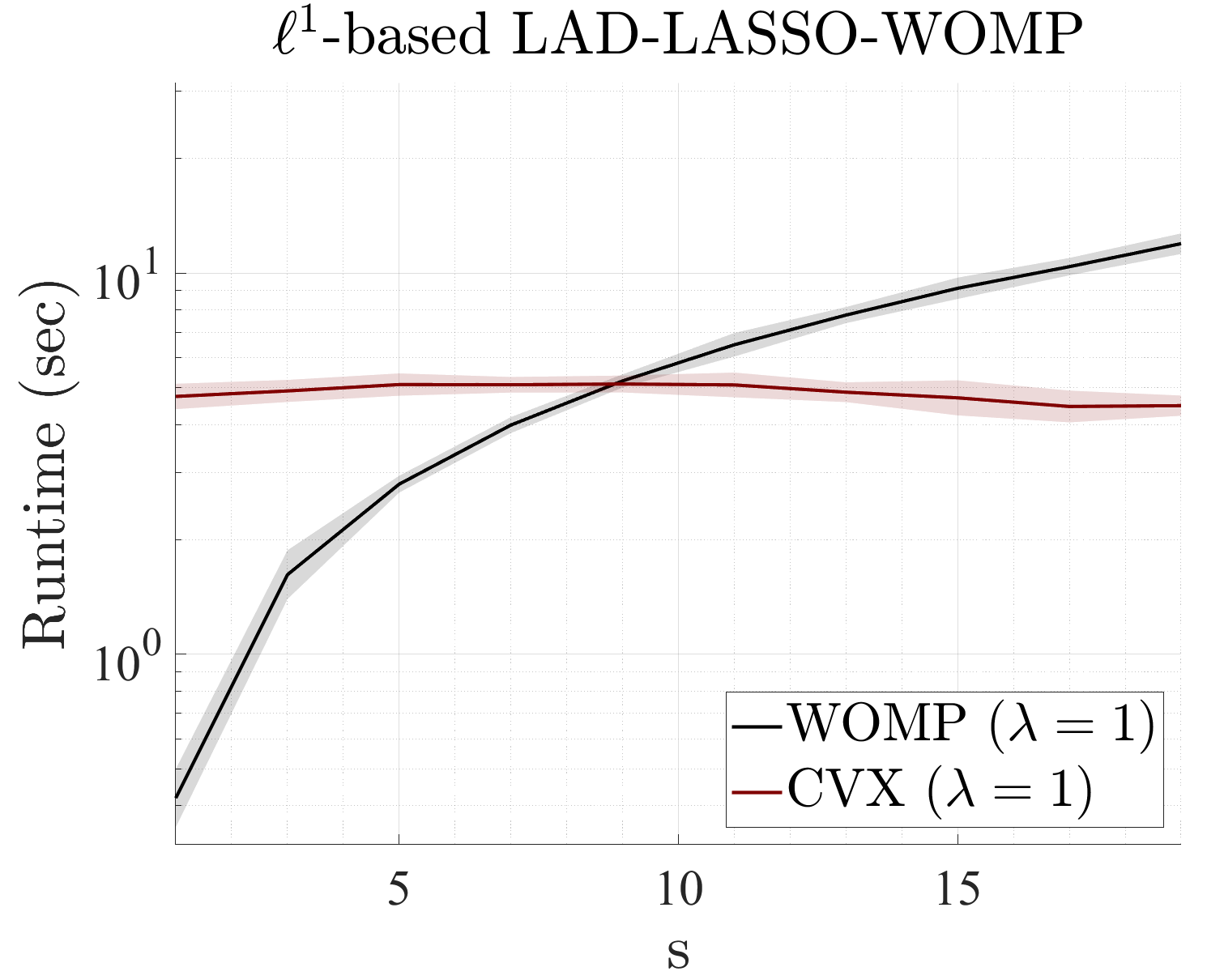} \hspace{0.5em}%
	\includegraphics[height = 5cm]{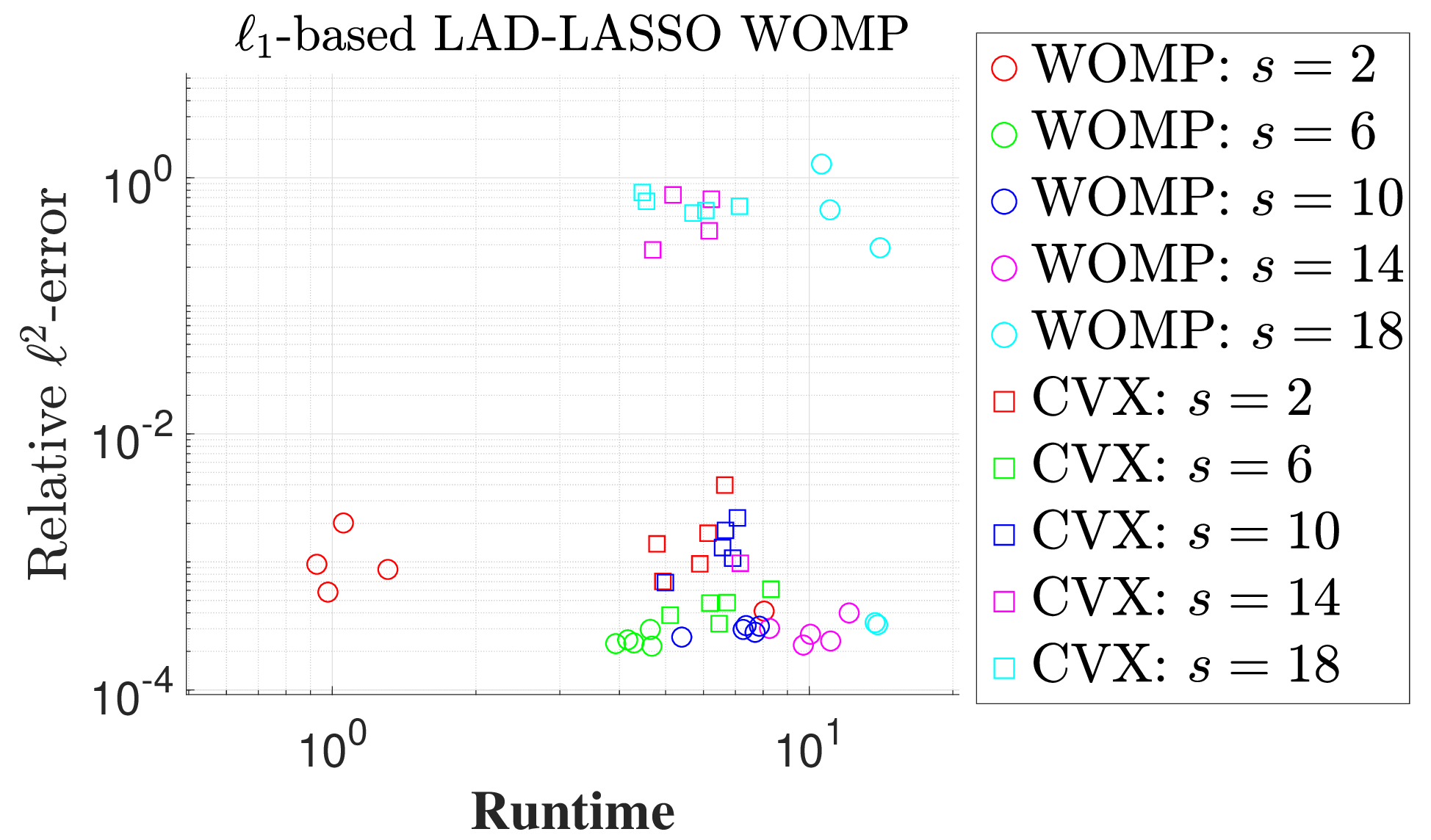}
	\caption{On the left: Runtime of LAD-LASSO WOMP and CVX as a function of solution sparsity $ s $.
	On the right: Relative $ \ell^2 $-error and runtime of LAD-LASSO and CVX for different values of sparsity. Both plots are generated using the same numerical settings.}
	\label{fig:error_vs_runtime_ladlassowomp}
\end{figure*}
For an ambient space dimension of $ N = 4000 $ and $ m = 120 $ measurements, we plot the runtime as a function of $ s $ over 15 trials for the left and 5 trials for the right figure. For LAD-LASSO WOMP, we always fix the number of iterations to $ 2s $. 
These plots vividly reveal the existence of a phase transition. For small enough values of $ s $, running LAD-LASSO-WOMP is cheaper than solving a LAD-LASSO problem with CVX. After a critical value of $ s $ (in this case, $ s = 10 $), CVX converges faster than LAD-LASSO WOMP. Nevertheless, it is worth noting that CVX is unable to attain a reliable solution for $ s \geq 14 $ (the relative error is close to 1), whereas LAD-LASSO WOMP is able to compute a more accurate solution for these values of $ s $.

\paragraph{Experiment VII (loss reduction).}
In Figure \ref{fig:loss_vs_iteration_sparse_gaussian} we plot the loss value, i.e., values of Equations \eqref{eq:loss_WLASSO}, \eqref{eq:loss_WSRLASSO} and \eqref{eq:loss_WLADLASSO}, respectively, against the iteration number. 
\begin{figure*}[t!]
	\begin{subfigure}[t]{0.33\linewidth}
		\includegraphics[width = \textwidth]{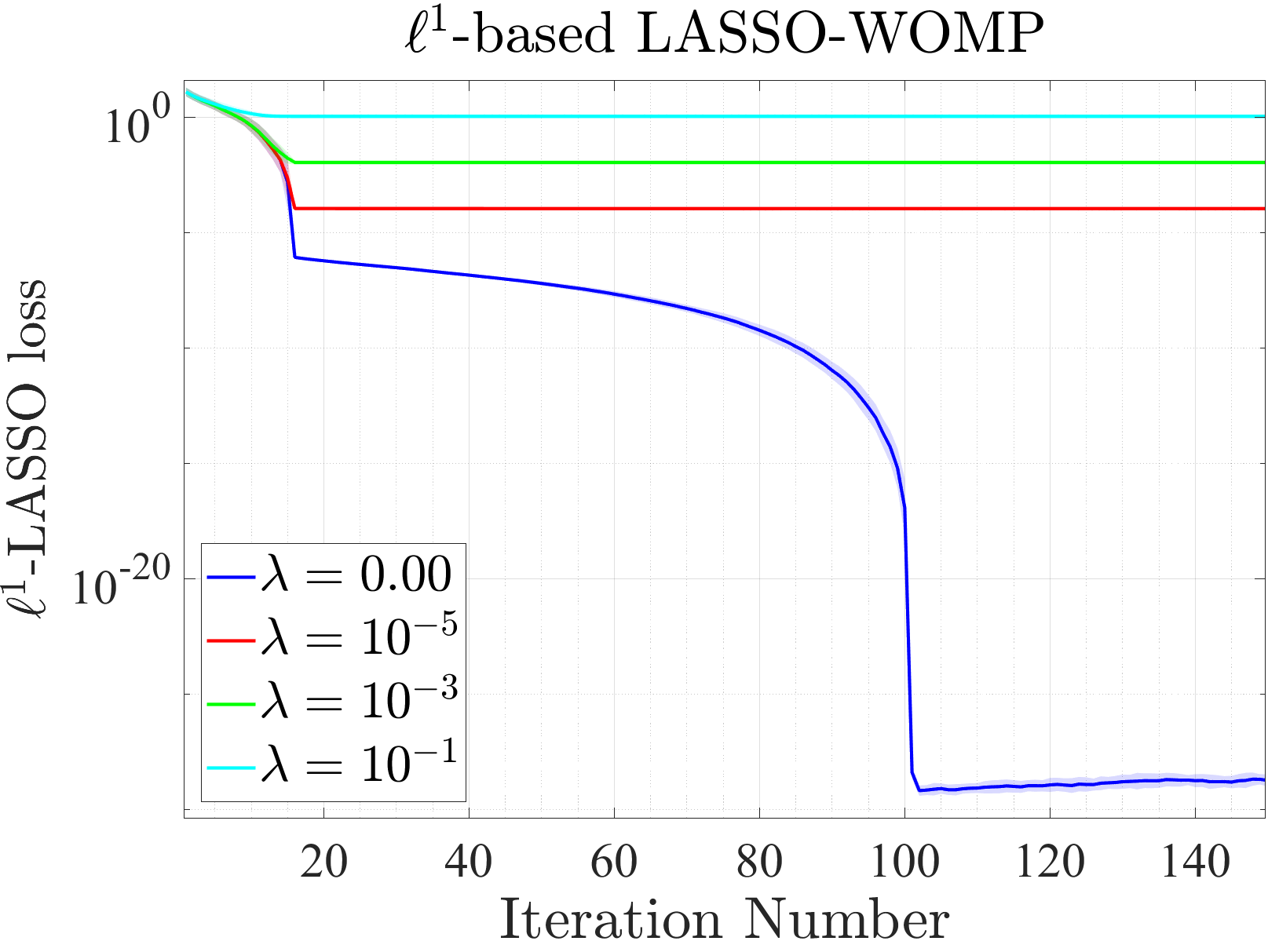}
	\end{subfigure}\hspace{0.5em}%
	\begin{subfigure}[t]{0.33\linewidth}
		\includegraphics[width = \textwidth]{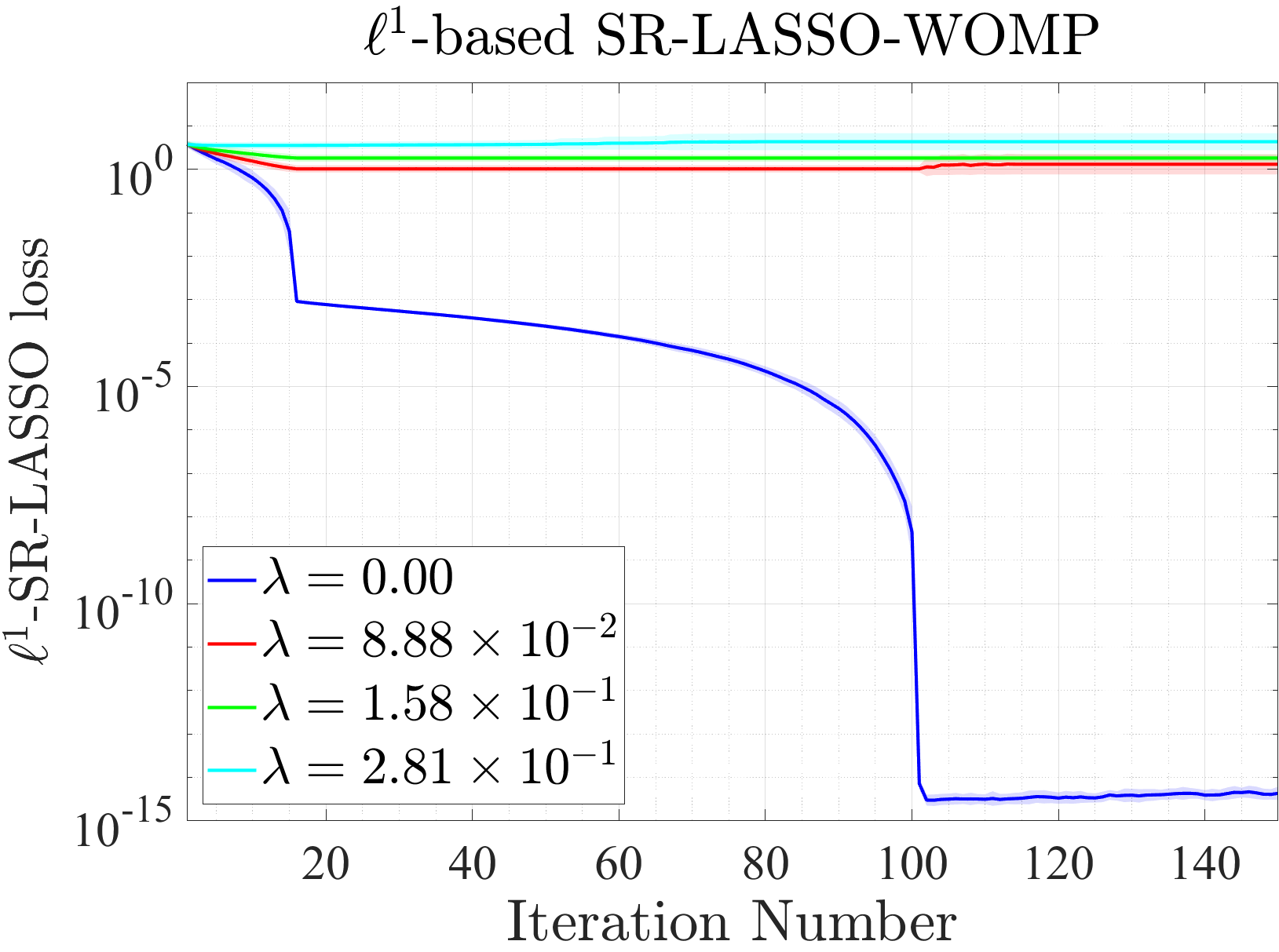}
	\end{subfigure}\hspace{0.5em}%
	\begin{subfigure}[t]{0.33\linewidth}
		\includegraphics[width = \textwidth, height = 4.1cm]{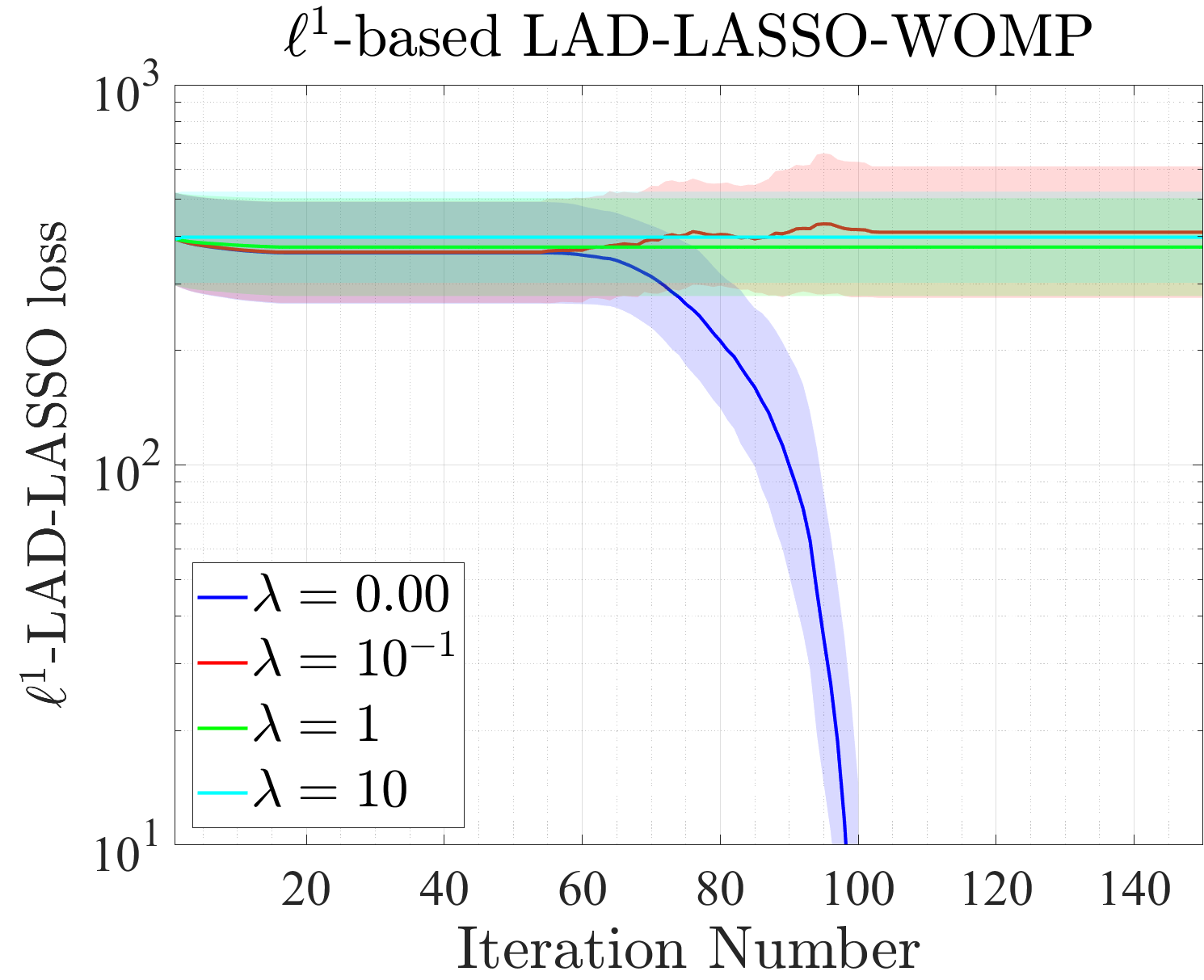}
	\end{subfigure}
	\caption{Loss value of algorithms as a function of the iteration number (same data as Experiment VI, sparse random Gaussian setting). The proposed $\ell^1$-based WOMP formulations are tested for different values of the tuning parameter $\lambda$. Colors are consistent with the ones in Figure \ref{fig:loss_vs_iteration_sparse_gaussian}.}
	\label{fig:loss_vs_iteration_sparse_gaussian}
\end{figure*}
We note again, in accordance with the results of Experiment VI, that for an appropriate choice of the tuning parameter, regularization also stabilizes the loss value after a certain iteration number.

\paragraph{Discussion on stopping criteria.}
The results of the previous experiments suggest that thanks to regularization of the loss functions in $*$-LASSO WOMP algorithms, there are several options for the stopping criterion. 
In case an estimate of the sparsity $s$ of the signal is available, one can run a number $ K > s $ iterations of $*$-LASSO WOMP and still ensure that for an appropriate choice of $ \lambda $ the algorithm does not overfit, in light of Figures \ref{fig:error_vs_iteration_sparse_gaussian} and \ref{fig:error_vs_iteration_func_approx} (note also that running more iterations than $ s $ would not imposes too much runtime overhead in light of Figure \ref{fig:runtime_vs_iteration_sparse_gaussian}).
If the signal's sparsity cannot be estimated in advance, one can monitor the loss value over iterations and halt the algorithm when the loss reaches a steady state, although this might be challenging when the loss reduction is very small. %\orange{Sina: [I tend to weaken the confidence of this statement (and other similar ones). In many cases (see figure 8) even before $ s $ the amount of loss reduction per iteration is very small and setting an appropriate stopping value based on loss function could be challenging. I think the third method (stopping criterion based on the indices) is the most reliable and natural stopping criterion and should be promoted.]}
Finally, one can halt the algorithm when the greedy selection leads to an index that already belongs to the current support. Based on our experience, the latter seems to be the most reliable option.

\section{Conclusions and future research}
\label{s:conclusions}
Adopting a loss-function perspective, we  proposed new generalizations of OMP for weighted sparse recovery based on $\ell^0$ and $\ell^1$ versions of the weighted LASSO, SR-LASSO and LAD-LASSO. Moreover, we showed that the corresponding greedy index selection rules admit explicit formulas (see Theorems~\ref{thm:WLASSO}, \ref{thm:WSRLASSO}, \ref{thm:WLADLASSO} and Appendix~\ref{app:l0_greedy_selection}). Through numerical illustrations, in \S\ref{sec:numerics} we observed that these algorithms inherit desirable characteristics from  some of the associated loss functions, i.e., independence of the tuning parameter to noise level for SR-LASSO, and robustness to sparse corruptions for LAD-LASSO.
There are many research pathways still to be pursued. We conclude by discussing some of them.

Although we focused on LASSO-type loss functions, many other regularizers and loss functions remain to be investigated, depending on the context and specific application of interest. This includes regularization based on total variation, nuclear norm, $\ell^p-\ell^q$ norms, and group (or joint) sparsity. One might also attempt to accelerate OMP's convergence by sorting indices based on the greedy selection rules derived in this paper and selecting more indices at each iteration. This procedure is employed in algorithms such as CoSaMP \cite{needell2009cosamp}, which can thus be easily generalized to the loss function-based framework. The same holds for a recently proposed sublinear-time variant of CoSaMP \cite{choi2021sparse}. It is also worth noting that a different method to incorporate weights into OMP is based on greedy index selection rules of the form
$$ j^{(k)} \in \arg\max_{j \in [N]}\left|(A^*(y - Ax^{(k)}))_j/w_j\right| $$
(see \cite{bouchot2017multi,li2013weighted,wu2013weighted}).  The comparison, both empirical and theoretical, of rules of this form with the loss function-based criteria considered here deserves further investigation.

Regarding future theoretical developments, Theorems \ref{thm:WLASSO}--\ref{thm:WLADLASSO} demonstrate that the greedy index selection rules considered here achieve maximal loss-function reduction  at each iteration. However, these theorems do not provide recovery guarantees for loss-function-based OMP. The development of rigorous recovery theorems based on the Restricted Isometry Property (RIP) or the coherence is an important open problem. An interesting related question is whether the theoretical recipes for the optimal choice of $\lambda$ available for convex optimization decoders (see \S\ref{subsub:weighted_l1_formulations}) remain valid in the greedy setting.

Finally, although in this paper we focused on high-dimensional function approximation, there are many more applications where  loss-function based OMP could be tested. A particularly promising one is video reconstruction from compressive measurements that arises in contexts such as dynamic MRI, where one can incorporate information on the ambient signal of the previously reconstructed frames through weights in order to improve the reconstruction quality of subsequent frames (see, e.g., \cite{friedlander2011recovering}). Exploring the benefits of loss-function based OMP in this and other applications will be the object of future research work.

\section*{Acknowledgements}

The authors acknowledge the support of the Natural Sciences and Engineering Research Council of Canada (NSERC) through grant RGPIN-2020-06766, the Fonds de recherche du Qu\'ebec Nature et Technologies (FRQNT) through
grant 313276, the Faculty of Arts and Science
of Concordia University and the CRM Applied Math Lab. The authors would like to thank the two anonymous referees for their insightful comments and helpful suggestions.

\section*{Statements and declarations}

On behalf of all authors, the corresponding author states that there is no conflict of interest.

\bibliographystyle{plain}
\bibliography{main_STSPDA_final}

\appendix

\section{Proof of the main results}
\label{app:proofs}

In this section we prove Theorems \ref{thm:WLASSO}, \ref{thm:WSRLASSO} and \ref{thm:WLADLASSO}.

\subsection{Proof of Theorem~\ref{thm:WLASSO}}
%\begin{proof}
	Let $G = G_{\ell^1_w}$, the weighted LASSO loss function defined in \eqref{eq:loss_WLASSO}. The argument is organized into to two cases: $j \notin S$ and $j \in S$. 
	
	\paragraph{Case 1: $ j \notin S$.} In this case, $j \notin \supp(x) \subseteq S$ and, recalling that the columns of $A$ are $\ell^2$-normalized, for any $t \in \mathbb{R}$, we can write
	\begin{flalign*}
		G(x + te_j) &= \|y - A (x + t e_j)\|_2^2 + \lambda\|x + te_j\|_{1, w} \\
		&= \| y - Ax\|_2^2 + |t|^2 -2 \Re(\overline{t} \langle y - Ax, Ae_j\rangle) + \lambda (\|x\|_{1,w} + |t| w_j).
	\end{flalign*}
	Our goal is to minimize the above quantity over $t \in \mathbb{C}$. 
 For any $ t \in \bC $ we let $ t = \rho e^{i\theta}$, with $ \rho \geq 0 $ and $ 0 \leq \theta < 2\pi $. Using the expression above we obtain
	\begin{align*}
		G(x + te_j) 
		& = \|y - Ax\|_2^2 + \rho^2 - 2\Re(\rho e^{-i\theta} (A^*(y - Ax))_j ) + \lambda \|x\|_{1,w} + \lambda \rho w_j\\
		& \geq \|y - Ax\|_2^2 + \rho^2 - 2\rho|(A^*(y - Ax))_j| + \lambda \|x\|_{1,w} + \lambda\rho w_j, \numberthis \label{eq:exp_to_min} 
	\end{align*}
	where the inequality holds as an equality for some $ 0 \leq \theta < 2\pi $. An explicit computation shows that \eqref{eq:exp_to_min} is minimized at 
	$ \rho = |(A^*(y - Ax))_j| - \frac{\lambda}{2}w_j, $
	if $|(A^*(y - Ax))_j| \geq \frac{\lambda}{2}w_j$, and at $\rho = 0$ otherwise.
	Plugging this value into \eqref{eq:exp_to_min} we obtain
	$$ \min_{t \in \bC} G(x + te_j) = G(x) - \max\left\{|(A^*(y - Ax))_j| - \frac{\lambda}{2} w_j, 0\right\}^2, $$
 as desired.
	
	\paragraph{Case 2: $ j \in S $.} Letting $r = y-Ax$, we see that $ (A^*r)_j = 0 $ since, by assumption, $Ax$ is the orthogonal projection of $y$ onto the span of the columns of $A$ indexed by $S$. Thus we can write
	\begin{flalign}
		G(x + te_j) &= \|y - A(x + te_j)\|_2^2 + \lambda \|x + te_j\|_{1,w} \nonumber \\
		&= \|y - A(x + te_j)\|_2^2 + \lambda \|x - x_je_j\|_{1, w} + \lambda|x_j + t|w_j  \nonumber \\
		&= \|r\|_2^2 + \underbrace{|t|^2 + \lambda w_j|x_j + t|}_{=:l(t)} + \lambda \|x - x_je_j\|_{1, w}. \label{eq:l1_LASSO_Case_2_def_l}
	\end{flalign}
	Now, we want to minimize $ l(t) $ over $t \in \mathbb{C}$. Let $\rho = |t|$. By the triangle inequality we have
	$$ l(t) 
	= |t|^2 + \lambda w_j|x_j + t| 
	\geq |t|^2 + \lambda w_j||x_j| - |t|| 
	= \rho^2 + \lambda w_j||x_j| - \rho|
	= :\hat{l}(\rho), $$
	where the first inequality holds as an equality only if $ t = \alpha x_j $ for some $\alpha \in\mathbb{R}$ with $\alpha \leq 0$. Therefore, $\min_{t\in\mathbb{C}} l(t) = \min_{\rho \in  [0,+\infty)} \hat{l}(\rho)$ (since given a minimizer $\hat{\rho}$ of $\hat{l}$, then $\hat{t} = -\hat{\rho} x_j /|x_j|$ is a minimizer of $l$) and it is sufficient to minimize $\hat{l}$. 
	If $ \rho \geq |x_j| $, then 
	$ \hat{l}(\rho) = \rho^2 + \lambda w_j\rho - \lambda w_j|x_j|, $
	which is minimized at $\rho = |x_j|$. Otherwise, if $ 0 \leq \rho \leq |x_j| $, we have
	$ \hat{l}(\rho) = \rho^2 - \lambda w_j\rho + \lambda w_j|x_j|. $
	In this case, a direct computation shows that $ \hat{l}(\rho) $ is minimized at $ \rho = \lambda w_j/2 $ if $\lambda w_j/2 < |x_j|$, or at $\rho = |x_j|$ otherwise.
	Summarizing the above discussion, we have
	$$
	\min_{t \in \bC} l(t) 
	= \min_{\rho \in [0, +\infty)} \hat{l}(\rho) 
	= \min\left\{\hat{l}\left(|x_j|\right), \hat{l}\left(\frac{\lambda w_j}{2}\right)\right\}
	= \min\left\{|x_j|^2, \left(\frac{\lambda w_j}{2}\right)^2 + \lambda w_j\left|\frac{\lambda w_j}{2} - |x_j|\right|\right\}.
	$$
	Therefore, recalling \eqref{eq:l1_LASSO_Case_2_def_l}, we see that
	\begin{flalign*}
		& \min_{t \in \bC} G(x + te_j) 
		= \|r\|_2^2  + \lambda \|x - x_je_j\|_{1, w} + \min\left\{|x_j|^2, \left(\frac{\lambda w_j}{2}\right)^2 + \lambda w_j\left|\frac{\lambda w_j}{2} - |x_j|\right|\right\}\\
		& \qquad = \|r\|_2^2  + \lambda \|x - x_je_j\|_{1, w} + \lambda w_j|x_j| - \lambda w_j|x_j| 
		+ \min\left\{|x_j|^2, \left(\frac{\lambda w_j}{2}\right)^2 + \lambda w_j\left|\frac{\lambda w_j}{2} - |x_j|\right|\right\} \\
		& \qquad = G(x) + \min\left\{|x_j|^2 - \lambda w_j|x_j|, \left(\frac{\lambda w_j}{2}\right)^2 + \lambda w_j\left|\frac{\lambda w_j}{2} - |x_j|\right| - \lambda w_j|x_j|\right\} \\
		& \qquad = G(x) - \max\left\{-|x_j|^2 + \lambda w_j|x_j|, -\left(\frac{\lambda w_j}{2}\right)^2 - \lambda w_j\left||x_j| - \frac{\lambda w_j}{2}\right| + \lambda w_j|x_j|\right\},
	\end{flalign*}
	which concludes the proof.

\hfill $\square$

\subsection{Proof of Theorem~\ref{thm:WSRLASSO}}
%\begin{proof}
Let $G = G^{\mathrm{SR}}_{\ell^1_w}$ be the weighted SR-LASSO loss function defined in \eqref{eq:loss_WSRLASSO} and recall that
	$ r = y - Ax $. The proof strategy is analogous to that of Theorem~\ref{thm:WLASSO} and is organized into two cases.
	
	\paragraph{Case 1: $ j \notin S$.}
	Letting $ t = \rho e^{i\theta} \in \bC $, we have
	\begin{align*}
		G(x + t e_j) &= \|r - ta_j\|_2 + \lambda \|x\|_{1,w} + \lambda \rho w_j\\
		&= \sqrt{\|r\|_2^2 + \rho^2 - 2\rho\Re( e^{-i\theta}\langle r, a_j\rangle)} + \lambda \|x\|_{1,w} + \lambda \rho w_j\\
		&\geq \underbrace{\sqrt{\|r\|_2^2 +\rho^2 - 2\rho|\langle r, a_j\rangle|} + \lambda \|x\|_{1,w} + \lambda \rho w_j}_{=:h(\rho)},
		\numberthis \label{eq:exp_to_min_SRLASSO} 
	\end{align*}
	where the last inequality holds as an equality for some $0 \leq \theta < 2\pi$.
	In order to minimize the right-hand side, we compute
	$$ h'(\rho) = \frac{\rho - |\langle r,a_j \rangle|}{\sqrt{\|r\|_2^2 + \rho^2 - 2\rho|\langle r,a_j \rangle|}} + \lambda w_j. $$
    If $w_j \lambda \geq 1$, the equation $h'(\rho) =0$ does not have any solution over $[0, +\infty)$ (since $|\langle r, a_j\rangle| \leq \|r\|_2$ due to the Cauchy-Schwarz inequality). Hence, in that case $h$ is minimized at $\rho =0$. On the other hand, if $ w_j \lambda < 1 $, the equation $ h'(\rho) = 0 $ has the unique solution
	$$ \widetilde\rho = |\langle r,a_j\rangle| - \sqrt{\frac{(\lambda w_j)^2 (\|r\|_2^2 - |\langle r,a_j \rangle|^2)}{1 - (\lambda w_j)^2}}. $$
	Therefore, the minimizer of $ h $ on $[0,+\infty)$ is either $\widetilde{\rho}$ or $0$. Plugging $ \rho = \widetilde \rho $ and $\rho=0$ into \eqref{eq:exp_to_min_SRLASSO}, we obtain
	\begin{flalign*}
		\min_{t \in \bC} G(x + te_j) &= \min\left\{\sqrt{(1 - (\lambda w_j)^2)(\|r\|_2^2 - |\langle r, a_j \rangle|^2)} + \lambda w_j |\langle r,a_j \rangle| + \lambda \|x\|_{1,w}, \|r\|_2 + \lambda \|x\|_{1, w}\right\} \\
		&= \min\left\{\sqrt{(1 - (\lambda w_j)^2)(\|r\|_2^2 - |\langle r, a_j \rangle|^2)} + \lambda w_j |\langle r,a_j \rangle| - \|r\|_2, 0\right\} + \|r\|_2 + \lambda \|x\|_{1, w} \\
		&= G(x) -\max\left\{\|r\|_2 - \lambda w_j |\langle r,a_j \rangle| - \sqrt{(1 - (\lambda w_j)^2)(\|r\|_2^2 - |\langle r, a_j \rangle|^2)}, 0\right\} ,
	\end{flalign*}
	which concludes the case $ j \notin S $.\\
	
	\paragraph{Case 2: $ j \in S$.} In this situation, $ |(A^*r)_j| = 0 $ since $x$ solves a least-squares problem. Thus we can write
	\begin{flalign*}
		G(x + te_j) &= \|y - A(x + te_j)\|_2 + \lambda \|x + te_j\|_{1,w} \\
		&= \sqrt{\|y - A(x + te_j)\|_2^2} + \lambda \sum_{i \in S\backslash \{j\}}w_i|x_i| + \lambda|x_j + t|w_j \\
		&= \underbrace{\sqrt{|t|^2 + \|r\|_2^2} + \lambda w_j|x_j + t|}_{=:l(t)} + \lambda \|x - x_je_j\|_{1, w}.
	\end{flalign*}
	We continue by minimizing $ l $ over $\mathbb{C}$. Letting $t = \rho e^{i\theta}$, by the triangle inequality we have
	$$ l(t) 
	= \sqrt{|t|^2 + \|r\|_2^2} + \lambda w_j|x_j + t| 
	\geq %\sqrt{|t|^2 + \|r\|_2^2} + \lambda w_j||x_j| - |t|| 
	\sqrt{\rho^2 + \|r\|_2^2} + \lambda w_j||x_j| - \rho| 
	=: \hat{l}(\rho), 
	$$
	where the inequality holds as an equality  when $ t = \alpha x_j $ for some $ \alpha \in \bR $ with $\alpha \leq 0$.
	If $ \rho \geq |x_j| $ we have
	$ \hat{l}(\rho) = \sqrt{\rho^2 + \|r\|_2^2} + \lambda w_j\rho - \lambda w_j|x_j|, $
	which is minimized at $ \rho = |x_j|$. Otherwise, if $0\leq \rho < |x_j| $, we have
	$ \hat{l}(\rho) = \sqrt{\rho^2 + \|r\|_2^2} - \lambda w_j\rho + \lambda w_j|x_j|, $
	and a direct computation shows that, when $\lambda w_j \geq 1$,
 $$ \hat{l}'(\rho) = \frac{\rho}{\sqrt{\rho^2 + \|r\|_2^2}} - \lambda w_j \leq 1 - 1  = 0, 
    $$
     and the equation $ \hat{l}'(\rho) = 0 $ is either solved for all $0 \leq \rho \leq |x_j|$ (if $r=0$ and $\lambda w_j = 1$) or for no values of $0 \leq \rho \leq |x_j|$ (otherwise). Hence, when $\lambda w_j \geq 1$,
    $\rho = |x_j|$ is a minimizer of $\hat{l}$ over $[0,|x_j|]$. Conversely, when $\lambda w_j < 1$, the equation $ \hat{l}'(\rho) = 0 $ is uniquely solved by
	$$ \rho = \frac{\lambda w_j \|r\|_2}{\sqrt{1 - (\lambda w_j)^2}}.
    $$
    Hence, in summary, $\hat{l}$ is minimized at 
    $$
    \widetilde{\rho} 
    := \begin{cases}
    |x_j| & \lambda w_j \geq 1 \\
    \min\left\{|x_j|, \frac{\lambda w_j \|r\|_2}{\sqrt{1 - (\lambda w_j)^2}}\right\} & \lambda w_j < 1 
    \end{cases}.
    $$
    This leads to 
    \begin{align*}
        \min_{t \in \bC}G(x + te_j) &= \min_{t \in \bC}l(t) + \lambda \|x - x_je_j\|_{1, w} \\
        & = \hat{l}(\widetilde{\rho}) + \|r\|_2 + \lambda \|x\|_{1, w}  - \|r\|_2 - \lambda w_j|x_j|\\
        & = G(x) - \left(\|r\|_2 + \lambda w_j|x_j| - \hat{l}(\widetilde{\rho})\right),
    \end{align*}
    which concludes the proof.	 \hfill $\square$

\subsection{Proof of Theorem~\ref{thm:WLADLASSO}}
In order to prove Theorem~\ref{thm:WLADLASSO}, we need the minimum for the LAD problem in the one-dimensional setting. To this purpose, we present the following lemma based on the arguments from \cite[Lemmas~1 \& 2]{bloomfield1983least}, that we include here for the sake of completeness.
\begin{lemma}[Explicit solution of univariate LAD]
    \label{lemma:1D-LAD-solution}
    Let $ y, a \in \bR^N$ and $L:\mathbb{R} \to [0,+\infty)$, defined by
    $$ L(t):=\|y - ta\|_1 = \sum_{i = 1}^N |y_i - ta_i|, \quad \forall t \in \mathbb{R}. $$
    Then a minimizer of $ L $ over $\mathbb{R}$ is 
    $$ t^* = \left(\frac{y_{\hat{i}}}{a_{\hat{i}}}\right), $$
    with $ \hat{i} = \tau(\hat{k}) $ where
	$$ \hat{k} = \min \left\{k \in [\|a\|_0]: \sum_{i = 1}^{k} \frac{|a_{\tau(i)}|}{\|a\|_1} \geq \frac{1}{2} \right\}, $$
	and where $ \tau:[\|a\|_0] \to \supp(a) $ is a bijective map defining a nondeacreasing rearrangement of the vector
    $$ \left(\frac{y_i}{a_i}\right)_{i \in \supp(a)} \in \bR^{\|a\|_0}, $$
    i.e., such that
    $$ \frac{y_{\tau(1)}}{a_{\tau(1)}} \leq \frac{y_{\tau(2)}}{a_{\tau(2)}} \leq \dots \leq \frac{y_{\tau(\|a\|_0)}}{a_{\tau(\|a\|_0)}}. $$
\end{lemma}

\begin{proof}
We define $ t_k := y_{\tau(k)}/a_{\tau(k)} $ and the open intervals $ I_0 := (-\infty, t_1),\ I_k := (t_k, t_{k + 1})$ for $k \in[\|a\|_0 - 1]$, and $I_{\|a\|_0} := (t_{\|a\|_0}, +\infty) $ (if $t_k=t_{k+1}$, we simply set $I_k = \emptyset$). 

Now, let $ t \in I_k^j $ for some $ k \in [\|\widetilde{a}_j\|_0 ]$. Then,
\begin{align*}
L(t) & = \sum_{i\notin \supp(a)}\left|y_i \right| + \sum_{i = 1}^{\|a\|_0}|a_{\tau(i)}||t_i - t| \\
& = \sum_{i\notin \supp(a)}\left|y_i \right| + \sum_{i = 1}^{k}|a_{\tau(i)}|(t - t_i) - \sum_{i = k + 1}^{\|a\|_0}|a_{\tau(i)}|(t - t_i). 
\end{align*}
Differentiating with respect to $ t $, we obtain
$$ 
L'(t) 
= \sum_{i = 1}^{k}|a_{\tau(i)}| - \sum_{i = k + 1}^{\|a\|_0}|a_{\tau(i)}| =: d_k, \quad \forall t \in I_k.
$$
Hence we see that, for any $ k \in [\|a\|_0-1]$, 
\begin{flalign*}
	d_{k + 1} &= \sum_{i = 1}^{k + 1}|a_{\tau(i)}| - \sum_{i = k + 2}^{\|a\|_0}|a_{\tau(i)}| 
	= \sum_{i = 1}^{k}|a_{\tau(i)}| - \sum_{i = k + 1}^{\|a\|_0}|a_{\tau(i)}| + 2|a_{\tau(k + 1)}| 
	= d_k + 2|a_{\tau(k + 1)}| 
		> d_k,
\end{flalign*}
implying that $ L'(t) $ is increasing with respect to $ t $ (wherever it is well defined). In summary, $ L $ is a positive and piecewise linear function with increasing derivative. Hence,
$$ 
\min_{t \in \bR} h(t) = h\left(t_{\hat{k}}\right), 
$$
where
\begin{flalign}
	\hat{k} &:= \min\left\{k \in [\|a\|_0]: d_k \geq 0
  \right\} \\
	&= \min\left\{k \in [\|a\|_0]: \sum_{i = 1}^{k}|a_{\tau(i)}| \geq \sum_{i = k + 1}^{\|a\|_0}|a_{\tau(i)}|\right\} \\
	& = \min\left\{k \in [\|a\|_0]: 2\sum_{i = 1}^{k}|a_{\tau(i)}| \geq \sum_{i = 1}^{\|a\|_0}|a_{\tau(i)}|\right\} \\
	&= \min\left\{k \in [\|a\|_0]: \sum_{i = 1}^{k}\frac{|a_{\tau(i)}|}{\|a\|_1} \geq \frac{1}{2}\right\}. \numberthis \label{w_median_l1}
\end{flalign}
(Note that $\|a\|_1\neq 0$ since we are assuming $a$ to be nonzero.) Finally, we let $\hat{i} := \tau(\hat{k})$, which concludes the proof.
\end{proof}
We are now in a position to prove Theorem~\ref{thm:WLADLASSO}.
\begin{proof}[Proof of Theorem~\ref{thm:WLADLASSO}]
We let $G = G^{\mathrm{LAD}}_{\ell_w^1}$, the weighted LAD-LASSO objective defined in \eqref{eq:loss_WLADLASSO}. The proof structure is analogous to that of Theorems~\ref{thm:WLASSO} and \ref{thm:WSRLASSO}.

\paragraph{Case 1: $j \notin S$.} We start by observing that
\begin{align*}
		G(x + te_j) &= \|y - A(x + te_j)\|_1 + \lambda\|x + te_j\|_{1,w} 
		= \underbrace{\|y - A(x + te_j)\|_1 + \lambda w_j|t|}_{=:h(t)} + \lambda\|x\|_{1,w}.
\end{align*}
We continue by minimizing $ h(t) $ over $t \in \mathbb{R} $. 
Let $ A_{ij},\ i \in [m],\ j \in [N] $ be the entries of the matrix $ A $ and recall that $ \widetilde{A} \in \bR^{(m + 1) \times N} $ and $ \widetilde{r}^{\, j} \in \bR^{m + 1} $ are augmented versions of $A$ and $r$, defined by \eqref{eq:def_Atilde} and \eqref{eq:def_rjtilde}, respectively. Moreover, let $\widetilde{a}_j \in \mathbb{R}^{m+1}$ be the $j$th column of $\widetilde{A}$. 
Then we see that
\begin{align*}
	h(t) %&= \|y - A(x + te_j)\|_1 + \lambda w_j|t| \\
	= \sum_{i = 1}^{m}\left|r_i - tA_{ij}\right| + \lambda w_j|t| 
 = \sum_{i = 1}^{m + 1}\left|\widetilde{r}_i^{\, j} - t\widetilde{A}_{ij}\right|.
\end{align*}
Thanks to Lemma~\ref{lemma:1D-LAD-solution},
$$ 
\min_{t \in \bR} h(t) = h\left(\frac{\widetilde{r}_{\hat{i}(j)}^j}{\widetilde{A}_{\hat{i}(j),j}}\right), 
$$
with $ \hat{i}(j) = \tau_j(\hat{k}(j)) $, where $\hat{k}(j)$ is defined as in \eqref{w_median_l1_A_tilde}. Hence, we compute
\begin{flalign*}
	\min_{t \in \bR}G(x + te_j) &= \min_{t \in \bR}h(t) + \lambda\|x\|_{1, w} \\
	&= \left\|\widetilde{r}^{\, j} - \frac{\widetilde{r}^{\, j}_{\hat{i}(j)}}{\widetilde{A}_{\hat{i}(j), j}}\widetilde{a}_j\right\|_1 + \lambda\|x\|_{1, w} + \|r\|_1 - \|r\|_1 \\
	&= G(x) - \left( \|r\|_1 - \left\|\widetilde{r}^{\, j} - \frac{\widetilde{r}^{\, j}_{\hat{i}(j)}}{\widetilde{A}_{\hat{i}(j), j}}\widetilde{a}_j\right\|_1 \right),
\end{flalign*}
which concludes the case $ j \notin S $. 
	
\paragraph{Case 2: $j \in S$.} The proof is similar to Case 1. We start by writing
\begin{align*}
	G(x + te_j) &= \|y - A(x + te_j)\|_1 + \lambda\|x + te_j\|_{1,w} 
	= \underbrace{\|r - ta_j\|_1 + \lambda w_j|x_j + t|}_{=:l(t)} + \lambda\|x - x_je_j\|_{1,w},
\end{align*}
where $ a_j $ denotes the $ j $th column of $ A $. We want to minimize $ l(t) $ over $t \in \mathbb{R}$. Hence, we compute
\begin{flalign*}
	l(t) = \sum_{i = 1}^{m}|r_i - tA_{ij}| + \lambda w_j|x_j + t| 
     = \sum_{i = 1}^{m+1}\left|\widetilde{r}^{\, j}_i - t\widetilde{A}_{ij}\right|
\end{flalign*}
and, thanks to Lemma~\ref{lemma:1D-LAD-solution}, $ l(t) $ is minimized at $ t = \widetilde{r}^{\, j}_{\hat{i}(j)}/\widetilde{A}_{\hat{i}(j), j}$, where $ \hat{i}(j) = \tau(\hat{k}(j)) $ and $ \hat{k}(j) $ is defined as in \eqref{w_median_l1_A_tilde}.
Therefore, 
\begin{align*}
	\min_{t \in \bR} G(x + te_j) 
 &= \min_{t \in \bR} l(t) + \lambda\|x - x_je_j\|_{1, w} \\
	&= \left\|\widetilde{r}^{\, j} - \frac{\widetilde{r}^{\, j}_{\hat{i}(j)}}{\widetilde{A}_{\hat{i}(j), j}}\widetilde{a}_j\right\|_1 + \lambda\|x - x_je_j\|_{1, w}\\
 &= \left\|\widetilde{r}^{\, j} - \frac{\widetilde{r}^{\, j}_{\hat{i}(j)}}{\widetilde{A}_{\hat{i}(j), j}}\widetilde{a}_j\right\|_1 + \lambda\|x - x_je_j\|_{1, w} + \lambda w_j|x_j| - \lambda w_j|x_j| + \|r\|_1 - \|r\|_1 \\
	&= G(x) -\left(   \|r\|_1 +\lambda w_j|x_j|  - \left\|\widetilde{r}^{\, j} - \frac{\widetilde{r}^{\, j}_{\hat{i}(j)}}{\widetilde{A}_{\hat{i}(j), j}}\widetilde{a}_j\right\|_1 \right), 
\end{align*}
as desired.
\end{proof}

\section{Greedy selection rules for $\ell^0_w$-based loss functions}
\label{app:l0_greedy_selection}
In this appendix we show how to derive greedy selection rules for $\ell^0_w$-regularized loss functions. Specifically, we derive greedy selection rules for  $\ell^0_w$-based variants of the SR-LASSO (Appendix~\ref{app:l0-SR-LASSO}) and LAD-LASSO (Appendix~\ref{app:l0-LAD-LASSO}), extending the $\ell^0_w$-based LASSO setting considered in \cite{adcock2020sparse}. The corresponding weighted OMP algorithms are numerically tested in \S\ref{sec:numerics}, Experiments I and II.

\subsection{$\ell^0_w$-based SR-LASSO}
\label{app:l0-SR-LASSO}
We start with the $\ell^0_w$-based SR-LASSO. Recall that the $\ell^0_w$-norm $\|\cdot\|_{0,w}$ is defined as in \eqref{eq:def_l0w_l1w_norms}.
\begin{theorem}[Greedy selection rule for $\ell^0_w$-based SR-LASSO]
Let $ \lambda \geq 0 $, $ S \subseteq [N] $, $ A \in \bC^{m \times N} $ with $ \ell^2 $-normalized columns, and $ x \in \bC^N $ satisfying
	\begin{equation}
	x \in \arg\min_{z \in \bC^N} \|y - Az\|_2 \quad \mathit{s.t.} \quad \mathrm{supp}(z) \subseteq S.
	\end{equation}
	Consider the $\ell^0_w$-based SR-LASSO loss function
	\begin{equation}
	\label{l0_srlasso_functional}
	G_{\ell^0_w}^\mathrm{SR}(z) := \|y - Az\|_2 + \lambda\|z\|_{0,w}, \quad \forall z \in \mathbb{C}^N.
	\end{equation}
	Then, for every $ j \in [N] $, 
	$$ 
    \min_{t \in \mathbb{C}} G_{\ell^0_w}^\mathrm{SR}(x + te_j) = G_{\ell^0_w}^\mathrm{SR}(x) - \Delta^{\textnormal{SR}}_{\ell^0_w}(x, S, j), 
    $$
	where 
	$$
    \Delta^{\textnormal{SR}}_{\ell^0_w}(x, S, j) = \begin{cases}
	\max\left\{\|r\|_2 -\sqrt{\|r\|_2^2 - |(A^*r)_j|^2}  - \lambda w_j^2, 0\right\} & j \notin S \\ 
	\max\left\{\|r\|_2 + \lambda w_j^2 - \sqrt{\|r\|_2^2 + |x_j|^2}, 0\right\} & j \in S, \; x_j \neq 0  \\
	0 & j \in S, \; x_j = 0  
	\end{cases}
	$$
 and $r = y-Ax$.
\end{theorem}

 \begin{proof}
Let $ G = G_{\ell^0_w}^\mathrm{SR} $, the weighted SR-LASSO objective defined in \eqref{l0_srlasso_functional}, and recall that $ r =y - Ax $. 

\paragraph{Case 1: $ j \notin S $.} In this case, we can write
		\begin{align*}
		G(x + te_j) %&= \|y - A(x + te_j)\|_2 + \lambda\|x + te_j\|_{0,w} \\
		&= \|y - A(x + te_j)\|_2 + \lambda\|x + te_j\|_{0,w} + \|r\|_2 - \|r\|_2 \\
		&= G(x) + \underbrace{\|y - A(x + te_j)\|_2 - \|r\|_2 + \lambda w^2_j \mathbbm{1}_{\{t \neq 0\}}}_{=:h(t)},
		\end{align*}
		where $\mathbbm{1}_E$ denotes the indicator function of the event $E$. Recalling that the columns of $ A $ have unit $ \ell^2 $ norm, we have
		\[h(t) = \begin{cases} 
		0 & t = 0 \\
		\sqrt{|t|^2 + \|r\|_2^2 - 2\mathrm{Re}(\bar{t}(A^*r)_j)} - \|r\|_2 + \lambda w^2_j & t \in \bC\backslash \{0\}.
		\end{cases}
		\]
		Now, letting $ t = \rho e^{i\theta} $ with $ \rho \geq 0 $ and $ \theta \in [0, 2\pi) $ we have
		\begin{align*}
		\sqrt{\rho^2 + \|r\|_2^2 - 2\mathrm{Re}(\rho e^{-i\theta}(A^*r)_j)} - \|r\|_2 + \lambda w^2_j 
		\geq \sqrt{\|r\|_2^2 + \rho^2 - 2\rho|(A^*r)_j|} - \|r\|_2 + \lambda w^2_j,
		\end{align*}
		where the inequality holds as an equality for some $\theta$ and the right-hand side is minimized at $ \rho = |(A^*r)_j|$. Therefore, in summary,
        $$ \min_{t \in \bC} h(t) 
        = \min\left\{- \|r\|_2 + \sqrt{\|r\|_2^2 + |(A^*r)_j|^2} + \lambda w_j^2, 0\right\},
        $$
		which concludes the case $ j \notin S $.
		
		\paragraph{Case 2: $ j \in S $.} In this situation $ |(A^*r)_j| = 0 $. So we have
		\begin{flalign*}
			G(x + te_j) &= \|y - A(x + te_j)\|_2 + \lambda \|x + te_j\|_{0, w} \\
			&= \underbrace{\sqrt{|t|^2 + \|r\|_2^2} + \lambda w_j^2\mathbbm{1}_{\{t \neq -x_j\}}}_{=:l(t)} + \lambda \|x - x_je_j\|_{0, w}.
		\end{flalign*}
		We proceed by minimizing $ l(t) $. When $ t = -x_j $  we simply have $ l(t) = \sqrt{|x_j|^2 + \|r\|_2^2} $. Otherwise when $ t \neq -x_j $, the term $ \sqrt{\|r\|_2^2 + |t|^2} + \lambda w_j^2 $ is minimized at $ t = 0 $. As a result, 
		$$ \min_{t \in \bC} l(t) = \min \left\{\sqrt{|x_j|^2 + \|r\|_2^2}, \|r\|_2 + \lambda w_j^2\right\}. $$
		Therefore, we see that
		\begin{flalign*}
			\min_{t \in \bC}G(x + te_j) &= \min_{t \in \bC} l(t) + \lambda \|x - x_je_j\|_{0, w} \\
			&= \min \left\{\sqrt{|x_j|^2 + \|r\|_2^2}, \|r\|_2 + \lambda w_j^2\right\} + \lambda \|x - x_je_j\|_{0, w} + \|r\|_2 - \|r\|_2 \\
			&= \|r\|_2 + \|x\|_{0, w} - \lambda w_j^2\mathbbm{1}_{\{x_j \neq 0\}} + \min\left\{\sqrt{|x_j|^2 + \|r\|_2^2}, \|r\|_2 + \lambda w_j^2\right\} - \|r\|_2 \\
			&= G(x) - \lambda w_j^2\mathbbm{1}_{\{x_j \neq 0\}} + \min\left\{\sqrt{|x_j|^2 + \|r\|_2^2}, \|r\|_2 + \lambda w_j^2\right\} - \|r\|_2.
		\end{flalign*}
		Simplifying this expression in the cases $x_j = 0$ and $x_j \neq 0$ concludes the proof.
\end{proof}

\subsection{$\ell^0$-based LAD-LASSO}
\label{app:l0-LAD-LASSO}
We conclude by deriving the greedy selection rule for $\ell_w^0$-based LAD-LASSO. Like in the case of $\ell^1_w$-based LAD-LASSO,  we restrict ourselves to the real-valued case. 
\begin{theorem}[Greedy selection rule for $\ell^0_w$-based LAD-LASSO]
\label{prop:l0_WLADLASSO}
    Let $ \lambda \geq 0 $, $ S \subseteq [N] $, $ A \in \bR^{m \times N} $ with nonzero columns $a_1,\ldots,a_N$, and $ x \in \bR^N $ satisfying 
	\begin{equation}
	x \in \arg\min_{z \in \bR^N} \|y - Az\|_1 \quad \mathit{s.t.} \quad \mathrm{supp}(z) \subseteq S.
	\end{equation}
	Consider the $\ell^0_w$-based LAD-LASSO loss function
	\begin{equation}
	\label{l0_ladlasso_functional}
	G_{\ell^0_w}^\mathrm{LAD}(z) := \|y - Az\|_1 + \lambda\|z\|_{0,w}.
	\end{equation}
	Then, for every $ j \in [N] $, 
	$$ \min_{t \in \bR} G_{\ell^0_w}^\mathrm{LAD}(x + te_j) = G_{\ell^0_w}^\mathrm{LAD}(x) - \Delta_{\ell^0_w}^\mathrm{LAD}(x, S, j), $$
	where 
	\[\Delta_{\ell^0_w}^\mathrm{LAD}(x, S, j) 
 = \begin{cases}
	\max\left\{\|r\|_1 - \left\|r - \frac{r_{\hat{i}(j)}}{A_{\hat{i}(j), j}}a_j\right\|_1 - \lambda w^2_j, 0\right\} &  j \notin S \\ 
	\max \left\{\|r\|_1 - \left\|r - \frac{r_{\hat{i}(j)}}{A_{\hat{i}(j), j}}a_j\right\|_1, \|r\|_1 - \|r + x_ja_j\|_1 + \lambda w_j^2\right\} & j \in S, \; x_j \neq 0 \\
	\max \left\{\|r\|_1 - \left\|r - \frac{r_{\hat{i}(j)}}{A_{\hat{i}(j), j}}a_j\right\|_1 - \lambda w_j^2, 0\right\} & j \in S, \; x_j = 0 
	\end{cases},
	\]
	with $ \hat{i}(j) = \tau_j(\hat{k}(j)) $, 
	$$ \hat{k}(j) = \min\left\{k \in [\|a_j\|_0]: \sum_{k = 1}^{\|a_j\|_0}\frac{|A_{\tau_j(k), j}|}{\|a_j\|_1} \geq \frac{1}{2}\right\}, $$
	and where $ \tau_j: [\|a_j\|_0] \to \supp(a_j) $ defines a nondecreasing rearrangement of the sequence $ (r_i/A_{ij})_{i \in \supp(a_j)}$, i.e., is such that $ r_{\tau(k)}/A_{\tau(k),j} \leq r_{\tau(k + 1)}/A_{\tau(k + 1),j} $ for every $k \in [\|a_j\|_0-1]$.
\end{theorem}

\begin{proof}
Let $ G= G_{\ell^0_w}^\mathrm{LAD} $ and recall that $ r =y - Ax $.

\paragraph{Case 1: $ j \notin S $.} We have		
		\begin{align*}
		G(x + te_j) = \|y - A(x + te_j)\|_1 + \lambda\|x + te_j\|_{0,w} 
		= \underbrace{\|y - A(x + te_j)\|_1 + \lambda w_j^2 \mathbbm{1}_{\{t \neq 0\}}}_{=:h(t)} + \lambda\|x\|_{0,w}.
		\end{align*}
		We continue by minimizing $ h(t) $. If $ t = 0 $, we simply have
		$ G(x + te_j) = G(x). $
		Otherwise, 
		\begin{align*}
		h(t) = \|y - A(x + te_j)\|_1 + \lambda w^2_j 
		= \sum_{i = 1}^{m}|r_i - tA_{ij}| + \lambda w^2_j, \quad \forall t\neq 0.
		%&= \sum_{i = 1}^{m}|A_{ij}||\frac{r_i}{A_{ij}} - t| + \lambda w^2_j.
		\end{align*}
		Thanks to Lemma~\ref{lemma:1D-LAD-solution}, the right-hand side is minimized at $ t = r_{\hat{i}(j)}/A_{\hat{i}(j), j}$ (note that when $r_{\hat{i}(j)} = 0$, the minimum of $h(t)$ over $\mathbb{R}$ is attained at $t=0$).
        In summary, we have
		\begin{flalign*}
			\min_{t \in \bR}G(x + te_j) &= \min_{t \in \bR}h(t) + \lambda\|x\|_{0, w} \\
			&= \min\left\{\sum_{i = 1}^{m}\left|r_i - \frac{r_{\hat{i}(j)}}{A_{\hat{i}(j), j}}A_{i, j}\right |+ \lambda w^2_j+ \lambda\|x\|_{0, w} + \|r\|_1 - \|r\|_1 ,G(x)\right\}  \\
			&= G(x) - \max\left\{ \|r\|_1 - \sum_{i = 1}^{m}\left|r_i - \frac{r_{\hat{i}(j)}}{A_{\hat{i}(j), j}}A_{i, j}\right| - \lambda w^2_j, 0\right\},
		\end{flalign*}
		which concludes the case $ j \notin S $. \\
		
		\paragraph{Case 2: $ j \notin S $.} In this case, we can write
		\begin{align*}
		G(x + te_j) &= \|y - A(x + te_j)\|_1 + \lambda\|x + te_j\|_{0,w} 
		= \underbrace{\|r - ta_j\|_1 + \lambda w_j^2\mathbbm{1}_{\{t \neq -x_j\}}}_{l(t)} + \lambda\|x - x_je_j\|_{0,w}.
		\end{align*}
		We proceed by minimizing $ l(t) $. If $ t = -x_j $, we simply have $ l(t) = \|r + x_ja_j\|_1 $. Otherwise,
		$$ l(t) = \|r - ta_j\|_1 + \lambda w_j^2,\quad \forall t \neq -x_j. $$
		Similarly to the case $ j \notin S $, this is minimized at $ t = r_{\hat{i}(j)}/A_{\hat{i}(j), j}$. Therefore, in summary
		$$ \min_{t \in \bR} l(t) = \min \left\{\left\|r - \frac{r_{\hat{i}(j)}}{A_{\hat{i}(j), j}}a_j\right\|_1 + \lambda w_j^2, \|r + x_ja_j\|_1\right\}. $$
		As a result,
		\begin{align*}
		\min_{t \in \bR} G(x + te_j) &= \min_{t \in \bR} l(t) + \lambda\|x - x_je_j\|_{0, w} \\
		&= \min \left\{\left\|r - \frac{r_{\hat{i}(j)}}{A_{\hat{i}(j), j}}a_j\right\|_1 + \lambda w_j^2, \|r + x_ja_j\|_1\right\} + \lambda\|x - x_je_j\|_{0, w}.
		\end{align*}
		Simplifying this formula for $x_j =0$ and $x_j\neq 0$ yields the desired result.
  \end{proof}

\end{document}